\documentclass[11pt]{article}

\usepackage[utf8]{inputenc}
\usepackage{hyperref}
\usepackage{amsmath}
\usepackage{amsthm}
\usepackage{amssymb}
\usepackage{authblk}
\usepackage{cite}
\usepackage{lipsum}
\usepackage{geometry}
\usepackage{enumerate}
\usepackage{array}
\usepackage{multicol}
\usepackage{multirow}
\usepackage{xcolor}

\newtheorem{theorem}{Theorem}[section]
\numberwithin{theorem}{section}
\newtheorem{lemma}[theorem]{Lemma}
\newtheorem{definition}[theorem]{Definition}
\newtheorem{corollary}[theorem]{Corollary}
\newtheorem{prop}[theorem]{Proposition}

\newtheorem{remark}{Remark}[section]

\DeclareMathOperator{\tr}{Tr}
\DeclareMathOperator{\crk}{Crk}

\newcommand{\F}{\mathbb F}
\newcommand{\Fbn}{\mathbb{F}_{2^n}}
\newcommand{\Fbnn}{\mathbb{F}_{2^n} \longrightarrow \mathbb{F}_{2^n}}
\newcommand{\ddt}{\textup{\texttt{DDT}}}
\newcommand{\bct}{\textup{\texttt{BCT}}}

\usepackage{mathrsfs}
\usepackage[cal=boondox]{mathalpha}

\usepackage{amsmath,amssymb,amsthm}
\usepackage{graphicx}
\usepackage{array}

\usepackage{multicol}
\usepackage[]{color}
\usepackage{tikz}

\newcommand{\bu}{\textup{\texttt{BU}}}
\newcommand{\du}{\textup{\texttt{DU}}}

\newcommand*\circled[1]{\tikz[baseline=(char.base)]{
            \node[shape=circle,draw,inner sep=2pt] (char) {#1};}}

\newcommand{\ca}{\textup{\texttt{Class}} \,\,\mathcal{A}}
\newcommand{\cb}{\textup{\texttt{Class}} \,\,\mathcal{B}}
\newcommand{\cc}{\textup{\texttt{Class}} \,\,\mathcal{C}}

\newcommand{\pole}{\mathcal{P}}
\newcommand{\hh}{\mathcal{h}}
\newcommand{\res}{\text{Res}}

\begin{document}

%\includepdf[pages=1-last]{Boomerang_20200910_submission_arxiv_2.pdf}

    \title{On the Boomerang Uniformity of \\
                 Permutations of Low Carlitz Rank}
    \author{Jaeseong Jeong$^1$, Namhun Koo$^2$, and Soonhak Kwon$^1$\\
        \small{\texttt{ Email: wotjd012321@naver.com, nhkoo@ewha.ac.kr, shkwon@skku.edu}}\\
        \small{$^1$Applied Algebra and Optimization Research Center, Sungkyunkwan University, Suwon, Korea}\\
        \small{$^2$Institute of Mathematical Sciences, Ewha Womans University, Seoul, Korea}
    }

   \date{}

    \maketitle
    \begin{abstract}
        Finding permutation polynomials with low differential and boomerang uniformity
        is an important topic in S-box designs of many block ciphers. For
        example, AES chooses the inverse function as its S-box, which is
         differentially 4-uniform and boomerang 6-uniform. Also there has been
         considerable research on many non-quadratic permutations which are obtained by modifying
        certain set of points from the inverse function.
        In this paper, we give a novel approach
        that shows that plenty of existing modifications of the inverse function are in fact affine
        equivalent to  permutations of low Carlitz rank and those modifications
        cannot be APN (almost perfect nonlinear) unless the Carlitz rank is very large.    Using nice properties
        of the permutations of Carlitz form, we present  the complete
         list of permutations of Carlitz rank 3 having the boomerang uniformity six, and also give the
        complete classification
        of the differential uniformity of  permutations of Carlitz rank 3.
        We also provide, up to affine equivalence, all the involutory permutations of Carlitz rank 3
         having the boomerang uniformity six.
    \end{abstract}

\section{Introduction}

Substitution boxes (S-Boxes) are the only nonlinear part in block ciphers, thus they play important roles to resist
several attacks on block ciphers. Mathematically, S-boxes are vectorial (multi-output) Boolean functions and rigorous
analysis of the Boolean functions is possible.   The boomerang attack is one of the well-known cryptanalysis proposed
by Wagner\cite{Wag99}. In Eurocrypt 2018, Cid et al.\cite{CHP18+} proposed a new cryptanalysis tool called the
Boomerang Connectivity Table(BCT), which measures the resistance against various styles of boomerang attacks, where it
is heuristically  shown that finding permutations with maximum BCT value $\leq 4$ is very difficult.
   Shortly afterwards, Dunkelman \cite{Dun18} improved the complexity for BCT construction.

By formalizing the concept of BCT,  Boura and Canteaut \cite{BC18} introduced a new invariant, the boomerang
uniformity, which is the maximum value of nontrivial parts of the BCT. They did in-depth study on  the boomerang
uniformity of  many differentially 4-uniform permutations such as the inverse function and some quadratic permutations
including Gold function.
 Later, Li et
al. \cite{LQSL19} gave an alternative definition of the boomerang uniformity, which reveals the connection between
 the boomerang and the differential uniformity more clear. They also showed that some family of quadratic binomial permutations
have the boomerang uniformity 4. In \cite{TBP20}, Tian et al. studied the boomerang uniformity of some popular S-boxes
and also derived the result that the boomerang uniformity does not change under EA equivalence for Gold function and
the inverse function.

 Recently,
several researchers, including the above mentioned \cite{BC18,LQSL19,TBP20}, studied  quadratic permutations of
boomerang uniformity 4 over $\Fbn$ with $n\equiv 2 \pmod{4}$.  In \cite{MTX20}, Mesnager et al. studied boomerang
properties of quadratic permutations when $n$ is even which generalized some previous results on quadratic
permutations.
  In
particular, they proved that the boomerang uniformity of the differentially 4-uniform quadratic permutation presented
in \cite{BTT14} equals 4. Very recently, Tu et al.\cite{TLZ20} proved that a class of quadratic quadrinomials has the
boomerang uniformity 4. Also, Li et al. \cite{LHXZ20} proposed permutations from generalized butterfly structure with
the boomerang uniformity 4. All these quadratic permutations have the optimal boomerang uniformity $4$ for $n\equiv 2
\pmod{4}$, but they are rather impractical to be applied in construction of S-boxes because they are quadratic. It is
known that permutations with low algebraic degree is vulnerable at high order differential attack, and is also known
that quadratic permutations have the maximal differential-linear uniformity \cite{CKL+19}.

Therefore it is necessary to consider permutations with low boomerang uniformity and high algebraic degree. The
boomerang uniformity of a modification of the inverse function swapping the image of 0 and 1 was investigated by Li et
al. \cite{LQSL19}. Recently, Calderini and Villa \cite{CV20} computed the boomerang uniformity of three  classes of
differentially 4-uniform permutations which are  modifications of the inverse function. There are many known
results\cite{CV20, LWY13e, TCT15, ZHS14, PTW16, LQSL19, QTTL13, QTLG16, YWL13,PT17} on the modifications of the inverse
function which have good cryptographic properties such as high algebraic degree, high nonlinearity and low differential
uniformity. However, it was not studied enough on the boomerang uniformity of these classes of permutations, and hence
it is necessary to be further studied.

The Carlitz rank, first introduced in \cite{ACMT09}, is the concept based on known result that all permutations can be
expressed by a composition of the inverse function and linear polynomials (see\cite{Nie15}). Differential uniformity of
permutations with Carlitz rank 1 or 2 in some finite fields of odd characteristic can be found in \cite{CMT14}, but for
even characteristic there was no further observation connecting Carlitz rank with cryptographic properties such as
differential
 and boomerang uniformity.
 It is easy to see that permutations of Carlitz rank 1 are affine equivalent to the inverse
 function. Also, permutations with Carlitz rank 2 are
  affine equivalent to a modification of the inverse function swapping the image of 0 and 1, whose
  cryptographic properties  are already
  studied in \cite{LWY13e} for differential uniformity and in  \cite{LQSL19} for boomerang uniformity.
  For permutations of Carlitz rank 3, the
  boomerang uniformity of some class of permutations (with coefficients in $\F_4$) has been studied
  in \cite{CV20}. Also, differentially
  4-uniform permutations of Carlitz rank 3  with even $n$ were studied in \cite{LWY13e}, but
  not completely classified for arbitrary permutation of Carlitz rank 3 and for arbitrary dimension $n$.
  Therefore, even if it is not explicitly mentioned in the above literature, the authors of the above articles in fact studied
  cryptographic properties of permutations of low Carlitz rank, where the case rank 3 is not completely settled yet.

 In this paper, we give a complete characterization of the condition that the boomerang
   uniformity equals six for permutations of Carlitz rank 3, which is minimal  in this class, and give the
        complete classification
        of the differential uniformity of  permutations of Carlitz rank 3.
        It should be mentioned that, among all existing modifications of the inverse function over the field $\Fbn$ with even $n$,
        boomerang uniformity six is the lowest value known at this moment.
         We also provide, up to affine equivalence, all the involutory permutations of Carlitz rank 3
         having the boomerang uniformity six.
       Moreover, using nice properties
        of the permutations of Carlitz form, we show that plenty of existing modifications of the inverse function are in fact affine
        equivalent to  permutations of low Carlitz rank and those modifications
        cannot be APN (almost perfect nonlinear) unless the Carlitz rank is very large.

The rest of this paper is organized as follows. In Section 2, we give some basic notions of vectorial Boolean functions
and Carlitz rank, which will be used in this paper. In Section 3, we discuss affine equivalence of permutations of
Carlitz form and explain the connections between previous works on boomerang uniformity and Carlitz form. In Section 4,
we introduce the notion of convergents and poles of Carlitz form, and prove that APN permutation on even dimension must
have a very large Carlitz rank. In Section 5, we explain fundamentals of permutations of Carlitz rank 3 which will be
discussed in detail in subsequent sections. In Section 6, We give a complete characterization of the differential
uniformity of the permutations with Carlitz rank 3.
 In Section 7, we give a complete characterization of the permutations of Carlitz rank 3 having the boomerang uniformity six.
 In Section 8, as an application, we explain that involutions are very easy to find in Carlitz form, and provide all the involutions
 of Carlitz rank 3 having the boomerang uniformity six. We also show implementation results using SageMath.
 Finally, in Section 9, we
give the concluding remarks.

\section{Preliminaries}

Let $q=2^n$ with positive integer $n$ and let $\mathbb F_q=\Fbn$ be a finite field with $q$ elements. A function $F :
\Fbnn$ is called an $(n,n)$-function or a   vectorial Boolean function on $\Fbn$.
 For a given $(n,n)$-function $F : \Fbnn$, there is a unique polynomial representation for $F$ such that
  $F(x)=\sum_{i=0}^{2^n-1} c_i x^i$ with $c_i\in \Fbn$. In case when $F$ is a permutation, we call $F$ a
  permutation polynomial. For a given $F : \Fbnn$,
  the difference
    distribution table, $\ddt$, consists of elements at each position $(a,b)\in\Fbn\times \Fbn$ given by
$$\ddt_F(a,b) = \# \{x \in \Fbn :F(x) + F(x+a) = b \}$$
The differential uniformity $\Delta_F$ of $F$ is defined as
 $$\Delta_F = \displaystyle \max_{a
    \in \Fbn \setminus \{0\},\ b \in \Fbn} \ddt(a,b), $$
and such  $F$ is said to be differentially $\Delta_F$-uniform. Note that the least possible value of $\Delta_F$ is two,
in which case the function is called almost perfect nonlinear (APN).

For a given permutation $F : \Fbnn$, the boomerang connectivity table, $\bct$, consists of elements at each position
$(a,b)\in \Fbn\times \Fbn$ given by
$$\bct_F(a,b) = \# \{x \in \Fbn : F^{-1}(F(x)+b) + F^{-1}(F(x+a)+b) = a\}$$
 The boomerang uniformity $\delta_F$ of $F$ is defined as
 $$\delta_F =\displaystyle \max_{a, b
    \in \Fbn \setminus \{0\}} \bct(a,b), $$
and such  $F$ is said to be boomerang $\delta_F$-uniform. For any permutation $F$, it is well known that
 $\Delta_F\leq \delta_F$, and $\delta_F=2$ if and only if  $F$ is APN. (See \cite{CHP18+})
  In particular, $\delta_F=4$ implies $\Delta_F=4$. However, many permutation polynomials in \cite{BC18, LQSL19, CV20} turned out to be
  $\Delta_F=4$ but $\delta_F\geq 6$.

The absolute field trace $\tr : \Fbn \longrightarrow \F_{2}$ is defined by   $\tr(x) = \sum_{i=0}^{n-1}x^{2^{i}}.$ Two
functions $F: \Fbnn$ and $ F' : \Fbnn$ are called affine equivalent  if there exist affine permutations $A',A'' :
\Fbnn$ satisfying  $F' = A' \circ F \circ A'' $. It is well-known that two affine equivalent functions are of same
differential uniformity and boomerang uniformity. (for details,  see \cite{Bud15, Can, Car10, CS17})

\medskip
We define a finite field analogue of continued fraction of real numbers, which is useful to study permutation
 polynomials of Carlitz form soon to be discussed.
\begin{definition}
For given $a_1, a_2, \cdots a_s \in \Fbn$ not necessarily distinct, we define
$$
[a_1,a_2,\ldots,a_s]=  ((\cdots (a_s^{q-2}+a_{s-1})^{q-2}+\cdots )^{q-2}+a_2)^{q-2}+a_1 \quad \mbox{with}\,\, q=2^n
$$
\end{definition}
\noindent Note that the usual continued fraction of positive real numbers replaces $a^{q-2}$ with ${a}^{-1}$ and is
always
 well defined since the convergents are well defined, i.e., the denominators of the convergents are never zero for the
 case of real numbers. In the case of finite field, although the above definition is well defined since there is no division,
 one should be very careful when extending the above notation
 to the convergents because the denominator of convergents may be zero in finite fields.

\medskip
 It is known \cite{ACMT09, Nie15} that, for any permutation $F : \Fbn \longrightarrow
\Fbn$, there is $m \geq 0$ and $a_i \in \Fbn \,\, (0\leq i \leq m)$ such that
\begin{equation}\label{car}
F(x) = (\cdots((a_0x + a_1)^{q-2}+a_2)^{q-2} \cdots +a_m)^{q-2} + a_{m+1}
\end{equation}
%\begin{equation}\label{car}
%F(x) = (\cdots((a_0x + a_1)^{-1}+a_2)^{-1} \cdots +a_m)^{-1} + a_{m+1}
%\end{equation}
where $a_0, a_2, \cdots , a_m \neq 0$.
 The above expression means that any permutation on $\Fbn$ is generated by inverse function $x^{q-2}$ and linear
function $ax+b\,\, (a\neq 0)$. Using our definition of continued fraction, equation \eqref{car} can be expressed as
\begin{equation}\label{car_CF} F(x) = [a_{m+1},a_m, \ldots , a_2, a_1 + a_0x]
\end{equation}

%By using continued fraction  notation $$[a_0,a_1, \ldots , a_m] =a_0 + (a_1 + (a_2 +  \cdots (a_{m-1}+   a_m^{-1} )
%\cdots )^{-1})^{-1},$$ Eq.\eqref{car} can be expressed as
%\begin{equation}\label{car_CF} F(x) = [a_{m+1},a_m, \ldots , a_2, a_1 + a_0x].
%\end{equation}

For a given permutation $F$, the above expression is not unique in general. However there is the least $m\geq 0$ among
all possible expressions of $F$. The \textbf{Carlitz rank} of  $F$ (denoted $\crk(F)$ ) is the least nonnegative
integer $m$ satisfying the above expression. Suppose that a permutation $F: \Fbn \rightarrow \Fbn$ has Carlitz rank
$m$. Then one may write $F$ as the form of the equation \eqref{car_CF}. For every $0 \leq k \leq m$, we define
\begin{equation}\label{fkxdefinition}
\begin{split}
F_k(x) &= [a_{k+1},a_k, \ldots , a_2 , a_1+a_0x] \\
&=(\cdots((a_0x + a_1)^{q-2}+a_2)^{q-2} \cdots +a_k)^{q-2} + a_{k+1}
\end{split}
\end{equation}
Then one has $F_k(x)=F_{k-1}(x)^{q-2}+a_{k+1}$ where $F_0(x)=a_0x+a_1, F_1(x)=(a_0x+a_1)^{q-2}+a_2,
F_2(x)=((a_0x+a_1)^{q-2}+a_2)^{q-2}+a_3, \cdots, F_m(x)=F(x)$.
 Also we  define $R_k(x)$ for $0\leq k \leq m$ as
\begin{equation}\label{rkxdefinition}
R_k(x) =(\cdots((a_0x + a_1)^{-1}+a_2)^{-1} \cdots +a_k)^{-1} + a_{k+1},
\end{equation}
where the domain of definition of $R_k(x)$ is the set of all $x\in \Fbn$ satisfying $F_j(x)\neq 0$ for all $0\leq j<k$.
 Since $F_j$ itself is a permutation, $F_j(x)=0$ has a unique root and it is trivial to
check the unique root is $x=\tfrac{1}{a_0}[a_1,a_2,\cdots, a_{j+1}]$.
 Also note that $R_k(x)=R_{k-1}(x)^{-1}+a_{k+1}$. It is known \cite{CMT08 } that $R_k$ has the following form of
fraction
\begin{equation}\label{alp,bet notation}
R_k(x)=\frac{\alpha_{k+1}x+\beta_{k+1}}{\alpha_kx+\beta_k},
\end{equation}
where
\begin{equation}\label{recurrence}
\alpha_{k+1}=a_{k+1}\alpha_k+\alpha_{k-1},\,\,\,  \beta_{k+1}=a_{k+1}\beta_k+\beta_{k-1} \quad  (1\leq k\leq m)
\end{equation}
with the initial conditions $\alpha_0=0, \alpha_1=a_0$ and $\beta_0=1, \beta_1=a_1$.
 The above recurrence relation can easily be derived using induction since
 \begin{align*}
 R_k(x)=R_{k-1}(x)^{-1}+a_{k+1}=\frac{\alpha_{k-1}x+\beta_{k-1}}{\alpha_k x+\beta_k}+a_{k+1}
                             =\frac{(a_{k+1}\alpha_k+\alpha_{k-1})x+(a_{k+1}\beta_k+\beta_{k-1})}{\alpha_k x+\beta_k}
 \end{align*}
From the constructions of $F_k$ and $R_k$, one concludes
\begin{align}
F(x)=R_m(x)   \quad \mbox{for all } x \notin \{\tfrac{1}{a_0}[a_1,a_2,\ldots, a_{i}]  \,\,|\,\,  i=1,2,\ldots, m\}
\label{frequality}
\end{align}
One can also easily verify
\begin{align}
F\left(\tfrac{1}{a_0}[a_1,a_2,\ldots, a_{i}]\right)=[a_{m+1}, a_m,\ldots, a_{i+1}]  \quad \mbox{for all }
               i=1,2,\ldots, m \label{frexception}
\end{align}

\section{Carlitz Form and Connection with Previous Works on Boomerang and Differential Uniformity}

\subsection{Carlitz Form and Affine Equivalence}

The following lemma explains some properties of continued fractions over finite fields.

\begin{lemma} \label{cont. frac lemma}
    Let $b_0,b_1, \ldots , b_m \in \Fbn$. Then the followings are satisfied.

 \begin{enumerate}[$(a)$]
        \item   For $ 0 \leq i < j \leq m$, one has $[b_0,b_1, \ldots , b_i] = [b_0,b_1, \ldots , b_j]$
              if and only if  $[b_{i+1},b_{i+2},\ldots, b_j] = 0$.
        %\item $[a_m,a_{m-1},\ldots,a_0,x] = 0$ if and only if $x = [0,a_0,a_1, \ldots ,a_m]$.
        \item  One has $[b_0, b_1, \ldots , b_j] = 0$ if and only if $[b_{j},b_{j-1}, \ldots ,b_0]=0$.
               More generally, for $0\leq i< j\leq m$, one has $[b_0, b_1, \ldots , b_j] = 0$ if and
only if $[b_{i+1},b_{i+2}, \ldots , b_j]^{q-2} =  [b_i,b_{i-1}, \ldots , b_0] $.
    \end{enumerate}
\end{lemma}
\begin{proof}
   From $ [b_0,b_1, \ldots , b_i] = b_0 + [b_1, \ldots , b_i]^{q-2}$
 and $ [b_0,b_1, \ldots , b_j] =b_0 + [b_1, \ldots , b_j]^{q-2}$,
one has $ [b_0,b_1, \ldots , b_i] =[b_0,b_1, \ldots , b_j]$ if and only if $[b_1, \ldots , b_i] = [b_1, \ldots ,b_j]$.
Inductively, we get the claim $(a)$. The claim $(b)$ follows in a similar way.
    \begin{align*}
      [b_0,b_1, \ldots , b_j] = 0  &\Leftrightarrow [b_1, \ldots , b_j]^{q-2}=b_0
    \Leftrightarrow [b_2, \ldots , b_j]^{q-2}= [b_1, b_0] \\
    &\Leftrightarrow  \cdots \Leftrightarrow b_j^{q-2}= [b_{j-1},b_{j-2}, \ldots ,b_0]
                          \Leftrightarrow 0= [b_j, b_{j-1}, \ldots ,b_0]
    \end{align*}
\end{proof}

 Many of the cryptographic parameters such as differential and boomerang uniformities are invariant under affine
 equivalence, and we will choose a standard expression of each permutation with Carlitz rank $m$ up to affine equivalence.
\begin{prop}\label{EAequivalent lemma}
    Let $F(x) =  [a_{m+1},a_m, \ldots , a_2 , a_1+a_0x] $ be a permutation on $\Fbn$ with $a_0, a_2, \ldots , a_m \neq 0$.
  \begin{enumerate}[$(1)$]
    \item If $m = 1$, then $F$ is affine equivalent to the inverse function $x^{q-2}$.
    \item If $m \geq2$, then $F$ is affine equivalent to $F'$ where
                $$F'(x) =[0,b_m, \ldots , b_3, 1 , x]  \quad \mbox{with }\, b_i =a_2^{(-1)^{i+1}}a_{i} \,\,\, (2\leq i\leq m)$$
    Therefore, to study cryptographic properties of a permutation
    $F(x) =  [a_{m+1}, a_m,\ldots, a_2, a_1+a_0x] $ with $m\geq 2$,
    one may assume $a_{m+1}=0, a_2=1, a_1=0$ and $a_0=1$ so that one only needs to consider the
following form
$$F(x) =  [0,a_m, \ldots, a_3, 1, x]$$
of permutations.
  \end{enumerate}
\end{prop}

\begin{proof}
    If $ m = 1$, then $F(x) = (a_0x + a_1)^{q-2}+ a_2$ and we have
    $ F \left (\dfrac{x+a_1}{a_0}\right ) +a_2 = x^{q-2}$. Now suppose $ m \geq 2$. Then it holds that
\begin{equation*}
\begin{split}
   a_2^{(-1)^{m}}F  &\left ( \frac{1}{a_0a_2}x + \frac{a_1}{a_0}\right ) + a_2^{(-1)^{m}}a_{m+1} \\
             &=
(\cdots ( ( ( x^{q-2}+1 )^{q-2} + a_2a_3 )^{q-2} + a_2^{-1}a_4 )^{q-2} \cdots +a_2^{(-1)^{m+1}}a_{m})^{q-2} \\
&=[0,a_2^{(-1)^{m+1}}a_{m}, \ldots , a_2^{-1}a_4,a_2a_3, 1 , x]
\end{split}
\end{equation*}
\end{proof}

Many of the previous works
  \cite{CV20, LWY13e, TCT15, ZHS14, PTW16, LQSL19, QTTL13, QTLG16, YWL13,PT17} on cryptographic parameters discuss
  on the functions which are modifications of the inverse function at some set of points. Although not explicitly mentioned,
  some of the above articles in fact discuss permutations of low Carlitz rank.
  %(Of course they DO consider
 %other types of permutations also, which cannot be tackled by our approach of Carlitz form.)
  Since one of the natural approaches is the study of inverse functions
modified at some (small) set of points, and these functions can be expressed as permutations of (low) Carlitz rank
under simple types of affine equivalences (which will be stated here) using the equations \eqref{alp,bet notation},
\eqref{frequality} and \eqref{frexception}, plenty of the previous works can be rephrased in the language of Carlitz
form. To be more specific, we present the following lemma which gives a connection between Carlitz form and
 previous works on differential and boomerang uniformity. Please note that, although our main interest
 is the reduced form $F(x) =  [0,a_m, \ldots, a_3, 1, x]$, we state the following lemma for arbitrary
 $F$ for the sake of generality.

    \begin{lemma}\label{inverse lemma}
        Let $F(x) =  [a_{m+1},a_m, \ldots , a_2 , a_1+a_0x] $ be a permutation on $\Fbn$ with $a_0, a_2, \ldots , a_m \neq 0$.
         Then the followings are satisfied.
         \begin{enumerate}[$(a)$]
        \item If $\alpha_{m} \neq 0$, then there is a subset $P \subset \Fbn $
        with $\# P \leq m$  and affine permutations $\ell_1, \ell_2 : \Fbn
        \rightarrow \Fbn$ satisfying $(\ell_2 \circ F \circ \ell_1)
        (x)=\frac{1}{x}$ for all $x\not \in P$.
        \item If $\alpha_{m} = 0$, then there is a subset $P \subset \Fbn $
        with $\# P \leq m$  and affine permutations $\ell : \Fbn
        \rightarrow \Fbn$ satisfying $(\ell \circ F)
        (x)=x$ for all $x\not \in P$.
         \end{enumerate}
    \end{lemma}
    \begin{proof} Recall that, from the equation \eqref{frequality},
    $$F(x) =\frac{\alpha_{m+1}x+\beta_{m+1}}{\alpha_mx+\beta_m}
        \text{ for all } x \notin \{\tfrac{1}{a_0}[a_1,a_2,\cdots, a_{i}]  \,\,|\,\,  i=1,2,\cdots, m\}$$
           Let us first consider the case $\alpha_{m} \neq 0$.
         It is trivial to check
           $\alpha_i\beta_{i+1}+\alpha_{i+1}\beta_i=a_0$ for all $0\leq i\leq m$
           (in particular, $\alpha_m\beta_{m+1}+\alpha_{m+1}\beta_m=a_0)$. Using this property, one has
         $$\frac{\alpha_{m+1}x+\beta_{m+1}}{\alpha_mx+\beta_m}
         =\frac{a_0}{\alpha_m (\alpha_m x+\beta_m)}+\frac{\alpha_{m+1}}{\alpha_m} $$
        Thus defining
        $\ell_1(x)=\dfrac{a_0 x+\beta_m}{\alpha_m} \text{ and }
        \ell_2(x)=\alpha_m x +\alpha_{m+1},
        $
        we have
        $$\ell_2 \circ F \circ \ell_1
        (x)=\frac{1}{x}
        \text{ for all } x \not \in
         \left\{   \frac1{a_0}\left(\alpha_m\cdot \frac1{a_0}[a_1,a_2,\ldots, a_{i}] +\beta_m\right)
                  \,\,|\,\, i = 1 , \ldots , m\right\}\overset{{\rm def}}{=}P $$
        Now consider the case $\alpha_{m} = 0$. In this case, both $\alpha_{m+1}$ and $\beta_m$ are nonzero  because
        $\alpha_{m+1}\beta_m=\alpha_m\beta_{m+1}+\alpha_{m+1}\beta_m=a_0\neq 0$.
          Therefore defining $\ell(x)=\dfrac{\beta_mx + \beta_{m+1}}{\alpha_{m+1}}$,
             \begin{align*}
         \ell \circ F (x)= \ell \left(\frac{\alpha_{m+1}x+\beta_{m+1}}{\beta_m}\right) = x
         \quad \mbox{for all} \,\,  x \not \in
           \left\{  \tfrac{1}{a_0}[a_1,a_2,\ldots, a_{i}] \,\,|\,\, i = 1 , \ldots , m\right\}\overset{{\rm def}}{=}P
           \end{align*}
    \end{proof}

\begin{remark}
 Lemma \ref{inverse lemma} implies that a permutation
     $F$ with $\crk(F) = m$ is affine equivalent to the inverse function or identity function
       except on a subset $P$ of $\Fbn$ with $\# P\leq m$.
\end{remark}

\subsection{Carlitz Form in Previous Works}

Now we explain the connections between  Lemma \ref{inverse lemma} and previous works.

\medskip
\noindent {\bf When $\crk(F) = 1$ :} from the Proposition \ref{EAequivalent lemma}-$(1)$, $F$ is affine equivalent to
inverse function, which has been extensively studied.

\medskip
\noindent {\bf When  $\crk(F) = 2$ :} from the Proposition \ref{EAequivalent lemma}-$(2)$, we may assume
 $F(x) = [0,1,x]=(x^{q-2}+1)^{q-2}$. Using Lemma \ref{inverse lemma}-$(a)$, it is easy to check $\ell_1(x)=x+1=\ell_2(x)$
such that
   $$\ell_2\circ F \circ \ell_1 (x)= F(x+1)+1=((x+1)^{q-2}+1)^{q-2}+1= f_2(x)$$ where
\begin{equation*}
f_2(x) = \begin{cases}
0 &\text{ if } x= 1 \\
1 &\text{ if } x= 0 \\
\frac{1}{x} &\text{ otherwise }
\end{cases}
\end{equation*}
 The differential uniformity of $f_2$ is  well known in \cite{LWY13e, ZHS14}. Also the boomerang uniformity of $f_2$
 is recently classified in \cite{LQSL19} completely. For example, it is shown in \cite{LQSL19} that one has the lowest
 boomerang uniformity $6$ in this family if and only if $n\not\equiv 0 \pmod{3}$.

\medskip

\noindent {\bf When $\crk(F) = 3$ :} from the same proposition,  we may assume $F(x) = [0,\gamma,1,x]$ for some
$\gamma\neq 0 \in \Fbn$. Since $\gamma\neq 0$, one has either $\gamma=1$ or $\gamma \in \Fbn\setminus \F_2$.
\begin{enumerate}[]
\item  {\bf Case : $\gamma = 1$ with $\crk(F)=3$} \\
 Then $F(x) = [0,1,1,x]=((x^{q-2}+1)^{q-2}+1)^{q-2}$ is written as
$F(x) =
\begin{cases}
0 &\text{ if } x = 0 \\
1 &\text{ if } x = 1 \\
x+1 &\text{ otherwise}.
\end{cases}$ \\

\noindent This is the case where $\alpha_3=0$ such that
 $F(x)=x+1=\tfrac{\alpha_4x+\beta_4}{\beta_3}$  for $x\notin \{0,1\}$. One easily checks that
    $F(x) + F(x+1) = 1$ for all $x\in \Fbn$.
 Therefore we have the value $2^n$ as the differential and boomerang uniformity, which is uninteresting.
 \end{enumerate}

\noindent
Now assume $\gamma\neq 0,1$. Then from the information of $a_0=1,a_1=0,a_2=1,a_3=\gamma, a_4=0$,
   using the recurrences in \eqref{recurrence},
one gets $\{(\alpha_i,\beta_i)=(0,1),(1,0),(1,1),(\gamma +1,\gamma),(1,1) \,\, |\,\, i=0,1,2,3,4 \}$. Therefore by the
equations \eqref{alp,bet notation}, \eqref{frequality}, \eqref{frexception}, on gets
$$
F(x)=R_3(x)=\frac{\alpha_4x+\beta_4}{\alpha_3x+\beta_3}=\frac{x+1}{(\gamma+1)x+\gamma} \quad \mbox{for}\,\, x\notin
           \left\{[a_1,\ldots,a_i] \,\, | \,\, i=1,2,3\right\}=\left\{0,1,\frac{\gamma}{\gamma+1}\right\}
$$
Also, from the Lemma \ref{inverse lemma}, one has
$$
\ell_1(x)=\frac{a_0x+\beta_3}{\alpha_3}=\frac{x+\gamma}{\gamma+1}, \quad \ell_2(x)=\alpha_3x+\alpha_4=(\gamma+1)x+1
$$
and
\begin{align*}
\ell_2\circ F\circ\ell_1(x) =
          \begin{cases}
            \dfrac{1}{x} &\text{ for }  x\notin
           \left\{\alpha_3[a_1,\ldots,a_i]+\beta_3 \,\, | \,\, i=1,2,3\right\}=\left\{\gamma,1,0\right\} \\
                 0, \dfrac{1}{\gamma}, 1 &\text{ for } x= \gamma,1,0, \,\,\mbox{respectively}
          \end{cases}
\end{align*}
That is, letting $f_3(x)=\ell_2\circ F\circ\ell_1(x)$, $F(x)=[0,\gamma,1, x]$ is affine equivalent to

\begin{enumerate}[]
\item  {\bf Case :  $\gamma \in \Fbn\setminus \F_2 $ with $\crk(F)=3$}\\
  $$  f_3(x) = (\gamma+1) F(\tfrac{x+\gamma}{\gamma +1})+ 1 =
   \begin{cases}
1 &\text{ if } x= 0 \\
\tfrac1\gamma &\text{ if }  x= 1 \\
0 &\text{ if }  x= \gamma \\
\frac{1}{x} &\text{ if } x\neq 0,1,\gamma
   \end{cases}
  $$
\end{enumerate}
When $\gamma\in \F_4$ (i.e., $\gamma^2+\gamma+1=0)$, the differential uniformity of $f_3$ is well known in
\cite{LWY13e,ZHS14}.  Also the boomerang uniformity $\delta_{f_3}$ of $f_3$ with $\Delta_{f_3}=4$ is recently computed
in \cite{CV20} completely. For example, it is shown in \cite{CV20} that one has $\delta_{f_3}=6$ and $\Delta_{f_3}=4$
if and only if
 $n\equiv 2 \pmod{4}$ with $6 \not| n$.

However for general  $\gamma \in \Fbn \setminus \F_4$, no satisfactory answer exist so far. The boomerang uniformity is
unknown  for this general case. Also,
 some plenty cases of
differential uniformity of $f_{3}$ is dealt in \cite{LWY13e} but no complete classification exist so far.
 Especially, there is no known \emph{polynomial time algorithm} for determining the differential uniformity of $f_3$ with
arbitrary $\gamma \in \Fbn \setminus \F_4$.
  In Section \ref{differentialsec} and \ref{boomerangsec}, we will give satisfactory answers for the differential
  and the boomerang uniformity of arbitrary $f_3$ (i.e., permutations of Carlitz rank 3).

\medskip

Related to the function $f_3$, there are some approach \cite{LWY13e,LQSL19,CV20} to study cryptographic parameters of
the permutations of the form $F(x)=\pi(x)^{q-2}$ with $\pi=(c_1\, c_2\, \ldots\, c_k)$ a (cyclic) $k$-cycle, i.e.,
 $\pi(c_i)=c_{i+1}$ and $\pi(x)=x$ for $x\notin \{c_1,\ldots, c_k\}$. For example, $f_3(x)=\pi(x)^{q-2}$ with
$\pi=(0\, 1\, \gamma)$. We briefly show that these are (as expected) special cases of permutations of Carlitz form.

\begin{prop}\label{piinverse}
Let $F(x)=\pi(x)^{q-2}$ with $k$-cycle  $\pi=(c_1\, c_2\, \ldots\, c_k),\,\, k\geq 3$. Then one has
 the Carlitz rank $\crk(F) \leq 3k+2$ when $\pi$ contains no zero, and $\crk(F) \leq 3k-4$ when $\pi$ contains
  zero.
\end{prop}
\begin{proof}
Let $c\neq 0\in \Fbn$. Then it is easy to check (See p. 358 in \cite{Nie15})
$$f_c(x)=[0,\tfrac1c,c,\tfrac1c+\tfrac1{c^2}x]=
 \left(\left(\left(\tfrac1{c^2}x+\tfrac1c\right)^{q-2}+c\right)^{q-2}+\tfrac1c\right)^{q-2}$$
is a transposition $(0\, c)$. If $c_1$ and $c_2$ are nonzero,
 then $(c_1\,c_2)=(0\,c_1)(0\, c_2)(0\, c_1)=f_{c_1}\circ f_{c_2}\circ f_{c_1}$.
Then, any $k$-cycle $\pi=(c_1\, c_2\, \ldots\, c_k)$ containing no zero is written as
 \begin{align*}
 \pi &=(c_1\, c_{k})(c_1\,c_{k-1})\cdots (c_1\, c_2) \\
     &=(f_{c_1}\circ f_{c_k}\circ f_{c_1})\circ(f_{c_1}\circ f_{c_{k-1}}\circ f_{c_1}) \circ \cdots\circ
                     (f_{c_1}\circ f_{c_2}\circ f_{c_1}) \\
     &=f_{c_1}\circ f_{c_k}\circ f_{c_{k-1}} \circ \cdots\circ f_{c_2}\circ f_{c_1}
 \end{align*}
Therefore $\pi(x)^{q-2}= f'_{c_1}\circ f_{c_k}\circ f_{c_{k-1}} \circ \cdots\circ f_{c_2}\circ f_{c_1}$
 with $f'_{c_1}=[\tfrac1c,c,\tfrac1c+\tfrac1{c^2}x]$ and one finds $\crk (\pi(x)^{q-2})\leq 3k+2$.
When $\pi=(c_1\, c_2\, \ldots\, c_k)$ contains $0$, say $c_1=0$, then the same method shows
 \begin{align*}
 \pi =(0\, c_{k})(0\,c_{k-1})\cdots (0\, c_2) =f_{c_k}\circ f_{c_{k-1}}\circ \cdots \circ f_{c_2}
 \end{align*}
and one concludes $\crk (\pi(x)^{q-2})\leq 3k-4$.
\end{proof}
Although not the exact form of $F(x)=\pi(x)^{q-2}$ is discussed, the articles \cite{PT17,TCT15,ZHS14,QTTL13,QTLG16}
also study
 the differential uniformity of modifications of the inverse function, and it is also possible to describe their
 constructions via Carlitz form with Carlitz rank depending on the size of exceptional sets. However,
 when the exceptional set is large, we do not have a satisfactory method for computing
   the differential uniformity other than the approach
 provided in \cite{PT17,TCT15,ZHS14,QTTL13,QTLG16}.
\begin{remark}
Any permutation on $\F_q$ with $q=2^n$ can be regarded as an element of $S_q$, a symmetric group on $q$ letters, and
any permutation in $S_q$ is a composition of at most $\tfrac{3q}{2}$ transpositions of the form $(0\,c)$ with $c\neq
0$. Therefore an obvious upper bound of the Carlitz rank of permutations on $\F_q$ is $\tfrac{9q}{2}$. We will show, in
next section, that when a permutation $F$ on even dimension $($i.e., $q=2^n$ with $n$ even$)$ has Carlitz rank not too
large $($i.e., if $\crk(F)<\frac{q}{6}$$)$, then $F$ is not APN.
\end{remark}

\section{New Cryptographic Tool using Convergents of Continued Fraction over Finite Fields}

\subsection{Notion of Poles and Convergents}

  Most of the articles \cite{ACMT09, CMT08, CMT14} on permutation $F(x)=[a_{m+1}, a_m,\ldots, a_2, a_1+a_0x]$ of
 Carlitz form define poles as the set of elements $\tfrac{\beta_i}{\alpha_i}$ (which is in fact
 an analogue of convergent of continued fraction of real numbers), where
$\alpha_i$ and $\beta_i$ are defined via recurrence relations in the equation \eqref{recurrence}.
 However a natural definition of poles would be the set of all $x\in \Fbn$
where the linear fractional transformation approximation by $R_m(x)$ differs from $F(x)$. That is, the set of all
$x=\tfrac{1}{a_0}[a_1,a_2,\cdots, a_{i}]$
 for $1\leq i\leq m$ in the equation \eqref{frequality}.
 Since $F$ is a permutation, the inverse permutation $F^{-1}$ also has poles, and both poles have one to one
correspondence with each other, which is explained in the equation \eqref{frexception}.
 To summarize, one may guess that there are some
relations between the two sets
\begin{align*}
\{\tfrac{1}{a_0}[a_1,a_2,\cdots, a_{i}]  \,\,|\,\,  i=1,2,\cdots, m\} \quad\mbox{and}\quad
\{\tfrac{\beta_i}{\alpha_i}\,\,|\,\,  i=1,2,\cdots, m\}
\end{align*}
\begin{remark} Unlike the real number case, the expression $\tfrac{\beta_i}{\alpha_i}$ maybe undefined for some $i$ because
 $\alpha_i$ may be zero multiple times for some $0\leq i\leq m+1$.
 \end{remark}

 Our aim in this section is to show that the two sets above are (if well defined)
the same sets (excluding multiplicities), and also to show that there is well defined correspondence between  the two
sets. For given $a_0, a_1, \ldots, a_{m+1} \in \Fbn$ with $a_0, a_2, \ldots, a_m\neq 0$ and
  sequences $\alpha_i, \beta_i \,\, (0\leq i\leq m+1)$ satisfying the recurrences in \eqref{recurrence}
with initial conditions $\alpha_0=0, \alpha_1 = a_0$ and $\beta_0=1, \beta_1=a_1$, we define as follows.
\begin{definition}
\begin{align*}
A_i=\tfrac{1}{a_0}[a_1,a_2,\cdots, a_{i}] \,\,\, (1\leq i\leq m+1)
  \qquad \mbox{and} \qquad K_{i,j}=\alpha_i\beta_j+\alpha_j\beta_i  \,\,\, (0\leq i,j\leq m+1)
\end{align*}
\end{definition}

\begin{definition} \hfill
\begin{enumerate}

 \item For a set $\mathbb Z_{[1,m+1]}\overset{\text{\rm def}}{=}\{1,2,\ldots,m+1\}$, we define a relation $\sim$ on  $i, j\in \mathbb Z_{[1,m+1]}$
            as $i\sim j$ if $A_i=A_j$.

 \item For a set $\mathbb Z_{[0,m+1]}\overset{\text{\rm def}}{=}\{0,1,2,\ldots,m+1\}$, we define a relation $\approx$ on  $i, j\in \mathbb Z_{[0,m+1]}$
            as $i\approx j$ if $K_{i,j}=0$.
\end{enumerate}
\end{definition}

\noindent Using the above definitions, we obtain the following simple lemma.
\begin{lemma}\label{Kfundamental} \hfill
\begin{enumerate}[$(a)$]
\item Both the relations $\sim$ and $\approx$ are the equivalence relations on $\mathbb Z_{[1,m+1]}$ and on $\mathbb
Z_{[0,m+1]}$ respectively.

 \item For all $0\leq i,j\leq m+1$ with $j\geq 2$, one has the recurrence $K_{i,j}=a_jK_{i,j-1}+K_{i,j-2}$.
 \item For all $1\leq i \leq m+1$, $K_{i,i-1}=a_0$ $($i.e., $\alpha_i\beta_{i-1}+\beta_i\alpha_{i-1}=a_0$$)$
 \item For all $0\leq i\leq m+1$, $\alpha_i$ or $\beta_i$ is nonzero.
 \item For all $1\leq i\leq m+1$, $\alpha_i$ or $\alpha_{i-1}$ is nonzero.
 \item Let $1\leq i\leq m+1$. Then for all $1\leq j\leq m+1$, $K_{i,j}$ or $K_{i,j-1}$ is
nonzero.
 \item Suppose $\alpha_k=0=\alpha_i$ for some $0\leq k<i\leq m+1$, and suppose $\alpha_s\neq 0$ for all $k<s<i$. Then
            $[a_i, a_{i-1}, \ldots, a_{k+2}]=0$ with $k+2\leq i$.
 \item Suppose $K_{i,k}=0=K_{i,j}$ for some $1\leq i\leq m+1$ and $1\leq k<j \leq m+1$, and suppose $K_{i,s}\neq 0$ for all $k<s<j$. Then
            $[a_i, a_{i-1}, \ldots, a_{k+2}]=0$ with $k+2\leq i$.
\end{enumerate}
\end{lemma}

\begin{proof} By the Lemma \ref{cont. frac lemma} and by the definitions of $\alpha_i, \beta_i$, the assertion $(a)$ is
trivial. For the assertion $(b)$,
 $K_{i,j}=\alpha_i\beta_j+\beta_i\alpha_j=\alpha_i(a_j\beta_{j-1}+\beta_{j-2})+\beta_i(a_j\alpha_{j-1}+\alpha_{j-2})=$
$a_jK_{i,j-1}+K_{i,j-2}$. For the statement $(c)$, we use induction. That is,
$K_{1,0}=\alpha_1\beta_0+\beta_1\alpha_0=a_0\cdot 1+a_1\cdot 0=a_0$, and inductively,
$K_{i+1,i}=\alpha_{i+1}\beta_i+\beta_{i+1}\alpha_{i}
  =(a_{i+1}\alpha_i+\alpha_{i-1})\beta_i+(a_{i+1}\beta_i+\beta_{i-1})\alpha_i=\alpha_{i-1}\beta_i+\beta_{i-1}\alpha_i=K_{i,i-1}$.
Assertions $(d), (e)$ are clear from the assertion $(c) : \alpha_i\beta_{i-1}+\beta_i\alpha_{i-1}=a_0\neq 0$. For the
assertion $(f)$, if $K_{i,j}=0=K_{i,j-1}$, then from the recurrence in the statement $(b)$, one gets $K_{i,j}=0$ for
all $0\leq j\leq m+1$ which is a contradiction to the assertion $(c) : K_{i,i-1}=a_0\neq 0$.
 For the statement $(g)$, starting from the recurrence $0=\alpha_i=a_i\alpha_{i-1}+\alpha_{i-1}$ with $\alpha_{i-1}\neq
0$ by $(e)$,
\begin{align*}
0=\dfrac{\alpha_i}{\alpha_{i-1}}&=a_i+\dfrac{\alpha_{i-2}}{\alpha_{i-1}}
                               =[a_i, \dfrac{\alpha_{i-1}}{\alpha_{i-2}}]
                               = [a_i,a_{i-1}, \dfrac{\alpha_{i-2}}{\alpha_{i-3}}] \\
                              &\cdots \\
                             &= [a_i,a_{i-1}, \ldots, a_{k+3},\dfrac{\alpha_{k+2}}{\alpha_{k+1}}] \\
                             &= [a_i,a_{i-1}, \ldots, a_{k+3},a_{k+2}] \quad
                \because \dfrac{\alpha_{k+2}}{\alpha_{k+1}}=a_{k+2}+\dfrac{\alpha_{k}}{\alpha_{k+1}} \,\, \mbox{with}\,\, \alpha_k=0
\end{align*}
The proof of $(h)$ is exactly  same to that of $(g)$ because $K_{i,j}$ satisfies the same recurrence with $\alpha_j$ in
view of $(b)$.
\end{proof}
\begin{remark}
From the above lemma, let us denote the equivalent class of $i\in \mathbb Z_{[1,m+1]}$ as
 $\overset{\sim}{i}$, and denote the equivalent class of $j\in \mathbb Z_{[0,m+1]}$ as $\overset{\approx}{j}$.
 \end{remark}

\begin{prop}\label{propfundamental}
We have the followings.
\begin{enumerate}[$(1)$]
\item Suppose $\alpha_i\neq 0$ for some $1\leq i\leq m+1$. Let $0\leq j$ be the least integer satisfying
 $K_{i,j}=0$ $($i.e., $j \in \overset{\approx}{i}$ is the least element
    in the equivalent class of $i$$)$. Then one has $\alpha_j\neq 0$ and
$$
 \dfrac{\beta_i}{\alpha_i}=\dfrac{\beta_j}{\alpha_j}=A_j
$$
\item Conversely, for given $A_i$ with $1\leq i\leq m+1$, let $1\leq j$ be the least integer satisfying $A_i=A_j$
  $($i.e., $j \in \overset{\sim}{i}$ is the least element
    in the equivalent class of $i$$)$.
  Then one has $\alpha_j\neq 0$ and
$$
 A_i=\dfrac{\beta_j}{\alpha_j}
$$
\end{enumerate}
\end{prop}

\begin{proof} \hfill

\noindent $(1) $ Such $j$ exists and $j\leq i$ because $K_{i,i}=0$. From $0=K_{i,j}=\alpha_i\beta_j+\beta_i\alpha_j$
with $\alpha_i\neq 0$ and from Lemma \ref{Kfundamental}-$(d)$, one has $\alpha_j\neq 0$ and thus
 $\tfrac{\beta_i}{\alpha_i}=\tfrac{\beta_j}{\alpha_j}$. Now, since $K_{i,j}=0$ and $K_{i,s}\neq 0$ for all $0\leq s < j$,
   using the recurrence in Lemma \ref{Kfundamental}-$(b)$ repeatedly,
 \begin{align*}
0=\dfrac{K_{i,j}}{K_{i,j-1}}&=a_j+\dfrac{K_{i,j-2}}{K_{i,j-1}}
                               =[a_j, \dfrac{K_{i,j-1}}{K_{i,j-2}}]
                               = [a_j,a_{j-1}, \dfrac{K_{i,j-2}}{K_{i,j-3}}] \\
                              &\cdots \\
                             &= [a_j,a_{j-1}, \ldots, a_{2},\dfrac{K_{i,1}}{K_{i,0}}]
                             = [a_j,a_{j-1}, \ldots, a_{2},\dfrac{\alpha_ia_1+\beta_ia_0}{\alpha_i}]
 \end{align*}
Therefore using Lemma \ref{cont. frac lemma}-$(a)$,
$$
0=[\dfrac{\alpha_ia_1+\beta_ia_0}{\alpha_i}, a_2, \ldots, a_{j-1},a_j],
$$
which implies
$$
\dfrac{\beta_i}{\alpha_i}=\dfrac{1}{a_0}[a_1,a_2,\ldots,a_j]
$$

\noindent $(2) $ First we show $\alpha_j\neq 0$. If $\alpha_j=0$, then choose the greatest $0\leq k < j$
  satisfying $\alpha_k=0$. Such $k$ exists because $\alpha_0=0$. Then, since $\alpha_s\neq 0$ for all $k<s<j$,
  we get $[a_j,a_{j-1},\ldots, a_{k+2}]=0=[a_{k+2},\ldots, a_{j-1},a_j]$ with $k+2\leq j$ by the Lemma \ref{Kfundamental}-$(g)$
   and the Lemma \ref{cont. frac lemma}-$(b)$. Therefore by the Lemma \ref{cont. frac lemma}-$(a)$, we have $A_{k+1}=A_j$
    with $k+1<j$ which is a contradiction to the minimality of $j$ in the equivalent class of $\overset{\sim}{i}$.
     Now we show $A_j=\tfrac{\beta_j}{\alpha_j}$. By the assertion $(a)$ in this Proposition, we have
    $\tfrac{\beta_j}{\alpha_j}=\tfrac{\beta_{j_0}}{\alpha_{j_0}}=A_{j_0}$
    where $1 \leq j_0$ is the least integer satisfying $K_{j,j_0}=0$. We claim $j_0=j$ which completes the proof.
    If $j_0< j$,  choose the greatest $j_0\leq k<j$ satisfying $K_{j,k}=0$. Then, since $K_{j,s}\neq 0$
    for all $k<s<j$, we have $[a_j,a_{j-1},\ldots, a_{k+2}]=0=[a_{k+2},\ldots, a_{j-1},a_j]$ with $k+2\leq j$
     by the Lemma \ref{Kfundamental}-$(h)$  and the Lemma \ref{cont. frac lemma}-$(b)$. In the same way,
     we have $A_{k+1}=A_j$  with $k+1<j$ which is a contradiction to
     the minimality of $j$ in the equivalent class of $\overset{\sim}{i}$.
\end{proof}

The implication of the above proposition is as follows. In general, $A_i\neq \tfrac{\beta_i}{\alpha_i}$ even if
$\alpha_i\neq 0$. However if one chooses the minimal index $j$ satisfying $A_j=A_i$, then one has
$A_i=A_j=\tfrac{\beta_j}{\alpha_j}$ and this $j$ is the minimal index in the equivalent class of $\overset{\approx}{j}$
 (i.e., $s=j$ is the minimal for all $s$ satisfying $K_{j,s}=0$).
 Also, if one chooses the minimal index $j$ satisfying $\tfrac{\beta_j}{\alpha_j}=\tfrac{\beta_i}{\alpha_i}$
  (i.e., minimal $j$ satisfying $K_{i,j}=0$),
   then one has $\tfrac{\beta_i}{\alpha_i}=\tfrac{\beta_j}{\alpha_j}=A_j$
   and this $j$ is the minimal index in the equivalent class of $\overset{\sim}{j}$
    (i.e., $s=j$ is the minimal for all $s$ satisfying $A_s=A_j$).
Therefore, if one excludes the equivalent class $\overset{\approx}{0}$ which corresponds all the indices $i$ with
$\alpha_i=0$ (in particular, $\alpha_0=0$), the number of equivalent classes of $\mathbb Z_{[1,m+1]}$ and that of
$\mathbb Z_{[1,m+1]} \setminus \overset{\approx}{0}$ are the same, and the set of minimal indices (representatives)
among all equivalent classes of $\mathbb Z_{[1,m+1]}$ and that of $\mathbb Z_{[0,m+1]} \setminus \overset{\approx}{0}$
 are the same set.

\medskip

\noindent {\bf Example.} Let us give an example to briefly explain the above proposition.
   Let $m=8$ and let $F(x) =[a_{9},a_8,\ldots, a_2, a_1+a_0 x]= [0,g^{102},g^{153},g^{164},g^{234},g^{73},1,1,x]$ over $\F_{2^8}$
   with primitive element $g \in \F_{2^8}$ having the minimal polynomial $x^8 + x^4 + x^3 + x^2 + 1$. Then we have
    $$ A_1=A_3= 0,\quad A_2=A_7 = 1,\quad  A_4 =g^{236},\quad A_5 = g^{65},\quad A_6=A_8= A_9 = g^{251}$$
    such that $\mathbb Z_{[1,9]}$ is written as a disjoint union
    $$
   \mathbb Z_{[1,9]}=\{1,3\}\cup \{2,7\}\cup \{4\}\cup \{5\}\cup \{6,8,9\}
     =\{\overset{\sim}{1}\}\cup\{\overset{\sim}{2}\}\cup \{\overset{\sim}{4}\}\cup \{\overset{\sim}{5}\}\cup  \{\overset{\sim}{6}\}
    $$
    Also one has
    $$ \alpha_0=\alpha_3 = 0, \quad
    \tfrac{\beta_1}{\alpha_1} =\tfrac{\beta_7}{\alpha_7} =
    \tfrac{\beta_9}{\alpha_9} = 0, \quad \tfrac{\beta_2}{\alpha_2}=1, \quad
    \tfrac{\beta_4}{\alpha_4} =g^{236},\quad
 \tfrac{\beta_5}{\alpha_5} = \tfrac{\beta_8}{\alpha_8}=g^{65},\quad \tfrac{\beta_6}{\alpha_6} = g^{251}$$
 such that  $\mathbb Z_{[0,9]}$ is expressed as a disjoint union
 $$
  \mathbb Z_{[0,9]}=\{0,3\}\cup\{1,7,9\}\cup \{2\}\cup \{4\}\cup \{5,8\}\cup \{6\}
     =\{\overset{\approx}{0}\}\cup\{\overset{\approx}{1}\}\cup\{\overset{\approx}{2}\}\cup
      \{\overset{\approx}{4}\}\cup \{\overset{\approx}{5}\}\cup  \{\overset{\approx}{6}\}
 $$

\subsection{Classifying Solutions of $\du$ Equations}

From Proposition \ref{EAequivalent lemma}, every permutation of Carlitz rank $m$ can be written,
 up to affine equivalence, as
\begin{align}
 F(x)=[a_{m+1}, a_m, \ldots, a_2, a_1+a_0 x] =[0, a_m, \ldots, a_3, 1, x]
\end{align}
with $a_{m+1}=0, a_2=1, a_1=0$ and $a_0=1$. Letting $y=[0, a_m, \ldots, a_3, 1, x]$, using Lemma \ref{cont. frac
lemma}, one has
 $[y, a_m, \ldots, a_3, 1, x]=0=[x, 1, a_3,\ldots, a_m, y]$. Therefore the (compositional) inverse function
 $G(x)=F^{-1}(x)$ is written as
\begin{align}
G(x)= [0, 1, a_3,\ldots, a_m, x]
\end{align}
with the same $a_0=1, a_1=0, a_2=1$ and $a_{m+1}=0$. Since the expression $[a_1, a_2, a_3,\ldots, a_m, x]$ has
 the form of increasing indices and it is closely related with the equation \eqref{frexception},
we will be interested in the cryptographic properties of $G(x)$ of this form. Having that in mind, we
 introduce new definitions of differential and boomerang uniformities, which are in fact swapped versions from the
 original ones.
\begin{definition}
Let $G$ be any permutation on $\Fbn$ and let $a, b, c \in \Fbn\setminus \{0\}$. Then we define
\begin{align*}
 \du_G(a,b) &= \# \{(x,y) \in \Fbn\times\Fbn  \,|\, G(x)+G(y)=a \,\, \mbox{{\rm and}}\,\, x+y=b \} \\
 \bu_G(a,c) &= \# \{(x,y) \in \Fbn\times\Fbn \,|\, G(x)+G(y)=a=G(x+c)+G(y+c) \}
\end{align*}
We will call the simultaneous equation $a=G(x)+G(y), \,\,\, b=x+y$ as $\du$ equation, and also
 call the simultaneous equation  $G(x)+G(y)=a=G(x+c)+G(y+c)$ as $\bu$ equation.
\end{definition}
%\begin{remark}
%We will call the simultaneous equation $a=G(x)+G(y), \,\,\, b=x+y$ as $\du$ equation, and also
% call the simultaneous equation  $G(x)+G(y)=a=G(x+c)+G(y+c)$ as $\bu$ equation.
%\end{remark}
Please note that our (equivalent) definition of boomerang uniformity comes from the result of
 theorem 2.3 in \cite{LQSL19}. When $G$
and $F$ are inverse to each other, one has  $\du_G(a,b)=\ddt_F(a,b)$
 and  $\bu_G(a,b)=\bct_F(a,b)$. We also define the maximum values among all $\du_G(a,b)$ and $\bu_G(a,b)$ as follows.
\begin{definition}
\begin{align*}
 \du_G =\displaystyle \max_{a,b \in \Fbn \setminus \{0\}} \du_G(a,b), \qquad
 \bu_G =\displaystyle \max_{a,b \in \Fbn \setminus \{0\}} \bu_G(a,b)
\end{align*}
\end{definition}
\begin{remark}
$\du_G$ and $\bu_G$ are same to the original differential and boomerang uniformities of $G$ $($and also $F$$)$. That
is,
 $\du_G=\du_F=\Delta_F=\Delta_G$ and $\bu_G=\bu_F=\delta_F=\delta_G$.
 \end{remark}

If $(u,v)$ is a solution of $\du$ equation $G(x)+G(x+b)=a$, then
\begin{align}\label{dueqn}
  a=G(u)+G(v)=[0,1,a_3,\ldots,a_m,u]+[0,1,a_3,\ldots,a_m,v], \quad \mbox{with}\,\, u+v=b
\end{align}
By the symmetric nature of the above equation with respect to $u, v$, we will use the notion of
 unordered  solution pair $\{u,v\}$, where one unordered solution pair $\{u,v\}$ corresponds to
 two (ordered) solution pairs $(u,v)$ and $(v,u)$. Note that $(u,v)\neq (v,u)$ because $a\neq 0 \neq b$.

\begin{definition}\label{defpole} We define a domain pole set $\pole$ $($of $G$$)$ consisting of all $p_i$ $(1\leq i\leq m)$ satisfying
$A_i=G(p_i)$ {\rm (}i.e., $[0,1,a_3,a_4,\ldots,a_i]=[0,1,a_3,\ldots, a_m, p_i]$ {\rm )}. That is, a set consisting of
all $p_i=[0,a_m,a_{m-1},\ldots, a_{i+1}]$.
\end{definition}
\begin{remark}\label{remarkcarlitz3}
When $G(x)=[0,1,\gamma,x]$ with $\gamma\neq 0,1$, then one has
  $$p_1=[0, \gamma, 1]=\tfrac{1}{\gamma+1},\quad p_2=[0,\gamma]=\tfrac{1}{\gamma}\quad \mbox{and}\,\,\, p_3=[0]=0$$
 \end{remark}
 \begin{remark}
 In general, when $\crk(G)>3$,
  $p_i$ are not distinct since $A_i$ are not. When duplication happens, we exclude multiplicities
  by considering minimal index $i$ in the equivalent class $\overset{\sim}{i}$ for $A_i$, so that we assume $\pole$ consists of
  all distinct elements.
  \end{remark}

  If $\{u,v\}$ is a solution pair of the $\du$ equation \eqref{dueqn}, there are two
  possibilities for $u$ (resp. $v$). That is, either $u=p_i\in \pole$ for some $i$ or $u\notin \pole$.
   If $u=p_i\in \pole$, then $G(u)=A_i=[0,1,a_3,\ldots, a_i]$. If $u\notin \pole$, then we have the following.
  \begin{lemma}\label{gulemma}
  If $u\notin \pole$, then one has $G(u)=\dfrac{\beta_m u+\beta_{m-1}}{\alpha_m u+\alpha_{m-1}}$. In particular,
   $\alpha_m u+\alpha_{m-1}\neq 0$ for all $u\not\in \pole$.
  \end{lemma}
  \begin{proof}
Since $G(u)=[0,1,a_3,\ldots, a_m, u] \neq A_i$ for $1\leq i \leq m$, by the Proposition \ref{propfundamental}-(2),
 $A_{m+1,u}\overset{{\rm def}}{=}G(u)=[0,1,a_3,\ldots, a_m, u]$ itself gives a single equivalence class  of the index $m+1$ such that
$G(u)=A_{m+1,u}=\dfrac{a_{m+1}'\beta_m+\beta_{m-1}}{a_{m+1}'\alpha_m+\alpha_{m-1}}=\dfrac{\beta_m
u+\beta_{m-1}}{\alpha_m u+\alpha_{m-1}}$ with $a_{m+1}'=u$.
  \end{proof}

\begin{definition}
Let $u\neq v$ be in $\Fbn$. For a pair $\{u,v\}$ $($resp. an ordered pair $(u,v)$$)$, we say  $\{u,v\}$  is
 \begin{enumerate}[--]
 \item of $\ca$ if $u,v\in \pole$, and the set consisting of all $\{u,v\}$ $($resp. $(u,v)$$)$ of $\ca$ is
            denoted by $\mathcal{A}$. $($i.e., $\{u,v\}\in \mathcal{A}$ means $u\in \pole, v\in\pole$.$)$
 \item of $\cb$ if only one of $u$ and $v$ is in $\pole$,  and the set consisting of
     all $\{u,v\}$ $($resp. $(u,v)$$)$ of $\cb$ is denoted by $\mathcal{B}$.
        $($i.e., $\{u,v\}\in \mathcal{B}$ means $u \in \pole, v\not\in \pole$ or $u\not\in \pole, v\in \pole$.$)$
 \item of $\cc$ if none of $u$ and $v$ is in $\pole$,  and the set consisting of all $\{u,v\}$ $($resp. $(u,v)$$)$
                of $\cc$ is denoted by $\mathcal{C}$  $($i.e., $\{u,v\}\in \mathcal{C}$ means
                $u \not\in \pole, v\not\in \pole$.$)$
 \end{enumerate}
 \end{definition}

\noindent For a given permutation $G : \Fbn \rightarrow \Fbn$, define a map
\begin{alignat}{3}\label{psig}
\psi_G : \Fbn &\times \Fbn &&\longrightarrow \quad &&\Fbn \times \Fbn \\
              (u&,v)       &&\mapsto             && (G(u)+G(v),u+v) \notag
\end{alignat}
 Then, for given nonzero $a,b\in\Fbn$, $\psi_G^{-1}(a,b)$ is the set consisting of $(u,v)\in \Fbn \times \Fbn$
 satisfying
 $a=G(u)+G(v),\,\, b=u+v$ so that one has $\du_G(a,b)=\# \psi_G^{-1}(a,b)$.

\bigskip

\begin{definition}
For a given permutation $G$ on $\Fbn$ and a map $\psi_G$ on $\Fbn\times \Fbn$, we define
$$
\mathcal{A}_G^{(a,b)}=\mathcal{A} \cap \psi_G^{-1}(a,b),
           \quad \mathcal{B}_G^{(a,b)}=\mathcal{B} \cap \psi_G^{-1}(a,b),
            \quad \mathcal{C}_G^{(a,b)}=\mathcal{C} \cap \psi_G^{-1}(a,b)
$$
\end{definition}
\begin{remark} From the above definition, one has the following disjoint union and the formula for $\du_G(a,b)$
$$\psi_G^{-1}(a,b)=\mathcal{A}_G^{(a,b)}\cup \mathcal{B}_G^{(a,b)}\cup \mathcal{C}_G^{(a,b)},
 \qquad \du_G(a,b)= \#\mathcal{A}_G^{(a,b)} +\# \mathcal{B}_G^{(a,b)}+\# \mathcal{C}_G^{(a,b)}
$$
\end{remark}

\noindent Now from the equation \eqref{dueqn}, Definition \ref{defpole} and Lemma \ref{gulemma}, one finds that
\begin{enumerate}
 \item[--] $\mathcal{A}_G^{(a,b)}$ consists of all ordered pair $(u,v)\in \mathcal{A}$ satisfying the equation (\ref{dueqn}).
     In this case, one has $u=p_i, v=p_j$ for some $1\leq i\neq j\leq m$ such that $a=G(p_i)+G(p_j)=A_i+A_j$
      with $b=p_i+p_j$.
 \item[--] $\mathcal{B}_G^{(a,b)}$ consists of all ordered pair $(u,v)\in \mathcal{B}$ satisfying the equation (\ref{dueqn}).
      In this case, since only one of $u,v$ is in $\pole$, say $u=p_i \in \pole$, using Lemma \ref{gulemma}, one has the
    equation $a=G(p_i)+G(v)=A_i+\dfrac{\beta_{m}v+\beta_{m-1}}{\alpha_m v+\alpha_{m-1}}$ with $b=p_i+v$.
 \item[--] $\mathcal{C}_G^{(a,b)}$ consists of all ordered pair $(u,v)\in \mathcal{C}$ satisfying the equation (\ref{dueqn}).
     In this case, since $u,v\notin \pole$, using Lemma \ref{gulemma},
     one has $a=G(u)+G(v)
     =\dfrac{\beta_{m}u+\beta_{m-1}}{\alpha_m u+\alpha_{m-1}}+\dfrac{\beta_{m}v+\beta_{m-1}}{\alpha_m v+\alpha_{m-1}}$
      with $b=u+v$.
\end{enumerate}
Again, by the symmetry of the mapping $\psi_G$ (i.e., $\psi_G(u,v)=\psi_G(v,u)$), the cardinality of
$\mathcal{A}_G^{(a,b)}, \mathcal{B}_G^{(a,b)}$ and $\mathcal{C}_G^{(a,b)}$ are even, and one can also define analogue
versions for unordered sets for $\mathcal{A}_G^{(a,b)}, \mathcal{B}_G^{(a,b)}$ and $\mathcal{C}_G^{(a,b)}$. So when we
say $\{u,v\}$ is in $\mathcal{A}_G^{(a,b)}$, we are abusing our notations so that it implies both $(u,v)$ and $(v,u)$
are in $\mathcal{A}_G^{(a,b)}$.

\begin{prop}\label{abalpham}
Suppose that a  permutation $G(x)=[0,1,a_3,\ldots,a_m,x]$ has $\alpha_m\neq 0$. Then one has $\#
\mathcal{C}_G^{(a,b)}=0$ or $2$, and one gets
 $\# \mathcal{C}_G^{(a,b)}=2$ if and only if $\tr(\frac{1}{ab{\alpha_m}^2})=0$
 and $x^2+b\alpha_mx + \tfrac{b}{a} \neq 0 $ for all $x \in \alpha_m\pole + \alpha_{m-1}$.
\end{prop}
\begin{proof}
Suppose $\{u,v\} \in  \mathcal{C}_G^{(a,b)}$. Since $u, v \notin \pole$, by the Lemma \ref{gulemma}, one has
 \begin{align*}
 a=G(u)+G(v)&=\frac{u\beta_m+\beta_{m-1}}{u\alpha_m+\alpha_{m-1}}+\frac{v\beta_m+\beta_{m-1}}{v\alpha_m+\alpha_{m-1}} \\
            &=\frac{(u+v)(\alpha_m\beta_{m-1}+\beta_m\alpha_{m-1})}{(u\alpha_m+\alpha_{m-1})(v\alpha_m+\alpha_{m-1})}
            =\frac{b}{(u\alpha_m+\alpha_{m-1})(v\alpha_m+\alpha_{m-1})}
 \end{align*}
 where the last equality comes from Lemma \ref{Kfundamental}-(c). Therefore, letting $M=u\alpha_m+\alpha_{m-1}$ and
  $N=v\alpha_m+\alpha_{m-1}$,   the following quadratic equation
  $$
  0=(X+M)(X+N)=X^2+b\alpha_m X+ \frac{b}{a} \in \Fbn[X]
  $$
  has unique solutions $M,N \in \Fbn$ and thus
  $\tr(\frac{1}{ab{\alpha_m}^2})=0$.  Since $u,v \notin \pole$,
  the roots $M,N$ are not in $\alpha_m\pole + \alpha_{m-1}$, which means
   $x^2+b\alpha_mx + \tfrac{b}{a} \neq 0 $ for all $x \in \alpha_m\pole + \alpha_{m-1}$.
   Conversely when  $\tr(\frac{1}{ab{\alpha_m}^2})=0$ with
  $x^2+b\alpha_mx + \tfrac{b}{a} \neq 0 $ for all $x \in \alpha_m\pole + \alpha_{m-1}$, then the solution
  $\{M,N\}$ of the quadratic equation satisfies $M,N \notin \alpha_m\pole + \alpha_{m-1}$. Therefore
  $u=\tfrac{M+\alpha_{m-1}}{\alpha_m}, v=\tfrac{N+\alpha_{m-1}}{\alpha_m}$ are not in $\pole$,
  and thus $\{u,v\} \in \mathcal{C}$.
\end{proof}

\begin{remark}
Letting $x_i=\alpha_m p_i+\alpha_{m-1}$ for each $p_i\in \pole$, the last condition $x_i^2+b\alpha_mx_i + \tfrac{b}{a}
\neq 0 $ in the above proposition can be rephrased as a parametrized form $a\neq \frac{b}{\alpha_mx_ib+x_i^2}$ for all
$i$.
\end{remark}
\begin{remark} Consequence of the above proposition is that, if $\alpha_m\neq 0$ and $\# \psi_G^{-1}(a,b)=\du_G(a,b)$
  is greater than two, then those extra solutions $(u,v)$ are ALL from $\ca$ or $\cb$
 (i.e., at least one of $u,v$ are in $\pole$).
\end{remark}
 Since the number $\# \pole$ is the number of distinct $A_i \, (1\leq i\leq
 m)$, it is at most $m$. Therefore, when $m$ is small, we can do enumerative counting or we may even use an algorithmic
 approach for relatively small $m$, which is  mechanical in nature and can be implemented on software.

\subsection{Carlitz Rank of APN Permutation}

Using Proposition \ref{abalpham}, we present a simple result on the relation between APN and Carlitz rank.
\begin{theorem}
    Let $n$ be even and let $G(x) = [0,1,a_3, \ldots, a_m,x]$ be a permutation of Carlitz rank $m$ over
    $\Fbn$ with $\#\pole = \ell$.
    If $\ell < \dfrac{2^{n-1}}{3} $, then $G$ is not APN.
\end{theorem}
\begin{proof}
 Let us first suppose $\alpha_m = 0$, then by Lemma \ref{gulemma},
  we have $G(x) = \frac{x\beta_m+\beta_{m-1}}{\alpha_{m-1}}$ for all $x \notin \pole$. Therefore fixing
  any nonzero $b\in \Fbn$, one has
   $$G(u) + G(u+b) = \frac{u\beta_m+\beta_{m-1}}{\alpha_{m-1}}+ \frac{(b+u)\beta_m+\beta_{m-1}}{\alpha_{m-1}}
   =\frac{b\beta_m}{\alpha_{m-1}}\neq 0$$
   for all $u\in \Fbn$ satisfying  $u, u+b  \notin \pole$. Since there are at most $2\ell$ possible values of $u\in \Fbn$
   such that either $u$ or $b+u$ are in $\pole$, there exists at least $2^n-2\ell >2^{n-1}+\ell > 8$ number of solutions for the
   $\du$ equation $a=G(x)+G(x+b)$ with $a=\frac{b\beta_m}{\alpha_{m-1}}$ and $b$.
   Now let us consider the case $\alpha_m\neq 0$. Using Proposition \ref{propfundamental}, one
    chooses $p_k\in \pole$ such that
$G(p_k)=\frac{\beta_m}{\alpha_m}$. Then, for all $b$ with $b+p_k\notin \pole$, one has
$$
a=G(p_k)+G(b+p_k)=\dfrac{\beta_m}{\alpha_m}+\dfrac{(b+p_k)\beta_m+\beta_{m-1}}{(b+p_k)\alpha_m+\alpha_{m-1}}
         =\dfrac{1}{\alpha_m(b\alpha_m+p_k\alpha_m+\alpha_{m-1})},
$$
 and thus
$$
ab{\alpha_m}^2+a\alpha_m(p_k\alpha_m+\alpha_{m-1})=1 \quad\mbox{i.e.,}\,\,\,
   \frac{1}{ab{\alpha_m}^2}=1+\frac{p_k\alpha_m+\alpha_{m-1}}{b\alpha_m}
$$
If $p_k\alpha_m+\alpha_{m-1}\neq 0$, then $1+\frac{p_k\alpha_m+\alpha_{m-1}}{b\alpha_m}$ takes all the values of
$\Fbn\setminus \{1\}$ and $\tr(\frac{1}{ab{\alpha_m}^2})=0$ for $2^{n-1}$ values of $b$.
   Because $a=\frac{1}{\alpha_m(b\alpha_m+p_k\alpha_m+\alpha_{m-1})}$ is expressed as a parametrization of $b$,
     for each $1\leq i\leq \ell$, $a= \frac{b}{\alpha_mx_ib+x_i^2}$ is possible only (if they exist) for two values of $b\in \Fbn$
   since $b$ is a root of nontrivial quadratic equation
    $\alpha_m^2 b^2+\alpha_m(p_k\alpha_m+\alpha_{m-1}+x_i)b+x_i^2=0$, where $x_i=\alpha_mp_i+\alpha_{m-1}$.
    Thus, by removing these $2\ell$ values of $b$, one can use
   Proposition \ref{abalpham}, and also removing at most $\ell$ values of $b$
   where $b+p_k\in \pole$, we conclude that, assuming $3\ell<2^{n-1}$, there exist $a, b\in \Fbn$ with
    $b+p_k\notin \pole$ and $u,v\notin \pole$ such that
    $$
    G(p_k)+G(b+p_k)=a=G(u)+G(v)
   $$
   Therefore one has $\du_G(a,b)\geq 4$ and $G$ is not APN.
  If $p_k\alpha_m+\alpha_{m-1}= 0$, then one has $\tfrac{1}{ab{\alpha_m}^2}=1$ for all $b$
  and thus $\tr(\tfrac{1}{ab{\alpha_m}^2})=0$ because
  $n$ is even.
\end{proof}
\begin{remark}

\hfill
\begin{enumerate}
 \item Since $\ell\leq m$ where $m=\crk(G)$, if $m< \dfrac{2^{n-1}}{3}$ then $G$ is not APN. That is, any APN permutation
$G$ on $\Fbn$ with $n\equiv 0 \pmod{2}$ must have $\crk(G)\geq \dfrac{2^{n-1}}{3}$. In particular, the APN permutation
on $\mathbb F_{2^6}$ found by Dillon \cite{BDMW} has Carlitz rank $\geq 11$.
 \item No APN permutation on $\F_{2^8}$ is found at this moment.
 However, if it exists, it must have Carlitz rank $\geq \dfrac{2^7}{3}\approx 42.67$
\end{enumerate}
\end{remark}

\section{Permutations of Carlitz Rank Three and $\alpha, \beta$ of Convergents}\label{sectioncarlitz3}

Using the tools that we presented in previous sections, we now show how the cryptographic properties such as boomerang
and differential uniformities of a permutation of Carlitz rank $3$ can be computed. Therefore let us consider Carlitz
rank $3$ case, i.e. $G(x) = [0,1,\beta,x]=((x^{q-2}+\beta)^{q-2}+1)^{q-2}$ with $\beta\neq 0,1$.
 Please note that $G(x)$ is a (compositional) inverse of $F(x)=[0,\beta,1,x]$ and they are affine equivalent to each
 other
 via the relation $G(x)=\beta F(\beta x)$. When $\beta\neq 0,1$,
it is trivial to check $\crk(G)=3$ (i.e., $\crk(G)<3$ is not possible.) As we already mentioned in Remark
\ref{remarkcarlitz3}, we have
  $p_1=[0, \beta, 1]=\tfrac{1}{\beta+1}, p_2=[0,\beta]=\tfrac{1}{\beta}, p_3=0$. We now introduce a new constant $\alpha$
   defined as $\alpha=\beta+1$, which will make our subsequent computations much easier.
Then we can rewrite the lists of  pole set $\pole$ and image pole set $G(\pole)$ as
\begin{align}
\pole &=\{p_1,p_2,p_2\}=\{[0, \beta, 1], [0,\beta], [0]\}=\{\tfrac1\alpha,\tfrac1\beta,0\} \label{infopole}\\
G(\pole)&=\{G(p_1),G(p_2),G(p_3)\}=\{A_1, A_2, A_3\}=\{0, 1, \tfrac{\beta}{\alpha}\} \label{infogpole}
\end{align}
Since all $A_i$ are distinct, one has the following nice correspondence by Proposition \ref{propfundamental},
$$A_1=0=\frac{\beta_1}{\alpha_1},\quad
A_2=1=\frac{\beta_2}{\alpha_2}, \quad A_3=\frac\beta\alpha=\frac{\beta_3}{\alpha_3}$$
 Also,
 since  $G$ has the form of
fraction when $x \notin \pole$ by  Lemma \ref{gulemma}, we have
 \begin{equation}
 G(x) = \begin{cases}
  A_i &\text{ if } x = p_i \in \pole \\
\dfrac{\beta x + 1}{\alpha x + 1} &\text{ if } x \notin \pole
\end{cases}
 \end{equation} where $\alpha = \beta + 1.$

\begin{remark}
Although our main interest in this article is about the boomerang and differential uniformity,
 the permutation
  $G(x)=[0,1,\beta,x]$ has other nice cryptographic properties such as high algebraic degree and high nonlinearity.
  For the definitions and fundamentals of the above notions, please refer \cite{Bud15,Can,Car10,Car11,CS17}.
  By Lagrange interpolation, one has
\begin{align*}
G(x) &= \tfrac{\beta}{\alpha} + \tfrac{1}{\alpha}(\alpha x + 1)^{2^n-2} + \tfrac{\beta}{\alpha}( x +
\tfrac{1}{\alpha})^{2^n-1} + ( x + \tfrac{1}{\beta})^{2^n-1} +\tfrac{1}{\alpha}x^{2^n-1} \\
     &= \tfrac{1}{\alpha\beta} x^{2^n-2}+ \,\,\, \mbox{lower terms},
\end{align*}
so that the algebraic degree of $G$ is $n-1$. It is shown in \cite{LWY13e} that the nonlinearity of $G$ is
  $\geq 2^{n-1}-2^{\frac{n}2}-3$, but using a refined technique, one can prove that the nonlinearity is
  $\geq 2^{n-1}-2^{\frac{n}2}-2$.
\end{remark}

For given nonzero $a,b \in \Fbn$, let us consider the quadratic polynomial (for the case $m=3$) in the proof of
Proposition \ref{abalpham},
\begin{align}\label{basicpoly}
 f(X)=X^2+b\alpha X+\frac{b}{a}=(X+\alpha u+1)(X+\alpha v+1) \in \Fbn[X],
\end{align}
where $a=G(u)+G(v), \,\, b=u+v$. Now we rephrase Proposition \ref{abalpham} for the case $G(x)=[0,1,\beta, x]$ of
Carlitz rank three.
\begin{corollary}\label{abalphamcoro}
For $G(x)=[0,1,\beta, x]$, one always has $\# \mathcal{C}_G^{(a,b)}\leq 2$,  and one has $\# \mathcal{C}_G^{(a,b)}=2$
 if and only if $\tr(\frac{1}{ab{\alpha}^2})=0$ with $a,b$
 satisfying $a\neq (\alpha a+1)b$ and  $a\neq \beta(\alpha a+\beta)b$.
\end{corollary}
\begin{proof}
In view of Proposition \ref{abalpham}, one only needs to check the condition $f(x)\neq 0$ for all $x\in \alpha\pole
+1$. Since $\pole=\{\tfrac1\alpha,\tfrac1\beta,0\}$, one has $\alpha\pole +1=\{0,\tfrac1\beta,1\}$. One can easily
 check $f(0)=\tfrac{b}{a}\neq 0$, $f(1)=1+b\alpha+\tfrac{b}{a}\neq 0$ is equivalent to $a\neq (\alpha a+1)b$,
  and $f(\tfrac1\beta)=\tfrac1{\beta^2}+\tfrac{b\alpha}{\beta}+\tfrac{b}{a}\neq 0$ is
  equivalent to $a\neq \beta(\alpha a+\beta)b$.
\end{proof}

By replacing $X$ by $\alpha X+1$ from $f(X)$, we get
\begin{align}
\hh^{(a,b)}(X) \stackrel{{\rm def}}{=} X^2+bX+\frac{b}{\alpha}+\frac{b}{\alpha^2a}+\frac{1}{\alpha^2}=(X+u)(X+v)
\label{eqnhh}
\end{align}

\noindent For each domain pole $p_i\in \pole=\{p_1, p_2, p_3\}=\{\frac{1}{\alpha}, \frac{1}{\beta}, 0\}$, we define
three polynomials
\begin{align}
h^{(a,b)}_1(X) &=
\hh^{(a,b)}(X+\frac{1}{\alpha})=X^2+bX+\frac{b}{\alpha^2a}=(x+\frac{1}{\alpha}+u)(X+\frac{1}{\alpha}+v) \label{h1x}\\
h^{(a,b)}_2(X) &=
\hh^{(a,b)}(X+\frac{1}{\beta})=
 X^2+bX+\frac{b}{\alpha\beta}+\frac{b}{\alpha^2a}+\frac{1}{\alpha^2\beta^2}=(x+\frac{1}{\beta}+u)(X+\frac{1}{\beta}+v) \label{h2x}\\
h^{(a,b)}_3(X) &= \hh^{(a,b)}(X+0)= X^2+bX+\frac{b}{\alpha}+\frac{b}{\alpha^2a}+\frac{1}{\alpha^2}=(x+u)(X+v)
\label{h3x}
\end{align}
%Note that both $h_2^{(a,b)}$ and  $h_3^{(a,b)}$ can be written as
%\begin{align}
%h_2^{(a,b)}(X)=h_1^{(a,b)}(X+\frac{1}{\alpha\beta}), \qquad h_3^{(a,b)}(X)=h_1^{(a,b)}(X+\frac{1}{\alpha})
%\end{align}
\begin{remark}\label{abalphamcororemark}
One easily checks $\hh^{(a,b)}(\tfrac1\alpha)=\tfrac{b}{\alpha^2 a}\neq 0$. Also $\hh^{(a,b)}(\tfrac1\beta)=0$ if and
only if $a=\beta(\alpha a+\beta)b$, and $\hh^{(a,b)}(0)=0$ if and only if $a=(\alpha a+1)b$. Therefore
 Corollary \ref{abalphamcoro} can be restated as follows; one has  $\#\mathcal{C}_G^{(a,b)}=2$ if and only if
$\hh^{(a,b)}(X)=0$ has solutions in $\Fbn$ none of which is a pole. Also note that the reducibility of $\hh$ is
equivalent to the reducibility of $h_i$ for any $i=1,2,3$.
\end{remark}

\noindent We also define bivariate polynomial $H_a(X,Y)$ as
\begin{align*}
H_a(X,Y)=X^2+XY+\frac{Y}{\alpha^2a}
\end{align*}
Note that we have a very natural relation (which we will use repeatedly in analyzing $\bu$ equations
 in Section \ref{boomerangsec}) between
$H_a(X,Y)$ and $h^{(a,t)}_1(X)=\hh^{(a,t)}(X+\frac{1}{\alpha})$. That is, if $X=s$ is a solution of $h^{(a,t)}_1(X)$,
then
\begin{align}
0=h^{(a,t)}_1(s)=s^2+st+\frac{t}{\alpha^2a}=H_a(s,t) \label{heqn}
\end{align}
We will check whether the above equation holds when $s,t$ are given as translations of $b_i$ which are linear
fractional transformations of $a$, and checking $H_a(s,t)=0$ is equivalent of finding a root in $\Fbn$ of certain
polynomials with indeterminate $a$.

\section{Differential Uniformity of $G$}\label{differentialsec}

By utilizing the methods that were introduced in previous sections, we want to derive information of the boomerang and
 differential uniformity of $G(x)=[0,1,\beta,x]$ on $\Fbn$. Our first goal is a complete classification
  of the differential uniformity of $G$, which will be needed when we discuss boomerang uniformity in Section
  \ref{boomerangsec}.

\begin{theorem}\label{dutheorem}
Let $\beta \in \Fbn\setminus \F_4$ with $\alpha=\beta+1$ and let
 $G(x)=[0,1,\beta,x]=((x^{2^n-2}+\beta)^{2^n-2}+1)^{2^n-2}$
 be a permutation of Carlitz rank three. Then one has the differential uniformity $\du_G\in \{4,6,8\}$
 where $\du_G=8$ if and only if $\beta^3+\beta^2+1= 0$.  If  $\beta^3+\beta^2+1\neq 0$, then by defining
  the following polynomials in $\Fbn[Z]$,
 \begin{align*}
    f_{12}(Z)&=Z^4+\dfrac1\beta Z^2+\dfrac\alpha\beta Z+\dfrac\alpha\beta,\\
    f_{13}(Z)&=Z^4+\alpha Z^2+\beta Z+1,
 \end{align*}
the following holds.
 \begin{enumerate}
 \item[1] If at least one of $f_{12}(Z)$ and $f_{13}(Z)$  has a root in $\Fbn$, then $\du_G=6$.
 \item[2] If none of $f_{12}(Z)$ and $f_{13}(Z)$  has a root in $\Fbn$, then $\du_G=4$.
 \end{enumerate}
\end{theorem}
\begin{remark}
The proof of the above theorem will be given at the end of this section. One novel property of the
 above theorem is that one can determine $\du_G$ very quickly. In other words, letting $f(Z)=f_{12}(Z)f_{13}(Z)$,
 one can compute $\gcd(f(Z), Z^{2^n}-Z)$ in just $O(n^3)$ bit operations using repeated squaring
  $Z^{2^i} \pmod{f}$. The $\gcd$ is nontrivial if and only if $f$ has a root in $\Fbn$.
  \end{remark}
\begin{remark}
   Please
   note that we replaced the trace conditions of many previous works with the solvability of quartic polynomials.
    $($In fact, the trace condition is equivalent to the solvability of quadratic polynomials.$)$  We will derive
  a similar theorem on boomerang uniformity later.
\end{remark}

From the information of $\pole$ and $G(\pole)$ in the equations \eqref{infopole} and \eqref{infogpole},
 there are only $3$ possibilities of  $\{u,v\} \in \mathcal{A}$
 and corresponding values of $(a,b)$. Namely
 \begin{alignat}{2}
 a&=1=G(\tfrac1\alpha)+G(\tfrac1\beta),\,\,\, b=\tfrac{1}{\alpha\beta} \quad & {when}\,\, & \{u,v\}=\{\tfrac1\alpha, \tfrac1\beta\} \label{Aclass12}\\
 a&=\tfrac\beta\alpha=G(\tfrac1\alpha)+G(0),\,\,\,b=\tfrac{1}{\alpha} \quad & {when}\,\, & \{u,v\}=\{\tfrac1\alpha, 0\} \\
 a&=\tfrac1\alpha=G(\tfrac1\beta)+G(0),\,\,\,b=\tfrac{1}{\beta} \quad & {when}\,\, & \{u,v\}=\{\tfrac1\beta, 0\}
 \end{alignat}

\begin{lemma}\label{aclasslemma}
Suppose that one has two pairs of  $\{u,v\} \neq \{u',v'\}$ satisfying
$$
G(u)+G(v)=a=G(u')+G(v')
$$
such that $\{u,v\}$ is of $\ca$. Then one has the followings.
\begin{enumerate}[$(a)$]
 \item $\{u',v'\}$ is not of $\ca$, i.e., $\{u',v'\}\in \mathcal{B}$ or $\{u',v'\}\in \mathcal{C}$.
 \item If $\{u',v'\}$ is of $\cb$, say $u'\in \pole$ and
 $v'\not\in \pole$, then one has $\{u,v,u'\}=\pole$ $($i.e., $\{u,v,u'\}$ and $\{\tfrac1\alpha,\tfrac1\beta,0\}$ are the same set$)$,
 and $v'$ is uniquely determined as $v'=\frac{\beta}{\alpha^2}$.
 \item If $u+v=u'+v'$, then $\{u',v'\} \in \mathcal{B}$ if and only if $\beta \in \F_4$.
\end{enumerate}
\end{lemma}

\begin{proof} \hfill

\noindent $(a)$. Suppose on the contrary that both $u',v'\in \pole$. Then, since there are only three poles, among
$u,v,u',v'$, two of them must be the same pole, say $u=u'$, then from $G(u)+G(v)=G(u')+G(v')$ and since $G$ is a
permutation, one has $v=v'$ which contradict $\{u,v\} \neq \{u',v'\}$.

\noindent $(b)$. It is clear (similarly as in the first statement) that all three $u,v,u'$ are distinct. Therefore the
equation $G(u)+G(v)=G(u')+G(v')$ is written as
$$
\tfrac1\alpha=0+1+\tfrac\beta\alpha=G(u)+G(v)+G(u')=G(v')
$$
  Then one can check $v'=\frac\beta{\alpha^2}$ is a unique solution, since $v'\not\in \pole$ implies
  $G(v')=\frac{\beta\cdot\frac{\beta}{\alpha^2}+1}{\alpha\cdot\frac{\beta}{\alpha^2}+1}=\frac{1}{\alpha}$.

  \noindent $(c)$. We use $(b)$ so that $\{u',v'\}\in \mathcal{B}$ if and only if
$\{u,v,u',v'\}=\{\tfrac1\alpha,\tfrac1\beta,0,\tfrac{\beta}{\alpha^2}\}$. Therefore $u+v=u'+v'$ if and only if
 $0=u+v+u'+v'=\tfrac1\alpha+\tfrac1\beta +\tfrac{\beta}{\alpha^2}=\tfrac{\alpha+\beta^2}{\alpha^2\beta}$, which
 happens exactly when $0=\alpha+\beta^2=1+\beta+\beta^2$.
\end{proof}

\begin{prop}\label{propaclass}
For  $(a,b) \in \{(1,\tfrac1{\alpha\beta}),(\tfrac\beta\alpha, \tfrac1\alpha), (\tfrac1\alpha, \tfrac1\beta)\}$, i.e.,
 when $a, b$ are expressed as
 $$
 a=G(p_i)+G(p_j), \quad b=p_i+p_j,
 $$
 one has $\du_G(a,b)=2,4,$ or $6$, and is classified as
 \begin{enumerate}[$(1)$]
  \item $\du_G(a,b)=6$ when {\rm (i)} $\beta\in \F_4$ and {\rm (ii)}
  $h_1^{(a,b)}(X)=X^2+bX+\tfrac{b}{\alpha^2 a}$ has a solution in $\Fbn$.
  \item $\du_G(a,b)=4$ if only one of the conditions among {\rm (i), (ii)} is satisfied.
  \item $\du_G(a,b)=2$ if none of the conditions among {\rm (i), (ii)} is satisfied.
 \end{enumerate}
\end{prop}

\begin{proof}
From Lemma \ref{aclasslemma}-$(a)$, one gets $\#\mathcal{A}_G^{(a,b)}=2$.
 From Lemma \ref{aclasslemma}-$(b)$ and $(c)$, one has $\# \mathcal{B}_G^{(a,b)}=0$ or $2$,
  and $\#\mathcal{B}_G^{(a,b)}=2$ happens exactly when $\beta\in \F_4$ such that
  $p_i+p_j+p_k+\tfrac{\beta}{\alpha^2}=0$ and $b=p_i+p_j=p_k+\tfrac{\beta}{\alpha^2}$ with
   $\{p_i,p_j\}\in \mathcal{A}$ and $\{p_k,\tfrac{\beta}{\alpha^2}\}\in \mathcal{B}$.
   Since one can easily check $a\neq (\alpha a+1)b$ and $a\neq \beta(\alpha a+\beta)b$ for
   all $(a,b) \in \{(1,\tfrac1{\alpha\beta}),(\tfrac\beta\alpha, \tfrac1\alpha), (\tfrac1\alpha, \tfrac1\beta)\}$,
   from Corollary \ref{abalphamcoro}, one has $\#\mathcal{C}_G^{(a,b)}=0$ or  $2$, and $\#\mathcal{C}_G^{(a,b)}=2$ if and only if
   $\tr(\tfrac1{ab\alpha^2})=0$ which is equivalent to the solvability of $h_1^{(a,b)}(X)=0$ in $\Fbn$ from the equation
       \eqref{h1x}.
\end{proof}

 The above proposition explains all possible values of $\du_G(a,b)$
 when the set $\mathcal{A}_G^{(a,b)}$ is not empty
 (i.e., when there are poles $p_i, p_j$ such that $a=G(p_i)+G(p_j)$ and $b=p_i+p_j$).
 Now we will discuss the other cases of $a,b$ satisfying $\du$ equation.
 Consider the following $\du$ equation, $a=G(x)+G(y), \,\,\, b=x+y$,
 where there is no solution $\{x,y\}=\{u,v\}$ $\in \mathcal{A}$ (i.e., no solution satisfying $u,v\in \pole$).
 Then any possible solution $\{u,v\}$ is either of $\cb$ or $\cc$. From Corollary \ref{abalphamcoro} and Remark \ref{abalphamcororemark}
 we know $\#\mathcal{C}_G^{(a,b)}=0$ or $2$, and  $\#\mathcal{C}_G^{(a,b)}=2$
 if and only if $\hh^{(a,b)}(X)=X^2+bX+\frac{b}{\alpha}+\frac{b}{\alpha^2a}+\frac{1}{\alpha^2}$
 has two solutions $u,v$ in $\Fbn$ such that $u,v\notin \pole$
 (i.e., $a\neq (\alpha a+1)b$ and $a\neq \beta(\alpha a+\beta)b$).
  Note that the reducibility of $\hh^{(a,b)}(X)$ in $\Fbn$ is also equivalent to
 $\tr(\tfrac1{ab\alpha^2})=0$.

 The case for $\cb$ can be explained as follows. When $\{u,v\} \in \mathcal{B}_G^{(a,b)}$, then since only
 one of $u, v$ is in $\pole=\{p_1,p_2,p_3\}=\{\tfrac1\alpha, \tfrac1\beta, 0\}$, one may write $\{u,v\}=\{p_i, b+p_i\}$
   such that
 \begin{align}
 a=G(p_i)+G(b+p_i) \quad \text{with}\,\,\, b+p_i \notin \pole
 \end{align}
For given nonzero $a$ and pole $p_i$, there is unique $b=b_i$ (depending on $a$ and $p_i$) satisfying the above
equation
 because $G$ is a permutation. For each pole $\tfrac1\alpha, \tfrac1\beta, 0$, one can easily find such $b_i$ as
 \begin{alignat}{3}
   \mathcal{B(1)} : \qquad a &=G(\tfrac1\alpha)+G(b_1+\tfrac1\alpha) \quad &&\text{with}&&\,\,\, b_1=\tfrac{1}{\alpha(\alpha a+\beta)} \label{eqnb1} \\
   \mathcal{B(2)} : \qquad a &=G(\tfrac1\beta)+G(b_2+\tfrac1\beta) \quad &&\text{with}&&\,\,\, b_2=\tfrac{a+1}{\beta(\alpha a+1)} \label{eqnb2} \\
  \mathcal{B(3)} : \qquad a &=G(0)+G(b_3) \quad       &&\text{with}&&\,\,\, b_3=\tfrac{\alpha a+1}{\alpha^2 a} \label{eqnb3}
 \end{alignat}
where $\mathcal{B({\it i})}$ means that $\{p_i, b_i+p_i\} \in \mathcal{B}$. Note that the denominators of $b_i$ are
never zero because $b_i+p_i \notin \pole$.
 \begin{remark} When we have $\{u,v\} \in \mathcal{B}_G^{(a,b)}$ such that one of $u,v$ is $p_i$,
         we simply say $\{u,v\}$ satisfies $\du$ equations of type $\mathcal{B}(i): a=G(u)+G(v),\,\, b_i=u+v$.
         Also,  throughout this paper, we will fix the values $b_1,b_2$ and $b_3$ as above.
 \end{remark}

For every $(a,b)=(a, b_i)$ parametrized by $a$, one has $\#\mathcal{B}_G^{(a,b_i)}\geq 2$ because
 there is an obvious solution pair $\{p_i, b_i+p_i\} \in \mathcal{B}$. Conversely, if there exists $(a,b)$ such
 that $\#\mathcal{B}_G^{(a,b)}\geq 2$, then there is $\{u,v\} \in \mathcal{B}_G^{(a,b)}=\mathcal{B} \cap \psi_G^{-1}(a,b)$
 satisfying $\du$ equation with $\{u,v\} \in \mathcal{B}$, which implies exactly one of $u, v$ (say
 $u$) is in $\pole$ so that $u=p_i$ for some $i$ and $b=b_i$ because $a=G(p_i)+G(b+p_i)$.

 Using the same argument, $\#\mathcal{B}_G^{(a,b)}\geq 4$ is possible only when $b=b_i=b_j$ for different $i, j$. In
 other words, if there exists $a\in \Fbn$ satisfying $b_i=b_j$, then letting $b=b_i(=b_j)$, one has $2$ solutions
  of unordered pairs  $\{u, v\}=\{p_i, b+p_i\}, \{p_j, b+p_j\}$ (or $4$ solutions of ordered pairs) satisfying
 $$
  a=G(u)+G(v), \quad b=u+v
 $$
 We summarize the conditions on $a$ satisfying $b_i=b_j$ as follows,
\begin{alignat}{3}
   a\in \Fbn \,\, \text{such that}\,\, b_1=b_2 \quad
   &\Leftrightarrow \quad &&\tfrac{1}{\alpha(\alpha a+\beta)}=\tfrac{a+1}{\beta(\alpha a+1)}
   &&\Leftrightarrow \quad a^2+a+\tfrac{\beta^2}{\alpha^2}=0 \quad
   \Leftrightarrow \quad \tr(\tfrac{\beta}{\alpha})=0 \label{eqnb1b2}\\
    a\in \Fbn \,\, \text{such that}\,\, b_1=b_3 \quad
    &\Leftrightarrow \quad &&\tfrac{1}{\alpha(\alpha a+\beta)}=\tfrac{\alpha a+1}{\alpha^2 a}
    &&\Leftrightarrow \quad  a^2+\tfrac{\beta}{\alpha}a+\tfrac{\beta}{\alpha^2}=0 \quad
    \Leftrightarrow \quad \tr(\tfrac{1}{\beta})=0 \label{eqnb1b3}\\
   a\in \Fbn \,\, \text{such that}\,\, b_2=b_3 \quad
   &\Leftrightarrow \quad && \tfrac{a+1}{\beta(\alpha a+1)}=\tfrac{\alpha a+1}{\alpha^2 a}
   &&\Leftrightarrow \quad a^2+\tfrac{1}{\alpha}a+\tfrac{\beta}{\alpha^3}=0 \quad
   \Leftrightarrow \quad \tr(\tfrac{\beta}{\alpha})=0 \label{eqnb2b3}
\end{alignat}

Finally, $\#\mathcal{B}_G^{(a,b)}= 6$ is possible only when $b_1=b_2=b_3$, which implies that
 three quadratic equations in $a$ \eqref{eqnb1b2},\eqref{eqnb1b3},\eqref{eqnb2b3} have a common root.
 By solving $3$ simultaneous quadratic equations, we find that a common solution exists if and only if $\beta^3+\beta^2+1=0$,
 and $a=\beta$ is a unique solution with corresponding $b=1$. In this case, letting $(a,b)=(\beta, 1)$, one has
 $$ h_1^{(a,b)}(X)=X^2+bX+\tfrac{b}{\alpha^2a}=X^2+X+\tfrac{1}{\alpha^2\beta}=X^2+X+\tfrac{1}{\beta^4},$$
 where
 $$
\tr(\tfrac{1}{\beta^4})=\tr(\tfrac{1}{\beta}+\tfrac{1}{\beta^2})=0 \quad (\because \beta^3+\beta^2+1=0)
 $$
 Therefore  $h_1^{(a,b)}(X)=0$ has
  a solution in $\Fbn$, which implies one also has a solution $\{u,v\}$ in $\mathcal{C}_G^{(\beta, 1)}$ by Corollary \ref{abalphamcoro}.

 To summarize, we have the following.
 \begin{prop}
 For a given permutation $G(x)=[0,1,\beta, x]$ on $\Fbn$ of Carlitz rank three, one has $\du_G=8$ if and only if $\beta^3+\beta^2+1=0$.
 \end{prop}
 \begin{remark}
 It should be mentioned that $(a,b)=(\beta, 1)$ is the unique pair having the maximum $\du_G(\beta,1)=8$ when
 $\beta^3+\beta^2+1=0$, and the four pairs $\{u,v\}$ of solutions satisfying $\du$ equation $\beta=G(u)+G(v), \,\, 1=u+v$
 can be computed explicitly using the equation \eqref{eqnb1},\eqref{eqnb2},\eqref{eqnb3} and they are
 $$
  \{u,v \} : \quad \{\beta^2, \beta^2+1\} \in \mathcal{B(1)}, \,\,\,
  \{\beta^2+\beta, \beta^2+\beta+1\} \in \mathcal{B(2)}, \,\,\,
  \{0, 1\} \in \mathcal{B(3)}, \,\,\,
  \{\beta, \beta+1\} \in \mathcal{C},
 $$
 where $\beta^2=\tfrac1\alpha, \beta^2+\beta=\tfrac1\beta \in \pole$. Also note that the set of all $u$ and $v$
 appearing in the above list covers all the elements of $\F_8$.
 \end{remark}

 Among the three basic equations \eqref{eqnb1},\eqref{eqnb2},\eqref{eqnb3} related to the $\cb$,
 two equations (\ref{eqnb1}) and (\ref{eqnb3}) have the following important properties
 which will be used repeatedly when one needs to compute $\tr(\tfrac{1}{ab\alpha^2})$.
 \begin{lemma}\label{tracelemma}
 \hfill
\begin{enumerate}[$(a)$]
\item  Suppose $a=G(\tfrac1\alpha)+G(b+\tfrac1\alpha)$ for some $a,b\in \Fbn$, then one has
      $$
       \frac{\beta}{\alpha a}=\frac{1}{ab\alpha^2}+1
      $$
\item  Suppose $a=G(0)+G(b)$ for some $a,b\in \Fbn$, then one has
      $$
       \frac{1}{\alpha b}=\frac{1}{ab\alpha^2}+1
      $$
\end{enumerate}
 \end{lemma}

\begin{proof}
The above two equations are, in fact, the same equations  as in (\ref{eqnb1}) and (\ref{eqnb3}),
 however we will give proofs for clarity. For the first statement,
$$
a=G(\tfrac{1}{\alpha})+G(b+\tfrac{1}{\alpha})=
  0+\frac{\beta\left(b+\tfrac{1}{\alpha}\right)+1}{\alpha\left(b+\tfrac{1}{\alpha}\right)+1}=\frac{\beta
b+\tfrac{1}{\alpha}}{\alpha b}=\frac{\beta}{\alpha}+\frac{1}{b\alpha^2}
$$
By multiplying $\tfrac{1}{a}$ to the above equation, one gets desired expression. For the second statement,
$$
a=G(0)+G(b)=\frac\beta\alpha+\frac{\beta b+1}{\alpha b+1}=\frac{1}{\alpha(\alpha b+1)}=\frac{1}{b\alpha^2+\alpha}
$$
Therefore by multiplying $b\alpha^2+\alpha$ to the above equation,
$$
ab\alpha^2+a\alpha=1
$$
and we get the desired expression dividing both sides of the above equation by $ab\alpha^2$.
\end{proof}

 Now we are ready to state our complete classification of differential uniformity of a permutation $G$
 of Carlitz rank three. We already mentioned that $\du_G=8$ is the maximum value possible and this
 happens exactly when $\beta^3+\beta^2+1=0$, where $\du_G=8$ is obtained as
  $8=\#\mathcal{A}_G^{(a,b)}+\#\mathcal{B}_G^{(a,b)}+\#\mathcal{C}_G^{(a,b)}=0+6+2$ with $(a,b)=(\beta, 1)$.
  It is obvious that, for any nonzero $a$, $b$ and $\beta \notin \F_4$, one
   has  $\#\mathcal{A}_G^{(a,b)}=0$ if $\#\mathcal{B}_G^{(a,b)}\geq 2$
  by Lemma \ref{aclasslemma}-$(c)$.  Therefore, by assuming $\beta \notin \F_4$ and $\beta^3+\beta^2+1\neq 0$,
  if one has $\du_G=6$, then the only possible case is that there exists $(a,b)$ such that
  $\du_G(a,b)=\#\mathcal{A}_G^{(a,b)}+\#\mathcal{B}_G^{(a,b)}+\#\mathcal{C}_G^{(a,b)}=0+4+2$.
  Since one has $\#\mathcal{B}_G^{(a,b)}=4$ if and only if exactly two of $b_1, b_2, b_3$ are the same,
  to find the conditions for which $\#\mathcal{C}_G^{(a,b)}=2$, we only need to check $\tr(\tfrac{1}{ab\alpha^2})$
  at these points $(a,b)$.

\bigskip

When $b=b_1=b_2$, then from the equation (\ref{eqnb1b2}),
\begin{align*}
0 &=a^2+a+\dfrac{\beta^2}{\alpha^2}=1+\dfrac{\alpha}{\beta}\left(\dfrac{\beta}{\alpha a}\right)+\left(\dfrac{\beta}{\alpha a}\right)^2 \\
  &=\left(\dfrac{\beta}{\alpha a}\right)^2 + \dfrac{\alpha}{\beta}\left(\dfrac{\beta}{\alpha a}\right)+ 1
    =\left(\dfrac{1}{ab\alpha^2}+1\right)^2+ \dfrac{\alpha}{\beta}\left(\dfrac{1}{ab\alpha^2}+1\right)+ 1,
\end{align*}
 where the last equality comes from Lemma \ref{tracelemma}-$(a)$. For $b=b_1=b_2$, one can easily
 check the conditions $a\neq (\alpha a+1)b,  \beta(\alpha a+\beta)b$ in Corollary \ref{abalphamcoro}
 are satisfied. Therefore, one has $\#\mathcal{C}_G^{(a,b)}=2$ if and
 only if $\tr(\tfrac{1}{ab\alpha^2})=0$ and this happens when $\tfrac{1}{ab\alpha^2}=z+z^2$ for some $z\in \Fbn$.
 Thus by defining the following polynomial
 \begin{align}
 f_{12}(Z)&= (Z^2+Z+1)^2+ \dfrac{\alpha}{\beta}(Z^2+Z+1)+ 1  \,\,\, \in \Fbn[Z] \notag \\
          &= Z^4+\dfrac1\beta Z^2+\dfrac\alpha\beta Z+\dfrac\alpha\beta, \label{f12}
 \end{align}
 $\du_G(a,b)=6$ at $(a,b)$ (with $a$ satisfying the equation in (\ref{eqnb1b2}) and $b=b_1=b_2$)
 implies that the above quartic polynomial $f_{12}(Z)$ has a root $z\in \Fbn$ satisfying $f_{12}(z)=0$.
 Conversely, if there is $z\in \Fbn$ with $f_{12}(z)=0$, then letting $a=\tfrac{\beta}{\alpha(z^2+z+1)}$
 and $b=\tfrac{1}{\alpha(\alpha a+\beta)}$, one finds $\du_G(a,b)=6$.

 \bigskip

When $b=b_1=b_3$, then from the equation (\ref{eqnb1b3}),
\begin{align*}
0 &=a^2+\dfrac{\beta}{\alpha}a+\dfrac{\beta}{\alpha^2}=1+\left(\dfrac{\beta}{\alpha a}\right)+\dfrac1\beta\left(\dfrac{\beta}{\alpha a}\right)^2 \\
  &=\left(\dfrac{\beta}{\alpha a}\right)^2 + \beta \left(\dfrac{\beta}{\alpha a}\right)+ \beta
    =\left(\dfrac{1}{ab\alpha^2}+1\right)^2+ \beta \left(\dfrac{1}{ab\alpha^2}+1\right)+ \beta,
\end{align*}
 where the last equality again comes from Lemma \ref{tracelemma}-$(a)$.
For $b=b_1=b_3$, one can also easily
 check the conditions $a\neq (\alpha a+1)b,  \beta(\alpha a+\beta)b$ in Corollary \ref{abalphamcoro}
 are satisfied.
 Since one has $\#\mathcal{C}_G^{(a,b)}=2$ if and
 only if $\tr(\tfrac{1}{ab\alpha^2})=0$ and this happens when $\tfrac{1}{ab\alpha^2}=z+z^2$ for some $z\in \Fbn$,
 by defining the following polynomial
 \begin{align}
 f_{13}(Z)&= (Z^2+Z+1)^2+ \beta (Z^2+Z+1)+ \beta \notag \\
          &= Z^4+\alpha Z^2+\beta Z+1, \label{f13}
 \end{align}
 $\du_G(a,b)=6$ at $(a,b)$ (with $a$ satisfying the equation in (\ref{eqnb1b3}) and $b=b_1=b_3$)
 implies that the above quartic polynomial $f_{13}(Z)$ has a root $z\in \Fbn$ satisfying $f_{13}(z)=0$.
 Conversely, if there is $z\in \Fbn$ with $f_{13}(z)=0$, then letting $a=\tfrac{\beta}{\alpha(z^2+z+1)}$
 and $b=\tfrac{1}{\alpha(\alpha a+\beta)}$, one finds $\du_G(a,b)=6$.

 \bigskip

When $b=b_2=b_3$, then by using the equations (\ref{eqnb2}) and (\ref{eqnb3}), one gets
\begin{align}
         \tfrac{\beta b+1}{\alpha\beta b+1}=G(\tfrac1\beta)+G(b+\tfrac1\beta)=a=G(0)+G(b)=\tfrac{1}{\alpha(\alpha b+1)},
         \notag
\end{align}
and solving the above equation in terms of $b$, one has
\begin{align}
0 &=b^2+\dfrac1{\beta}b+\dfrac{1}{\alpha^2}
 =1+\dfrac\alpha\beta\left(\dfrac{1}{\alpha b}\right)+\left(\dfrac{1}{\alpha b}\right)^2 \label{bbeqnb2b3} \\
  &=\left(\dfrac{1}{\alpha b}\right)^2 + \dfrac\alpha\beta\left(\dfrac{1}{\alpha b}\right)+ 1
    =\left(\dfrac{1}{ab\alpha^2}+1\right)^2 + \dfrac\alpha\beta\left(\dfrac{1}{ab\alpha^2}+1\right)+ 1, \label{f23expression}
\end{align}
 where the last equality come from Lemma \ref{tracelemma}-$(b)$. Since one has $\#\mathcal{C}_G^{(a,b)}=2$ if and
 only if $\tr(\tfrac{1}{ab\alpha^2})=0$, we get the same polynomial
 $f_{12}(Z)$ by replacing $\tfrac{1}{ab\alpha^2}$ with $z+z^2$ from the expression (\ref{f23expression}).
Therefore, $\du_G(a,b)=6$ at $(a,b)$ (with $b$ satisfying the equation in (\ref{bbeqnb2b3}) and $b=b_2=b_3$)
 implies that $f_{12}(Z)$ has a root $z\in \Fbn$.
 Conversely, if there is $z\in \Fbn$ with $f_{12}(z)=0$, then letting $b=\tfrac{1}{\alpha(z^2+z+1)}$
 and $a=\tfrac{1}{\alpha(\alpha b+1)}$, one finds $\du_G(a,b)=6$.

\bigskip
\noindent {\bf Proof of Theorem \ref{dutheorem} :} \hfill

\medskip
\noindent Everything is already shown except that $G$ is not APN. That is, we must show that there do exist $a$ and $b$
such that  $\du_G(a,b)\geq 4$. By the equation \eqref{eqnb1}, one has  $\#\mathcal{B}_G^{(a,b)} \geq 2$ for every $(a,
b)=(a, b_1)$ parametrized by $a$.  Also, the parametrized $a=a$ and $b=\tfrac{1}{\alpha(\alpha a+\beta)}$ must satisfy,
by Lemma \ref{tracelemma}-$(a)$, $\tfrac{1}{ab\alpha^2} = 1 + \tfrac{\beta}{\alpha a}$. The exceptional values of $a$
where Corollary \ref{abalphamcoro} is not applicable (i.e., the $a$ satisfying $b=\tfrac{a}{\alpha a+1}$ or
$b=\tfrac{a}{\beta(\alpha a+\beta)}$) is at most two,
  since $\tfrac{1}{\alpha(\alpha a+\beta)}=\tfrac{a}{\beta(\alpha a+\beta)}$ is impossible and
   $\tfrac{1}{\alpha(\alpha a+\beta)}=\tfrac{a}{\alpha a+1}$ is possible only when there is $a\in \Fbn$
  satisfying $a^2+a+\tfrac{1}{\alpha^2}=0$.
    Since half of the trace values  $\tr(1 + \tfrac{\beta}{\alpha a})$ take zero,
    there exist $a,b$ such that $\tr(\tfrac{1}{ab\alpha^2}) = \tr( 1 + \tfrac{\beta}{\alpha a}) = 0$, so that
$\# \mathcal{C}_G^{(a,b)} = 2$ and one concludes  $\du_G(a,b) \geq  4$.  \qed

%\bigskip

%\bigskip
%??? When $\beta\in \F_4$ (i.e., $\beta^2+\beta+1=0$), then using previous propositions, we can show that
%   $\du_G=6$ if any one of the three quadratic polynomials in $a$ in the equations
%   $(\ref{eqnb1b2}, \ref{eqnb1b3}, \ref{eqnb2b3})$ has a root in $\Fbn$
%   (equivalently, if $\tr(\tfrac1\beta)=0$ or $\tr(\tfrac\beta\alpha)=0$), and $\du_G=4$ if not. ???

%\bigskip

%\bigskip\bigskip
% ??? We must add similarity with class $\mathcal{C}$: That is...
% $b_i=b_j$ implies $G(p_i)+G(b+p_i)=G(p_j)+G(b+p_j)$ has a solution $b$. Thus
% $$
% G(p_i)+G(p_j)=G(x)+G(y), \quad p_i+p_j=x+y
% $$
% has a solution $\{x,y\} \in \mathcal{C}$. Therefore the polynomial $X^2+b'X+\tfrac{b'}{\alpha^2 a'}=0$
% has a solution in $\Fbn$ with $a'=G(p_i)+G(p_j), b'=p_i+p_j$. See Proposition \ref{propaclass}.
% ???

%\bigskip

\section{Boomerang Uniformity of $G$}\label{boomerangsec}

Our goal in this section is to classify all $G(x)=[0,1,\beta, x]$ having the boomerang uniformity $6$, which
 is the least possible case for a permutation of Carlitz rank three. We first state the following
  main theorem.

\begin{theorem}\label{butheorem}
Let $\beta \in \Fbn\setminus \F_4$ with $\alpha=\beta+1$ and let
 $G(x)=[0,1,\beta,x]=((x^{2^n-2}+\beta)^{2^n-2}+1)^{2^n-2}$
 be a permutation of Carlitz rank three.
 Define the following polynomials in $\Fbn[Z]$.
 \begin{align*}
    h_{12}(Z)&=Z^4+Z^3+Z^2+\frac{\beta^4}{\alpha^4}, \\
    h_{13}(Z)&=Z^4+\frac{\beta}{\alpha}Z^3+\frac{\beta}{\alpha^2}Z^2+\frac{\beta^2}{\alpha^4}, \\
    g_1(Z)&=Z^6+Z^5+Z^3+Z+1+\frac{\beta}{\alpha}(Z^4+Z^2), \\
    g_2(Z)&=Z^6+Z^5+Z^4+Z^3+\beta Z^2+\beta Z+1, \\
     \phi(Z) &= Z^2+\frac{\beta}{\alpha}Z+\frac{\beta}{\alpha^2}
 \end{align*}
Then $\bu_G=6$ if and only if none of the above five polynomials has a root in $\Fbn$.
\end{theorem}
\begin{remark}
The proof of the above theorem will be given through a series of steps by eliminating possible cases
 of $(a,c)\in \Fbn\times\Fbn$ where $\bu_G(a,c)\geq 8$ happen. In fact, if one of the above polynomials
 has a root in $\Fbn$, then the root is in the list of $a$ or $c$ with the property $\bu_G(a,c)\geq 8$.
 That is why we stick to the expression of $\phi(Z)$ even if we may replace it by simpler form $Z^2+Z+\tfrac1\beta$
 or even by the condition on $\tr(\tfrac1\beta)$.
\end{remark}
\begin{remark}
As in the case of differential uniformity,
 our method determines,  in polynomial time, whether a given $G$ has boomerang uniformity $6$.
 That is, letting $g(Z)$ be the product of all the above five polynomials,
  $\gcd(g(Z), Z^{2^n}-Z)$ is computed  using repeated squaring
  $Z^{2^i} \pmod{g}$ in $O(n^3)$ bit operations.
\end{remark}
\begin{remark}
To the authors' knowledge,  on $\Fbn$ with even $n$, nobody so far found a permutation which is a modification of the
inverse function having the boomerang uniformity 4. Therefore our result on boomerang uniformity 6 is the lowest in
this class of permutations.
\end{remark}

To give a proof of Theorem \ref{butheorem}, we will repeatedly use the following definition of boomerang uniformity,
which is in fact obtained in \cite{LQSL19} (Theorem 2.3)
 as an equivalent statement of the original boomerang uniformity. That is, $\bu_G(a,c)$ is the number
 of $(a,c)\in \Fbn\times \Fbn$ satisfying the following $\bu$ equation,
 \begin{align}\label{buequation}
 G(x)+G(y)=a=G(x+c)+G(y+c)
 \end{align}
An alternative formula for $\bu_G(a,c)$ is also obtained in \cite{BC18} (Prop. 3), and the boomerang formulaes in
\cite{BC18} and \cite{LQSL19} imply  that
  any solution $\{x,y\}=\{u,v\}$ of the above $\bu$ equation falls in one of the two cases
 depending on whether $u+v\neq c$ or $u+v=c$.

\begin{enumerate}
\item If $u+v\neq c$, then $\{u,v\}$ generates $2$ solutions $\{u,v\},\{u+c,v+c\}$
           of unordered pairs for the $\bu$ equation.
             That is, $(x,y)=(u,v), (v,u), (u+c,v+c), (v+c, u+c)$ are $4$ solutions of the $\bu$ equation.
 \item If $u+v= c$, then $\{u,v\}$ in fact  is a solution of the $\du$ equation $G(x)+G(y)=a, \,\, x+y=c$ which gives only
                  $2$ solutions of ordered pair $(u,v), (v,u)$ for the $\bu$ equation
              $G(x)+G(y)=a=G(x+c)+G(y+c)$, because the last equation $a=G(x+c)+G(y+c)$ becomes redundant  when $(x,y)=(u,v)$.
      However when $\du_G(a,c)\geq 4$ and if $\{u,v\}, \{u',v'\}$ are two (unordered pairs of) solutions of the $\du$
      equation $G(x)+G(y)=a, \,\, x+y=c$, then we can also say that $\{u,v\}, \{u',v'\}$ are two solutions of
         $G(x)+G(y)=a=G(x+c)+G(y+c)$ arising from $\du_G(a,c)\geq 4$.
\end{enumerate}

\noindent Therefore, when $\bu_G(a,c)\geq 4$, we may say $\{u,v\}, \{u',v'\}$ are two solutions of the $\bu$ equation
         $G(x)+G(y)=a=G(x+c)+G(y+c)$, and we may assume $\{u',v'\}=\{u+c,v+c\}$ in the first case where $u+v\neq c$.

         Also, we can classify the two solutions $\{u,v\}, \{u',v'\}$ depending on whether
         they are of $\ca,$ $\cb$ or $\cc$. In other words,  when we say the $\bu$ equation
              $G(u)+G(v)=a=G(u+c)+G(v+c)$ of type $\{\mathcal{B}, \mathcal{C}\}$, it means that
                $\{u,v\} \in\mathcal{B}$ and  $\{u+c,v+c\} \in\mathcal{C}$. Since $\{u,v\}\in\mathcal{B}$
                means only one of $u,v$ are in $\pole$, when we say the $\bu$ equation is of type $\{\mathcal{B}(i),
                \mathcal{C}\}$, then it means that $\{u,v\} \in \mathcal{B}$ with $u=p_i$  and $\{u+c,v+c\} \in\mathcal{C}$.
                 Similarly, when we say the $\bu$ equation
              $G(u)+G(v)=a=G(u+c)+G(v+c)$ of type $\{\mathcal{A}(i,j), \mathcal{C}\}$, it means that
                $\{u,v\} \in\mathcal{A}$ with $\{u,v\}=\{p_i, p_j\}$ and  $\{u+c,v+c\} \in\mathcal{C}$.

\medskip
 Similarly as in the analysis of the differential uniformity, we will classify possible solutions
  $\{u,v\}$ of the $\bu$ equation in \eqref{buequation} depending on whether
  $\{u,v\}$ is of $\ca, \cb$ or $\cc$. When $\bu_G(a,c)\geq 8$ at some point $(a,c)$, then there are
  $4$ different pairs of solutions $\{u_i,v_i\}$ $(1\leq i\leq 4)$ and we will give a detailed analysis
  for each possible case. Fortunately it will turn out that there are only three possible
  cases
$$
\{\mathcal{A},\mathcal{C}\}+\{\mathcal{B},\mathcal{C}\},\quad
 \{\mathcal{B},\mathcal{C}\}+\{\mathcal{B},\mathcal{C}\}\quad \mbox{and}\,\,\,
 \{\mathcal{B},\mathcal{B}\}+\{\mathcal{B},\mathcal{C}\},
$$
whose meaning will be explained soon.

\subsection{$\{\mathcal{A},\mathcal{C}\}+\{\mathcal{B},\mathcal{C}\}$}\label{sectionacbc}

\begin{lemma}\label{acbcclasslemma}
 Suppose $\beta\not\in \F_4$ and suppose  that one has
 the following relations satisfied by $4$ different pairs of $\{u_1, v_1\},\{u_2, v_2\},\{x_1, y_1\},\{x_2, y_2\},$
 \begin{align*}
  G(u_1)+G(v_1) = a &=G(u_2)+G(v_2), \quad u_1+v_1 = u_2+v_2 \\
  G(x_1)+G(y_1) = a &=G(x_2)+G(y_2), \quad x_1+y_1 = x_2+y_2
 \end{align*}
 such that $\{u_1,v_1\}$ is of $\ca$. Then $a\in \{1, \tfrac{\beta}{\alpha}, \tfrac{1}{\alpha}\}$ and one has the followings.
 \begin{enumerate}[$(a)$]
 \item $\{u_2,v_2\}$ is of $\cc$. That is, $u_2,v_2 \not\in \pole$.
 \item Among two pairs $\{x_1,y_1\}, \{x_2,y_2\}$, one is of $\cb$ $($say $\{x_1,y_1\}$$)$
   such that $\{u_1,v_1,x_1,y_1\}=\{\tfrac1\alpha,\tfrac1\beta,0,\tfrac{\beta}{\alpha^2}\}$
    and the other is of $\cc$.
 \end{enumerate}
\end{lemma}
\begin{proof} $a\in \{1, \tfrac{\beta}{\alpha}, \tfrac{1}{\alpha}\}$ is obvious since $a=G(p_i)+G(p_j)$
with $\{p_1,p_2,p_3\}=\{\tfrac{1}{\alpha}, \tfrac{1}{\beta},0\}$.

\noindent $(a)$. This comes from Lemma \ref{aclasslemma}-$(a), (c)$.

\noindent $(b)$. This comes from Lemma \ref{aclasslemma}-$(a), (b)$, where $\{x_1,y_1\} \in \mathcal{B}$
 is uniquely determined as $0=G(u_1)+G(v_1)+G(x_1)+G(y_1)=G(p_1)+G(p_2)+G(p_3)+G(\tfrac{\beta}{\alpha^2})$.
\end{proof}

\noindent From Lemma \ref{aclasslemma} and Lemma \ref{acbcclasslemma}, we conclude that, assuming $\beta \not\in \F_4$,
if one has $\bu_G(a,c)\geq 8$ such that the $\bu$ equation $G(x)+G(y)=a=G(x+c)+G(y+c)$ has a solution $\{p_i,p_j\}\in
\mathcal{A}$, then we have a combination two $\bu$ equations
 of type $\{\mathcal{A},\mathcal{C}\}+\{\mathcal{B},\mathcal{C}\}$. That is, one has
 \begin{align*}
 G(p_i)+G(p_j)=a &=G(u)+G(v), \quad p_i+p_j = u+v \quad \mbox{with}\,\,\, \{p_i,p_j\}\in \mathcal{A},\,\, \{u,v\}\in \mathcal{C}\\
 G(p_k)+G(\tfrac{\beta}{\alpha^2})= a &=G(u')+G(v'), \quad p_k+\tfrac{\beta}{\alpha^2} = u'+v'
                              \quad \mbox{with}\,\,\, \{p_k,\tfrac{\beta}{\alpha^2}\}\in \mathcal{B},\,\, \{u',v'\}\in \mathcal{C}
 \end{align*}
 where $\{p_i,p_j,p_k\}=\pole$ such that one of the following holds;

 \begin{enumerate}[(i)]
 \item $\{u+c,v+c\}=\{p_i, p_j\}$ and $\{u'+c,v'+c\}=\{p_k, \tfrac{\beta}{\alpha^2}\} :$ In this case,
    one has
    $$c \in \{p_i+u, p_i+v\} \cap \{p_k+u',p_k+v'\}$$
     because $\{p_i+u, p_i+v\}=\{c, p_i+p_j+c\}$ and $\{p_k+u',p_k+v'\}=\{c, p_k+ \tfrac{\beta}{\alpha^2}+c\}$.
 \item $\{u+c,v+c\}=\{p_i, p_j\}$ with $c=p_k+\tfrac{\beta}{\alpha^2} :$
  In this case, one has
   $$ p_k+\tfrac{\beta}{\alpha^2}\in \{p_i+u, p_i+v\}$$
   because $\{p_i+u, p_i+v\}=\{ c, p_i+p_j+c\}$ with $c=p_k+\tfrac{\beta}{\alpha^2}$.
 \item $\{u'+c,v'+c\}=\{p_k, \tfrac{\beta}{\alpha^2}\}$ with  $c=p_i+p_j :$
  In this case, one has
  $$p_i+p_j \in \{p_k+u', p_k+v'\}$$
   because $\{p_k+u', p_k+v'\}=\{c,
  p_k+\tfrac{\beta}{\alpha^2}+c\}$ with $c=p_i+p_j$.
 \end{enumerate}

\noindent To summarize, $\bu_G(a,c)\geq 8$ for the case the case {\rm (i)} comes from $\du_G(a,b)\geq 4$ with $b=u+v$
and
 $\du_G(a,b')\geq 4$ with $b'=u'+v'$, and the case {\rm (ii)} comes from $\du_G(a,b)\geq 4$ with $b=u+v$ and
 $\du_G(a,c)\geq 4$ with $c=u'+v'$, and the case {\rm (iii)} comes from $\du_G(a,c)\geq 4$ with $c=u+v$ and
 $\du_G(a,b')\geq 4$ with $b'=u'+v'$. Also note that one always have $u+v\neq u'+v'$ because $p_i+p_j\neq
 p_k+\tfrac\beta{\alpha^2}$ when $\beta\not\in \F_4$.

\bigskip

Now will discuss three possible combinations of $\{\mathcal{A},\mathcal{C}\}+\{\mathcal{B},\mathcal{C}\}$.

\bigskip

\noindent
 $ \begin{aligned}
\{\mathcal{A(1,2)},&\mathcal{C}\}+\{\mathcal{B(3)},\mathcal{C}\}:
                 (a,b,b')=(1,\tfrac{1}{\alpha\beta}, \tfrac{\beta}{\alpha^2}) \\
  & \{\mathcal{A(1,2)},\mathcal{C}\} :
G(\tfrac{1}{\alpha})+G(\tfrac{1}{\beta})=1=G(u)+G(v),
  \,\,\, u+v=\tfrac{1}{\alpha\beta} \\
  &\{\mathcal{B(3)},\mathcal{C}\} :
  G(0)+G(\tfrac{\beta}{\alpha^2})=1=G(u')+G(v'),
  \,\,\, u'+v'=\tfrac{\beta}{\alpha^2}
\end{aligned}
$

\medskip
\noindent where one of the following three cases holds;

\begin{enumerate}[(i)]
\item $\{u+c,v+c\}=\{\tfrac1\alpha, \tfrac1\beta\}$ and $\{u'+c,v'+c\}=\{0, \tfrac{\beta}{\alpha^2}\}$ \\
    $\,\,\,\Rightarrow\,\,\, c \in \{\tfrac1\alpha+u, \tfrac1\alpha+v\} \cap \{u',v'\}$
    $\,\,\,\Rightarrow\,\,\, \res(h_1^{(1,\frac{1}{\alpha\beta})},h_3^{(1,\frac{\beta}{\alpha^2})})=0$
 \item $\{u+c,v+c\}=\{\tfrac1\alpha, \tfrac1\beta\}$ with $c=\tfrac{\beta}{\alpha^2}$
   $\,\,\,\Rightarrow\,\,\,  \tfrac{\beta}{\alpha^2}\in \{\tfrac1\alpha +u, \tfrac1\alpha +v\}$
    $\,\,\,\Rightarrow\,\,\, h_1^{(1,\frac{1}{\alpha\beta})}(\tfrac{\beta}{\alpha^2})=0 $
 \item $\{u'+c,v'+c\}=\{0, \tfrac{\beta}{\alpha^2}\}$ with  $c=\tfrac{1}{\alpha\beta}$
 $\,\,\,\Rightarrow\,\,\, \tfrac{1}{\alpha\beta} \in \{u', v'\}$
 $\,\,\,\Rightarrow\,\,\, h_3^{(1,\frac{\beta}{\alpha^2})}(\tfrac{1}{\alpha\beta})=0$
\end{enumerate}
where $h_1$ and $h_3$ defined in Section \ref{sectioncarlitz3} can be written as
 \begin{align*}
 h_1^{(1,\frac{1}{\alpha\beta})}(X)=\hh^{(1,\frac{1}{\alpha\beta})}(X+\tfrac{1}{\alpha})
                     &=(X+\tfrac1\alpha+u)(X+\tfrac1\alpha+v)=X^2+\tfrac{1}{\alpha\beta}X+\tfrac{1}{\alpha^3\beta}, \\
h_3^{(1,\frac{\beta}{\alpha^2})}(X)=\hh^{(1,\frac{\beta}{\alpha^2})}(X)
           &=(X+u')(X+v')=X^2+\tfrac{\beta}{\alpha^2}X+\tfrac{1}{\alpha^4},
 \end{align*}
 and $\res$ means the resultant of two polynomials. When the given two polynomials $f(X)=f_0X^2+f_1X+f_2$ and
$g(X)=g_0X^2+g_1X+g_2$ are quadratic, the resultant $\res(f,g)$ is defined as
\begin{align}\label{resultant}
\res(f,g) =
\begin{vmatrix}
f_0 & 0 & g_0 & 0 \\
f_1 & f_0 & g_1 & g_0 \\
f_2 & f_1  & g_2 & g_1  \\
0 &  f_2 & 0 & g_2
\end{vmatrix}
\end{align}
the determinant of the $4$ by $4$ matrix.
 Direct computations using the above two quadratic polynomials show
 \begin{align*}
 \res(h_1^{(1,\frac{1}{\alpha\beta})},h_3^{(1,\frac{\beta}{\alpha^2})})&=\tfrac{1}{\alpha^8\beta}(\beta^3+\beta^2+1), \\
 h_1^{(1,\frac{1}{\alpha\beta})}(\tfrac{\beta}{\alpha^2})&=\tfrac{1}{\alpha^4\beta}(\beta^3+\beta^2+1),\\
 h_3^{(1,\frac{\beta}{\alpha^2})}(\tfrac{1}{\alpha\beta})&= \tfrac{1}{\alpha^4\beta^2}(\beta^3+\beta^2+1)
 \end{align*}
Therefore, if one has $\bu_G(a,c)\geq 8$ with the combination of
$\{\mathcal{A(1,2)},\mathcal{C}\}+\{\mathcal{B(3)},\mathcal{C}\}$, then one has $\beta^3+\beta^2+1=0$.

\bigskip

\noindent
 $ \begin{aligned}
\{\mathcal{A(1,3)},&\mathcal{C}\}+\{\mathcal{B(2)},\mathcal{C}\}:
                 (a,b,b')=(\tfrac\beta\alpha,\tfrac{1}{\alpha}, \tfrac{1}{\alpha^2\beta}) \\
  & \{\mathcal{A(1,3)},\mathcal{C}\} :
G(\tfrac{1}{\alpha})+G(0)=\tfrac\beta\alpha=G(u)+G(v),
  \,\,\, u+v=\tfrac{1}{\alpha} \\
  &\{\mathcal{B(2)},\mathcal{C}\} :
  G(\tfrac1\beta)+G(\tfrac{\beta}{\alpha^2})=\tfrac\beta\alpha=G(u')+G(v'),
  \,\,\, u'+v'=\tfrac{1}{\alpha^2\beta}
\end{aligned}
$

\medskip
\noindent where one of the following three cases holds;

\begin{enumerate}[(i)]
\item $\{u+c,v+c\}=\{\tfrac1\alpha, 0\}$ and $\{u'+c,v'+c\}=\{\tfrac1\beta, \tfrac{\beta}{\alpha^2}\}$ \\
    $\,\,\,\Rightarrow\,\,\, c \in \{\tfrac1\alpha+u, \tfrac1\alpha+v\} \cap \{\tfrac1\beta+u',\tfrac1\beta+v'\}$
    $\,\,\,\Rightarrow\,\,\, \res(h_1^{(\frac{\beta}{\alpha},\frac{1}{\alpha})},h_2^{(\frac{\beta}{\alpha},\frac{1}{\alpha^2\beta})})=0$
 \item $\{u+c,v+c\}=\{\tfrac1\alpha, 0\}$ with $c=\tfrac{1}{\alpha^2\beta}$
   $\,\,\,\Rightarrow\,\,\,  \tfrac{1}{\alpha^2\beta}\in \{\tfrac1\alpha +u, \tfrac1\alpha +v\}$
    $\,\,\,\Rightarrow\,\,\, h_1^{(\frac{\beta}{\alpha},\frac{1}{\alpha})}(\tfrac{1}{\alpha^2\beta})=0 $
 \item $\{u'+c,v'+c\}=\{\tfrac1\beta, \tfrac{\beta}{\alpha^2}\}$ with  $c=\tfrac{1}{\alpha}$
 $\,\,\,\Rightarrow\,\,\, \tfrac{1}{\alpha} \in \{\tfrac1\beta+u', \tfrac1\beta+v'\}$
 $\,\,\,\Rightarrow\,\,\, h_2^{(\frac{\beta}{\alpha},\frac{1}{\alpha^2\beta})}(\tfrac{1}{\alpha})=0$
\end{enumerate}
where
 \begin{align*}
h_1^{(\frac{\beta}{\alpha},\frac{1}{\alpha})}(X)=\hh^{(\frac{\beta}{\alpha},\frac{1}{\alpha})}(X+\tfrac{1}{\alpha})
                     &=(X+\tfrac1\alpha+u)(X+\tfrac1\alpha+v)=X^2+\tfrac{1}{\alpha}X+\tfrac{1}{\alpha^2\beta}, \\
h_2^{(\frac{\beta}{\alpha},\frac{1}{\alpha^2\beta})}(X)=\hh^{(\frac{\beta}{\alpha},\frac{1}{\alpha^2\beta})}(X+\tfrac1\beta)
           &=(X+\tfrac1\beta+u')(X+\tfrac1\beta+v')=X^2+\tfrac{1}{\alpha^2\beta}X+\tfrac{1}{\alpha^2\beta^2}
 \end{align*}
Direct computations using the above two polynomials show
 \begin{align*}
 \res(h_1^{(\frac{\beta}{\alpha},\frac{1}{\alpha})},h_2^{(\frac{\beta}{\alpha},\frac{1}{\alpha^2\beta})})&=\tfrac{1}{\alpha^6\beta^4}(\beta^3+\beta^2+1), \\
 h_1^{(\frac{\beta}{\alpha},\frac{1}{\alpha})}(\tfrac{1}{\alpha^2\beta})&=\tfrac{1}{\alpha^4\beta^2}(\beta^3+\beta^2+1),\\
 h_2^{(\frac{\beta}{\alpha},\frac{1}{\alpha^2\beta})}(\tfrac{1}{\alpha})&= \tfrac{1}{\alpha^3\beta^2}(\beta^3+\beta^2+1)
 \end{align*}
Therefore, if one has $\bu_G(a,c)\geq 8$ with the combination of
$\{\mathcal{A(1,3)},\mathcal{C}\}+\{\mathcal{B(2)},\mathcal{C}\}$, then one has $\beta^3+\beta^2+1=0$.

\bigskip

\noindent
 $ \begin{aligned}
\{\mathcal{A(2,3)},&\mathcal{C}\}+\{\mathcal{B(1)},\mathcal{C}\}:
                 (a,b,b')=(\tfrac1\alpha,\tfrac{1}{\beta}, \tfrac{1}{\alpha^2}) \\
  & \{\mathcal{A(2,3)},\mathcal{C}\} :
G(\tfrac{1}{\beta})+G(0)=\tfrac1\alpha=G(u)+G(v),
  \,\,\, u+v=\tfrac{1}{\beta} \\
  &\{\mathcal{B(1)},\mathcal{C}\} :
  G(\tfrac1\alpha)+G(\tfrac{\beta}{\alpha^2})=\tfrac1\alpha=G(u')+G(v'),
  \,\,\, u'+v'=\tfrac{1}{\alpha^2}
\end{aligned}
$

\medskip
\noindent where one of the following three cases holds;

\begin{enumerate}[(i)]
\item $\{u+c,v+c\}=\{\tfrac1\beta, 0\}$ and $\{u'+c,v'+c\}=\{\tfrac1\alpha, \tfrac{\beta}{\alpha^2}\}$ \\
    $\,\,\,\Rightarrow\,\,\, c \in \{\tfrac1\beta+u, \tfrac1\beta+v\} \cap \{\tfrac1\alpha+u',\tfrac1\alpha+v'\}$
    $\,\,\,\Rightarrow\,\,\, \res(h_2^{(\frac1\alpha,\frac{1}{\beta})},h_1^{(\frac1\alpha,\frac{1}{\alpha^2})})=0$
 \item $\{u+c,v+c\}=\{\tfrac1\beta, 0\}$ with $c=\tfrac{1}{\alpha^2}$
   $\,\,\,\Rightarrow\,\,\,  \tfrac{1}{\alpha^2}\in \{\tfrac1\beta +u, \tfrac1\beta +v\}$
    $\,\,\,\Rightarrow\,\,\, h_2^{(\frac1\alpha,\frac{1}{\beta})}(\tfrac{1}{\alpha^2})=0 $
 \item $\{u'+c,v'+c\}=\{\tfrac1\alpha, \tfrac{\beta}{\alpha^2}\}$ with  $c=\tfrac{1}{\beta}$
 $\,\,\,\Rightarrow\,\,\, \tfrac{1}{\beta} \in \{\tfrac1\alpha+u',\tfrac1\alpha+ v'\}$
 $\,\,\,\Rightarrow\,\,\, h_1^{(\frac1\alpha,\frac{1}{\alpha^2})}(\tfrac{1}{\beta})=0$
\end{enumerate}
where
 \begin{align*}
 h_2^{(\frac1\alpha,\frac{1}{\beta})}(X)=\hh^{(\frac1\alpha,\frac{1}{\beta})}(X+\tfrac{1}{\beta})
                     &=(X+\tfrac1\beta+u)(X+\tfrac1\beta+v)=X^2+\tfrac{1}{\beta}X+\tfrac{1}{\alpha^2}, \\
h_1^{(\frac1\alpha,\frac{1}{\alpha^2})}(X)=\hh^{(\frac1\alpha,\frac{1}{\alpha^2})}(X+\tfrac1\alpha)
           &=(X+\tfrac1\alpha+u')(X+\tfrac1\alpha+v')=X^2+\tfrac{1}{\alpha^2}X+\tfrac{1}{\alpha^3}
 \end{align*}
Direct computations using the above two polynomials show
 \begin{align*}
 \res(h_2^{(\frac1\alpha,\frac{1}{\beta})},h_1^{(\frac1\alpha,\frac{1}{\alpha^2})})&=\tfrac{1}{\alpha^5\beta^2}(\beta^3+\beta^2+1), \\
h_2^{(\frac1\alpha,\frac{1}{\beta})}(\tfrac{1}{\alpha^2})&=\tfrac{1}{\alpha^4\beta}(\beta^3+\beta^2+1),\\
 h_1^{(\frac1\alpha,\frac{1}{\alpha^2})}(\tfrac{1}{\beta})&= \tfrac{1}{\alpha^3\beta^2}(\beta^3+\beta^2+1)
 \end{align*}
Therefore, if one has $\bu_G(a,c)\geq 8$ with the combination of
$\{\mathcal{A(2,3)},\mathcal{C}\}+\{\mathcal{B(1)},\mathcal{C}\}$, then one has $\beta^3+\beta^2+1=0$.

\bigskip
\noindent From the above observations regarding $\{\mathcal{A},\mathcal{C}\}+\{\mathcal{B},\mathcal{C}\}$, one obtains
 the following result.

\begin{prop}
Let $\beta \not\in \F_4$ and let $G(x)=((x^{2^n-2}+\beta)^{2^n-2}+1)^{2^n-2}=[0,1,\beta,x]$ be a permutation with
Carlitz rank three. Suppose that $\bu_G(a,c)\geq 8$ for some $(a,c)\in \Fbn\times\Fbn$ such that the $\bu$ equation
 $G(x)+G(y)=a=G(x+c)+G(y+c)$ has a solution $\{x,y\}=\{p_i,p_j\} \in \mathcal{A}$. Then one has $\beta^3+\beta^2+1=0$.
\end{prop}
%for some $a\in \{1, \tfrac\beta\alpha, \tfrac1\alpha\}$ and $c\in \Fbn$ with a combination of
%$\{\mathcal{A},\mathcal{C}\}+\{\mathcal{B},\mathcal{C}\}$. Then one has $\beta^3+\beta^2+1=0$.

\begin{remark}\label{betaremark}
Since the case $\beta^3+\beta^2+1=0$ is uninteresting in the sense that
 $\du_G=8$ which we already showed in Section $\ref{differentialsec}$ and
  since the above proposition settled the case
where the $\bu_G(a,c)\geq 8$ with one of the solutions of the
 $\bu$ equation  in $\ca$, from now on,
we always assume that $\beta^3+\beta^2+1\neq 0$ and  no $\ca$ solution exist when we discuss the $\bu$ equation
$G(x)+G(y)=a=G(x+c)+G(y+c)$.
\end{remark}

\subsection{$\{\mathcal{B},\mathcal{C}\}+\{\mathcal{B},\mathcal{C}\}$}

Since the $\bu$ equation with a solution $\{u,v\}\in \mathcal{A}$ is settled in previous section, and since
 $\{\mathcal{C},\mathcal{C}\}$ is not possible by Proposition \ref{abalpham} or by Corollary \ref{abalphamcoro}, the only remaining types are
$\{\mathcal{B},\mathcal{C}\}+\{\mathcal{B},\mathcal{C}\}$ and
$\{\mathcal{B},\mathcal{B}\}+\{\mathcal{B},\mathcal{C}\}$. We will first discuss the type
$\{\mathcal{B},\mathcal{C}\}+\{\mathcal{B},\mathcal{C}\}$.
 Similarly as in the previous section, we divide the cases into Case {\rm (i)} and Case {\rm (ii), (iii)}.
  By the symmetry of the type $\{\mathcal{B},\mathcal{C}\}+\{\mathcal{B},\mathcal{C}\}$, the case
  {\rm (ii), (iii)} can be discussed together.

\subsubsection{$\{\mathcal{B},\mathcal{C}\}+\{\mathcal{B},\mathcal{C}\} :$ Case {\rm (i)}}\label{sectionh12x}

In this subsection, we will show that, if $\bu_G(a,c)\geq 8$ happens, then either $h_{12}(Z)$ or $h_{13}(Z)$
 introduced in Theorem \ref{butheorem} has a root in $\Fbn$.

 \bigskip

 $ \begin{aligned}
   \{\mathcal{B(1)},&\mathcal{C}\} + \{\mathcal{B(2)},\mathcal{C}\}:   \\
  & \{\mathcal{B(1)},\mathcal{C}\} :
G(\tfrac{1}{\alpha})+G(b_1+\tfrac{1}{\alpha})=a=G(u_1)+G(v_1),
  \,\,\, \{u_1+c,v_1+c \}=\{\tfrac{1}{\alpha}, b_1+\tfrac{1}{\alpha}\} \\
  &\{\mathcal{B(2)},\mathcal{C}\} :
  G(\tfrac{1}{\beta})+G(b_2+\tfrac{1}{\beta})=a=G(u_2)+G(v_2),
  \,\,\, \{u_2+c,v_2+c \}=\{\tfrac{1}{\beta}, b_2+\tfrac{1}{\beta}\}
\end{aligned}
$

\medskip
\noindent In this case, one has
 $c\in \{u_1+\tfrac{1}{\alpha}, v_1+\tfrac{1}{\alpha}\} \cap \{u_2+\tfrac{1}{\beta}, v_2+\tfrac{1}{\beta}\} $.
 Therefore $c$ is a common root of two polynomials
 \begin{align*}
  h_1^{(a,b_1)}(X) &=\hh^{(a,b_1)}(X+\tfrac{1}{\alpha})=(X+u_1+\tfrac{1}{\alpha})(X+v_1+\tfrac{1}{\alpha})
                        =X^2+b_1X+\tfrac{b_1}{\alpha^2 a}  \\
  h_2^{(a,b_2)}(X) &=\hh^{(a,b_2)}(X+\tfrac{1}{\beta})=(X+u_2+\tfrac{1}{\beta})(X+v_2+\tfrac{1}{\beta})
                        =X^2+b_2X+\tfrac{b_2}{\alpha\beta}+\tfrac{b_2}{\alpha^2a}+\tfrac{1}{\alpha^2\beta^2}
 \end{align*}
Since $b_1\neq b_2$ and two quadratic polynomials $ h_1^{(a,b_1)}(X),  h_2^{(a,b_2)}(X)$ have a unique common root $c$,
  the resultant of two polynomials must be zero. Therefore, after a routine computation
  with $b_1=\tfrac{1}{\alpha(\alpha a+\beta)}$ and $b_2=\tfrac{a+1}{\beta(\alpha a+1)}$
  (See the expression of $b_1,b_2,b_3$ in the equations \eqref{eqnb1},\eqref{eqnb2} and \eqref{eqnb3}), one has
\begin{align*}
0=\res(h_1^{(a,b_1)},h_2^{(a,b_2)})=\left\{\frac{1}{\alpha\beta a(\alpha a+\beta)(\alpha a+1)}\right\}^2
                                     \left(a^4+a^3+a^2+\frac{\beta^4}{\alpha^4}\right)
\end{align*}
Since $a(\alpha a+\beta)(\alpha a+1)\neq 0$ because  $b_1+\tfrac{1}{\alpha}, b_2+\tfrac{1}{\beta}\notin \pole$, one
must have
$$
0=a^4+a^3+a^2+\frac{\beta^4}{\alpha^4}
$$
Therefore by defining
\begin{align}
h_{12}(Z)=Z^4+Z^3+Z^2+\frac{\beta^4}{\alpha^4} \,\, \in \Fbn[Z], \label{h12x}
\end{align}
one concludes that, if there are $a,c$ with $\bu(a,c)\geq 8$ which come from the combination
 $\{\mathcal{B(1)},\mathcal{C}\} + \{\mathcal{B(2)},\mathcal{C}\}$, then one has $h_{12}(a)=0$.
 Conversely, the existence of a root $z\in \Fbn$ satisfying $h_{12}(z)=0$
 guarantees the existence of $a,c \in \Fbn$ having $\bu_G(a,c)\geq 8$.
 Note that one can explicitly construct such point $(a,c)$ from a root of $h_{12}(Z)=0$. Namely,
 if $h_{12}(z)=0$ for some $z\in \Fbn$, then the common solution $c$ satisfying $h_1^{(a,b_1)}(c)=0=h_2^{(a,b_2)}(c)$
 can be expressed as $c=\tfrac{\alpha z+\beta^2}{\alpha^2 z (\alpha^2 z^2+\alpha^2z+\beta^2)}$ so that
 one has $\bu_G(a,c)\geq 8$ with $(a,c)=(z, \tfrac{\alpha z+\beta^2}{\alpha^2 z (\alpha^2 z^2+\alpha^2z+\beta^2)})$.

\bigskip

 $ \begin{aligned}
   \{\mathcal{B(1)},&\mathcal{C}\} + \{\mathcal{B(3)},\mathcal{C}\}:   \\
  & \{\mathcal{B(1)},\mathcal{C}\} :
G(\tfrac{1}{\alpha})+G(b_1+\tfrac{1}{\alpha})=a=G(u_1)+G(v_1),
  \,\,\, \{u_1+c,v_1+c \}=\{\tfrac{1}{\alpha}, b_1+\tfrac{1}{\alpha}\} \\
  &\{\mathcal{B(3)},\mathcal{C}\} :
  G(0)+G(b_3)=a=G(u_3)+G(v_3),
  \,\,\, \{u_3+c,v_3+c \}=\{0, b_3\}
\end{aligned}
$

\medskip
\noindent In this case, one has
 $c\in \{u_1+\tfrac{1}{\alpha}, v_1+\tfrac{1}{\alpha}\} \cap \{u_3, v_3\} $.
 Therefore $c$ is a common root of two polynomials
 \begin{align*}
  h_1^{(a,b_1)}(X) &=\hh^{(a,b_1)}(X+\tfrac{1}{\alpha})=(X+u_1+\tfrac{1}{\alpha})(X+v_1+\tfrac{1}{\alpha})
                        =X^2+b_1X+\tfrac{b_1}{\alpha^2 a}  \\
  h_3^{(a,b_3)}(X) &=\hh^{(a,b_3)}(X)=(X+u_3)(X+v_3)
                        =X^2+b_3X+\tfrac{b_3}{\alpha}+\tfrac{b_3}{\alpha^2a}+\tfrac{1}{\alpha^2}
 \end{align*}
Since the two quadratic polynomials $ h_1^{(a,b_1)}(X),  h_3^{(a,b_3)}(X)$ have a unique common root $c$,
  the resultant of two polynomials must be zero. Therefore, after a routine computation
  with $b_1=\tfrac{1}{\alpha(\alpha a+\beta)}$ and $b_3=\tfrac{\alpha a+1}{\alpha^2 a}$, one has
\begin{align*}
0=\res(h_1^{(a,b_1)},h_3^{(a,b_3)})=\left\{\frac{1}{\alpha^2 a^2(\alpha a+\beta)}\right\}^2
                                     \left(a^4+\frac{\beta}{\alpha}a^3+\frac{\beta}{\alpha^2}a^2+\frac{\beta^2}{\alpha^4}\right)
\end{align*}
which implies
$$
0=a^4+\frac{\beta}{\alpha}a^3+\frac{\beta}{\alpha^2}a^2+\frac{\beta^2}{\alpha^4}
$$
By defining
\begin{align}
h_{13}(Z)=Z^4+\frac{\beta}{\alpha}Z^3+\frac{\beta}{\alpha^2}Z^2+\frac{\beta^2}{\alpha^4} \,\, \in \Fbn[Z],
\end{align}
one concludes that, if there are $a,c$ with $\bu_G(a,c)\geq 8$  coming from the combination
 $\{\mathcal{B(1)},\mathcal{C}\} + \{\mathcal{B(3)},\mathcal{C}\}$, then one has $h_{13}(a)=0$.
 Conversely, the existence of a root $z\in \Fbn$ satisfying $h_{13}(z)=0$
 guarantees the existence of $a,c \in \Fbn$ having $\bu_G(a,c)\geq 8$.
 Note that one can explicitly construct such point $(a,c)$ from a root of $h_{12}(Z)=0$. Namely,
 if $h_{13}(z)=0$ for some $z\in \Fbn$, then the common solution $c$ satisfying $h_1^{(a,b_1)}(c)=0=h_3^{(a,b_3)}(c)$
 can be expressed as $c=\tfrac{\beta}{\alpha^2 z (\alpha^2 z^2+\alpha\beta z+\beta)}$ so that
 one has $\bu_G(a,c)\geq 8$ with $(a,c)=(z, \tfrac{\beta}{\alpha^2 z (\alpha^2 z^2+\alpha\beta z+\beta)})$.

\bigskip

 $ \begin{aligned}
   \{\mathcal{B(2)},&\mathcal{C}\} + \{\mathcal{B(3)},\mathcal{C}\}:   \\
  & \{\mathcal{B(2)},\mathcal{C}\} :
G(\tfrac{1}{\beta})+G(b_2+\tfrac{1}{\beta})=a=G(u_2)+G(v_2),
  \,\,\, \{u_2+c,v_2+c \}=\{\tfrac{1}{\beta}, b_2+\tfrac{1}{\beta}\} \\
  &\{\mathcal{B(3)},\mathcal{C}\} :
  G(0)+G(b_3)=a=G(u_3)+G(v_3),
  \,\,\, \{u_3+c,v_3+c \}=\{0, b_3\}
\end{aligned}
$

\medskip
\noindent In this case, one has
 $c\in \{u_2+\tfrac{1}{\beta}, v_2+\tfrac{1}{\beta}\} \cap \{u_3, v_3\} $.
 Therefore $c$ is a common root of two polynomials
 \begin{align*}
  h_2^{(a,b_2)}(X) &=\hh^{(a,b_2)}(X+\tfrac{1}{\beta})=(X+u_2+\tfrac{1}{\beta})(X+v_2+\tfrac{1}{\beta})
                        =X^2+b_2X+\tfrac{b_2}{\alpha\beta}+\tfrac{b_2}{\alpha^2a}+\tfrac{1}{\alpha^2\beta^2}  \\
  h_3^{(a,b_3)}(X) &=\hh^{(a,b_3)}(X)=(X+u_3)(X+v_3)
                        =X^2+b_3X+\tfrac{b_3}{\alpha}+\tfrac{b_3}{\alpha^2a}+\tfrac{1}{\alpha^2}
 \end{align*}
Since the two quadratic polynomials $ h_2^{(a,b_2)}(X),  h_3^{(a,b_3)}(X)$ have a unique common root $c$,
  the resultant of two polynomials must be zero. Therefore, after a routine computation
  with $b_2=\tfrac{a+1}{\beta(\alpha a+1)}$ and $b_3=\tfrac{\alpha a+1}{\alpha^2 a}$, one has
\begin{align*}
0=\res(h_2^{(a,b_2)},h_3^{(a,b_3)})=\frac{1}{\alpha^3\beta}\left\{\frac{1}{a^2(\alpha a+1)}\right\}^2
                                     \left(a^4+\frac{1}{\alpha}a^3+\frac{1}{\alpha\beta}a^2+\frac{\beta}{\alpha^5}\right)
\end{align*}
which implies
$$
0=a^4+\frac{1}{\alpha}a^3+\frac{1}{\alpha\beta}a^2+\frac{\beta}{\alpha^5}
$$
By defining
\begin{align}
h_{23}(Z)=Z^4+\frac{1}{\alpha}Z^3+\frac{1}{\alpha\beta}Z^2+\frac{\beta}{\alpha^5} \,\, \in \Fbn[Z],
\end{align}
one concludes that the existence of a root $z\in \Fbn$ satisfying $h_{23}(z)=0$
 guarantees the existence of $a,c \in \Fbn$ having $\bu_G(a,c)\geq 8$.
 However, it should be mentioned that $h_{23}(Z)$ is birationally isomorphic to $h_{12}(Z)$ over $\Fbn$ in the sense that
 $$
 h_{23}(Z)=\frac{1}{\alpha\beta^7}(\alpha Z+\beta)^4h_{12}\left(\frac{\beta^2Z}{\alpha Z+\beta}\right)
 $$
such that the existence of a root of $h_{12}(Z)$ in $\Fbn$ is equivalent to the existence of a root of $h_{23}(Z)$ in
$\Fbn$. Therefore the polynomial $h_{23}(Z)$ is redundant in our analysis of the boomerang uniformity and will not be
used later.

\begin{remark}\label{homo}
More precisely, the two polynomials $h_{12}(Z)$ and $h_{23}(Z)$ are isomorphic via
 linear transformation between homogeneous versions. That is,
  letting
  \begin{align*}
 \tilde{h}_{12}(Z,W) &=Z^4+Z^3W+Z^2W^2+\frac{\beta^4}{\alpha^4}W^4, \\
 \tilde{h}_{23}(Z,W) &=Z^4+\frac{1}{\alpha}Z^3W+\frac{1}{\alpha\beta}Z^2W^2+\frac{\beta}{\alpha^5}W^4,
  \end{align*}
  it holds that
  $$
  \tilde{h}_{23}(Z,W)=\dfrac{\beta}{\alpha}\tilde{h}_{12}(Z,\dfrac{\alpha}{\beta^2}Z+\dfrac{1}{\beta}W)
  $$
 \end{remark}

\subsubsection{$\{\mathcal{B},\mathcal{C}\}+\{\mathcal{B},\mathcal{C}\}' :$ Case {\rm (ii),(iii)}}

In this subsection, we will show that, if $\bu_G(a,c)\geq 8$ happens, then either $g_{1}(Z)$ or $g_{2}(Z)$
 introduced in Theorem \ref{butheorem} has a root in $\Fbn$.

We consider the situation $\bu_G(a,c)\geq 8$ where this happens as a combination of
  $\{\mathcal{B},\mathcal{C}\}+\{\mathcal{B},\mathcal{C}\}' :$ Case {\rm (ii),(iii)}, where
  the prime symbol is used at $\{\mathcal{B({\it j})},\mathcal{C}\}'$ to emphasize that $\du_G(a,c)\geq 4$ happens
  at the second $\bu$ equation. That is,

\bigskip

$ \begin{aligned}
   \{\mathcal{B({\it i})},&\mathcal{C}\} +\{\mathcal{B({\it j})},\mathcal{C}\}' : \\
 & \{\mathcal{B({\it i})},\mathcal{C}\} : G(p_i)+G(b_i+p_i)=a=G(u_i)+G(v_i), \,\,\, \{u_i+c,v_i+c \}=\{p_i, b_i+p_i\} \\
 &\{\mathcal{B({\it j})},\mathcal{C}\}' : G(p_j)+G(p_j+c)=a=G(u_j)+G(v_j), \,\,\, u_j+v_j=c
\end{aligned}
$

\medskip
   \noindent where the two equations $\{p_i, b_i+p_i\}$ and $\{p_j, b_j+p_j\}$
 are of $\cb$ (i.e., $b_i+p_i, c+p_j \notin \pole$), and $\{u_i,v_i\}$ and $\{u_j,v_j\}$
  are of $\cc$ (i.e., $u_i,v_i, u_j, v_j \notin \pole$).

Note that such $c$ satisfying the above system of simultaneous $\bu$ equations, if it exists,  is uniquely determined
as $c=b_j$ satisfying $a=G(p_j)+G(p_j+b_j)$, and the solutions of the second $\bu$ equation account for $\du_G(a,c)\geq
4$. Since $i\neq j$, there are exactly $6$ types of such combination $\{\mathcal{B({\it i})},\mathcal{C}\}
+\{\mathcal{B({\it j})},\mathcal{C}\}'$ for $(i,j)=(1,2),(2,1),$ $(1,3),(3,1),(2,3),(3,2)$.

%\medskip
%{\color{red} Please ??? note that ??? this is the only case where $c\in \pole$ when $\bu_G(a,c)\geq 8$. ??? In
%$\{\mathcal{B},\mathcal{C}\}+\{\mathcal{B},\mathcal{C}\}$ and in
%$\{\mathcal{B},\mathcal{B}\}+\{\mathcal{B},\mathcal{C}\}$, $\bu_G(a,c)\geq 8$ implies $c \not\in \pole$. ???}
%\medskip

For each of the six types of combinations, we will repeatedly use the following bivariate equation
$$
H_a(X,Y)=X^2+XY+\frac{Y}{\alpha^2a}
$$
which was introduced in the equation \eqref{heqn}, and we will derive that $X$ and $Y$ can be parametrized in terms of
 either $a$ of $c$ such that $H_a(X,Y)$ becomes a polynomial of one variable and the existence of $a,c$ with $\bu_G(a,c)\geq 8$
 is equivalent of solvability of the polynomial in $\Fbn$.
In this way, we will use the equation $H_a(X,Y)=0$ repeatedly for each of six types of
$\{\mathcal{B},\mathcal{C}\}+\{\mathcal{B},\mathcal{C}\}'$
 so that we will derive the following corresponding equations as
follows;
\begin{align*}
  \circled{1}\quad \{\mathcal{B(2)},\mathcal{C}\}+\{\mathcal{B(1)},\mathcal{C}\}' & \quad\Rightarrow\quad
 H_a(b_1+\tfrac{1}{\alpha\beta},b_2)=0 \\
 \circled{2}\quad\{\mathcal{B(3)},\mathcal{C}\}+\{\mathcal{B(1)},\mathcal{C}\}' & \quad\Rightarrow\quad
 H_a(b_1+\tfrac{1}{\alpha},b_3)=0 \\
  \circled{3}\quad\{\mathcal{B(2)},\mathcal{C}\}+\{\mathcal{B(3)},\mathcal{C}\}' & \quad\Rightarrow\quad
 H_a(b_3+\tfrac{1}{\alpha\beta},b_2)=0 \\
   \circled{4}\quad\{\mathcal{B(1)},\mathcal{C}\}+\{\mathcal{B(3)},\mathcal{C}\}' & \quad\Rightarrow\quad
 H_a(b_3,b_1)=0 \\
  \circled{5}\quad\{\mathcal{B(1)},\mathcal{C}\}+\{\mathcal{B(2)},\mathcal{C}\}' & \quad\Rightarrow\quad
 H_a(b_2,b_1)=0 \\
  \circled{6}\quad \{\mathcal{B(3)},\mathcal{C}\}+\{\mathcal{B(2)},\mathcal{C}\}' & \quad\Rightarrow\quad
 H_a(b_2+\tfrac{1}{\alpha},b_3)=0
\end{align*}

Now we will discuss the above six types of $\{\mathcal{B},\mathcal{C}\}+\{\mathcal{B},\mathcal{C}\}'$ and will conclude
that one
 only needs to consider solvability of the polynomials $g_1(Z)$ and $g_2(Z)$, and
 other polynomials arising in this subsection are all
  isomorphic to one of $g_1(Z), g_2(Z)$ or $h_{12}(Z)$.

\bigskip

\noindent
 $ \begin{aligned}
   \circled{1}\quad  \{\mathcal{B(2)},&\mathcal{C}\} + \{\mathcal{B(1)},\mathcal{C}\}' :
                      \\
  & \{\mathcal{B(2)},\mathcal{C}\} :
G(\tfrac{1}{\beta})+G(b_2+\tfrac{1}{\beta})=a=G(u_2)+G(v_2),
 \,\,\, \{u_2+c,v_2+c \}=\{\tfrac{1}{\beta}, b_2+\tfrac{1}{\beta}\} \\
  &\{\mathcal{B(1)},\mathcal{C}\}' :
  G(\tfrac{1}{\alpha})+G(c+\tfrac{1}{\alpha})=a=G(u_1)+G(v_1),
\,\,\, u_1+v_1=c
\end{aligned}
$

\medskip

\noindent In this case,  one has  $b_1=c \in \{\tfrac{1}{\beta}+u_2, \tfrac{1}{\beta}+v_2\}$ because
$\{u_2+\tfrac{1}{\beta},v_2+\tfrac{1}{\beta} \}=\{c, b_2+c\}$. Since $\tfrac{1}{\beta}+u_2, \tfrac{1}{\beta}+v_2$ are
two roots of $\hh^{(a,b_2)}(X+\tfrac{1}{\beta})$ with $b_1 \in \{\tfrac{1}{\beta}+u_2, \tfrac{1}{\beta}+v_2\}$,
\begin{align}
0 &=\hh^{(a,b_2)}(b_1+\tfrac{1}{\beta})
  =\hh^{(a,b_2)}(b_1+\tfrac{1}{\alpha\beta}+\tfrac{1}{\alpha}) =h^{(a,b_2)}_1(b_1+\tfrac{1}{\alpha\beta}) \notag \\
  &= H_a(b_1+\tfrac{1}{\alpha\beta},
  b_2)=H_a(\tfrac{a}{\beta(\alpha a+\beta)},\tfrac{a+1}{\beta(\alpha a+1)})
   =\tfrac{1}{\beta a(\alpha a+\beta)^2(\alpha a+1)}\left(
 a^3+\tfrac{\beta^2}{\alpha^2}a+\tfrac{\beta^2}{\alpha^2}\right) \label{b2d1cubic}
\end{align}
with $b_1+\frac{1}{\alpha\beta}=\frac{a}{\beta(\alpha a+\beta)}$ and $b_2=\frac{a+1}{\beta(\alpha a +1)}$.
 Therefore, the existence of $a$ satisfying the above cubic
 polynomial determines $c$ via the relation $c=b_1=\tfrac{1}{\alpha(\alpha a+\beta)}$
  such that $\bu_G(a,c)\geq 6$ is
 satisfied. To guarantee $\bu_G(a,c)\geq 8$, one must have $u_1,v_1 \in \Fbn
 $ with $u_1, v_1 \notin \pole$ satisfying
 $$G(u_1)+G(v_1)=a,\quad  u_1+v_1=c $$
 Since $0=a^3+\tfrac{\beta^2}{\alpha^2}a+\tfrac{\beta^2}{\alpha^2}$ from the equation \eqref{b2d1cubic} and
since $c=b_1=\tfrac{1}{\alpha(\alpha a+\beta)}$, one finds $c\neq \tfrac{a}{\alpha a+1}, \tfrac{a}{\beta(\alpha a+1)}$.
 Therefore the existence of such $u_1,v_1$ in $\Fbn$ is equivalent of having
 $\tr(\tfrac{1}{ac\alpha^2})=0$ by Corollary \ref{abalphamcoro}.
One can re-express the cubic polynomial
 in (\ref{b2d1cubic}) as
 \begin{align}
0&=a^3+\frac{\beta^2}{\alpha^2}a+\frac{\beta^2}{\alpha^2}
  =1+\left(\frac{\beta}{\alpha a}\right)^2+\frac{\alpha}{\beta}\left(\frac{\beta}{\alpha a}\right)^3
 =\left(\frac{\beta}{\alpha a}\right)^3+\frac{\beta}{\alpha}\left(\frac{\beta}{\alpha
 a}\right)^2+\frac{\beta}{\alpha} \notag \\
  &=\left(\frac{1}{ac\alpha^2}+1\right)^3
  +\frac{\beta}{\alpha}\left(\frac{1}{ac\alpha^2}+1\right)^2+\frac{\beta}{\alpha} \label{b2d1cubicresuse}
 \end{align}
 where the last equality comes from Lemma \ref{tracelemma}-(a).
Since $\frac{1}{ac\alpha^2}=z^2+z$ for some $z\in \Fbn$ if and only if $\tr(\frac{1}{ac\alpha^2})=0$, by defining
\begin{align}
g_1(Z)&=(Z^2+Z+1)^3+\frac{\beta}{\alpha}(Z^2+Z+1)^2+\frac{\beta}{\alpha} \quad \in \Fbn[Z] \notag \\
      &=Z^6+Z^5+Z^3+Z+1+\frac{\beta}{\alpha}(Z^4+Z^2),
\end{align}
 the existence of $a,c$ with $\bu_G(a,c)\geq 8$ for the case $\{\mathcal{B(2)},\mathcal{C}\} + \{\mathcal{B(1)},\mathcal{C}\}'$
guarantees the existence of a root of $g_1(Z)=0$ in $\Fbn$. Conversely, if there is $z\in \Fbn$ such that $g_1(z)=0$,
then by defining
$$
a=\frac\beta\alpha \cdot\frac{1}{z^2+z+1}, \quad c=b_1=\frac{1}{\alpha(\alpha a+\beta)},
$$
one has $\bu_G(a,c)\geq 8$.

\bigskip

\noindent  $ \begin{aligned}
  \circled{2}\quad  \{\mathcal{B(3)},&\mathcal{C}\} +\{\mathcal{B(1)},\mathcal{C}\}' :
              \\
  & \{\mathcal{B(3)},\mathcal{C}\} : G(0)+G(b_3)=a=G(u_3)+G(v_3), \,\,\, \{u_3+c,v_3+c \}=\{0, b_3\} \\
  &\{\mathcal{B(1)},\mathcal{C}\}' : G(\tfrac{1}{\alpha})+G(c+\tfrac{1}{\alpha})=a=G(u_1)+G(v_1), \,\,\, u_1+v_1=c
\end{aligned}
$

\medskip

\noindent In this case,  one has  $b_1=c \in \{u_3,v_3\}$
       because $\{u_3,v_3 \}=\{c,
b_3+c\}$. Therefore $c=b_1$ is a root of $\hh^{(a,b_3)}(X)=(X+u_3)(X+v_3)$, and one has
\begin{align}
0&=\hh^{(a,b_3)}(b_1)=\hh^{(a,b_3)}(b_1+\tfrac1\alpha+\tfrac1\alpha) =h_1^{(a,b_3)}(b_1+\tfrac1\alpha) \notag \\
     &=H_a(b_1+\tfrac1\alpha,b_3)=H_a(\tfrac{a+1}{\alpha a+\beta},\tfrac{\alpha a+1}{\alpha^2 a})
      =\tfrac{1}{\alpha a^2(\alpha a+\beta)^2}\left( a^3+a^2+\tfrac{\beta}{\alpha^2}a+\tfrac{\beta^2}{\alpha^3}\right)
      \label{hab1b3}
\end{align}
The cubic polynomial in the above equations can be rewritten as
\begin{align}
0&=1+\frac{\alpha}{\beta}\left(\frac{\beta}{\alpha a}\right)+\frac{1}{\beta}\left(\frac{\beta}{\alpha
a}\right)^2+\frac{1}{\beta}\left(\frac{\beta}{\alpha a}\right)^3 =\left(\frac{\beta}{\alpha
a}\right)^3+\left(\frac{\beta}{\alpha a}\right)^2+\alpha\left(\frac{\beta}{\alpha
a}\right)+\beta \notag \\
 &=\left(\frac{1}{ac\alpha^2 }+1\right)^3+\left(\frac{1}{ac\alpha^2
}+1\right)^2+\alpha\left(\frac{1}{ac\alpha^2 }+1\right)+\beta, \label{b3d1eqn}
\end{align}
where the last equality again comes from Lemma \ref{tracelemma}-(a). In a similar way, by defining
\begin{align}
g_2(Z)&=(Z^2+Z+1)^3+(Z^2+Z+1)^2+\alpha(Z^2+Z+1)+\beta \quad \in \Fbn[Z] \notag \\
 &=Z^6+Z^5+Z^4+Z^3+\beta Z^2+\beta Z+1,
\end{align}
 the existence of $a,c$ with $\bu_G(a,c)\geq 8$
 for the case $\{\mathcal{B(3)},\mathcal{C}\} + \{\mathcal{B(1)},\mathcal{C}\}'$
guarantees the existence of a root of $g_2(Z)=0$ in $\Fbn$. Conversely, the
  existence of a solution $z\in \Fbn$ satisfying $g_2(z)=0$ guarantees the existence
  of $a,c\in \Fbn $ such that $\bu_G(a,c)\geq 8$, and such point $(a,c)$ can be
  constructed explicitly using $z$.

\bigskip

\noindent
  $ \begin{aligned}
  \circled{3} \quad \{\mathcal{B(2)},&\mathcal{C}\} +\{\mathcal{B(3)},\mathcal{C}\}' :
        \\
  & \{\mathcal{B(2)},\mathcal{C}\} :
    G(\tfrac1\beta)+G(b_2+\tfrac1\beta)=a=G(u_2)+G(v_2), \,\,\, \{u_2+c,v_2+c \}=\{\tfrac1\beta, b_2+\tfrac1\beta\} \\
  &\{\mathcal{B(3)},\mathcal{C}\}' : G(0)+G(c)=a=G(u_3)+G(v_3), \,\,\, u_3+v_3=c
\end{aligned}
$

\medskip

\noindent In this case,  one has  $b_3=c \in \{u_2+\tfrac1\beta,v_2+\tfrac1\beta\}$ because
$\{u_2+\tfrac1\beta,v_2+\tfrac1\beta \}=\{c, b_2+c\}$. Therefore $c=b_3$ is a root of
$\hh^{(a,b_2)}(X+\tfrac1\beta)=(X+u_2+\tfrac1\beta)(X+v_2+\tfrac1\beta)$, and one has
\begin{align*}
0=\hh^{(a,b_2)}(b_3+\tfrac1\beta)=\hh^{(a,b_2)}(b_3+\tfrac1{\alpha\beta}+\tfrac1\alpha)
     =h_1^{(a,b_2)}(b_3+\tfrac1{\alpha\beta})=H_a(b_3+\tfrac1{\alpha\beta},b_2)
\end{align*}
 Now we want to express the input values $b_3+\tfrac1{\alpha\beta}$ and $b_2$ of the bivariate polynomial $H_a$
 using the parameter $c$ to derive the same polynomial that we considered before.
Using $c=b_3=\frac{\alpha a+1}{\alpha^2 a}$, one has $b_2=\frac{a+1}{\beta(\alpha a+1)}=\frac{\alpha^2
c+\beta}{\alpha^2\beta c}$. Therefore
 \begin{align}
0&=H_a(b_3+\tfrac{1}{\alpha\beta},b_2)
 =H_a(c+\tfrac{1}{\alpha\beta},\tfrac{\alpha^2 c+\beta}{\alpha^2\beta c})
 =\tfrac{1}{ c}\left( c^3+\tfrac{1}{\alpha^2}c+\tfrac{1}{\alpha^2\beta}\right) \label{b2d3cubic}
 \end{align}
In a similar way to previous cases, one may use
 Lemma \ref{tracelemma}-$(b)$ to rewrite the above cubic polynomial as
\begin{align*}
0&=c^3+\frac{1}{\alpha^2}c+\frac{1}{\alpha^2\beta}
   =1+\left(\frac{1}{\alpha
c}\right)^2+\frac{\alpha}{\beta}\left(\frac{1}{\alpha c}\right)^3
  =\left(\frac{1}{\alpha c}\right)^3+\frac{\beta}{\alpha}\left(\frac{1}{\alpha c}\right)^2+\frac{\beta}{\alpha}\\
  &=\left(\frac{1}{ac\alpha^2}+1\right)^3+\frac{\beta}{\alpha}\left(\frac{1}{ac\alpha^2}+1\right)^2+\frac{\beta}{\alpha},
\end{align*}
 where the last expression is exactly same to the expression in the equation (\ref{b2d1cubicresuse}).
 Therefore we have exactly the same polynomial
$$
g_1(Z)=Z^6+Z^5+Z^3+Z+1+\frac{\beta}{\alpha}(Z^4+Z^2)
$$
as in the case $\{\mathcal{B(2)},\mathcal{C}\}+\{\mathcal{B(1)},\mathcal{C}\}'$ so that the existence of a root $z\in
\Fbn$
 satisfying $g_1(z)=0$ is equivalent to the existence of $a,c\in \Fbn$ satisfying $\bu_G(a,c)\geq 8$ with the type
 $\{\mathcal{B(2)},\mathcal{C}\} +\{\mathcal{B(3)},\mathcal{C}\}'$.
  Note that such  $a$ and $c$ are explicitly given as
 $$
 c=\frac{1}{\alpha}\cdot \frac{1}{z^2+z+1}, \quad a=\frac{1}{\alpha (\alpha c+1)}(=G(0)+G(c))
 $$

\bigskip

\noindent
  $ \begin{aligned}
   \circled{4} \quad \{\mathcal{B(1)},&\mathcal{C}\} +\{\mathcal{B(3)},\mathcal{C}\}' :
    \\
  & \{\mathcal{B(1)},\mathcal{C}\} :
G(\tfrac{1}{\alpha})+G(b_1+\tfrac{1}{\alpha})=a=G(u_1)+G(v_1),
 \,\,\, \{u_1+c,v_1+c \}=\{\tfrac{1}{\alpha}, b_1+\tfrac{1}{\alpha}\} \\
  &\{\mathcal{B(3)},\mathcal{C}\}' :
  G(0)+G(c)=a=G(u_3)+G(v_3),
\,\,\, u_3+v_3=c
\end{aligned}
$

\medskip
\noindent In this case,  one has  $b_3=c \in \{u_1+\frac{1}{\alpha},v_1+\frac{1}{\alpha} \}$ because
$\{u_1+\frac{1}{\alpha},v_1+\frac{1}{\alpha} \}=\{c, b_1+c\}$.
 Since $\frac{1}{\alpha}+u_1, \frac{1}{\alpha}+v_1$ are
two roots of $\hh^{(a,b_1)}(X+\frac{1}{\alpha})=h^{(a,b_1)}_1(X)$ with $c=b_3 \in \{\frac{1}{\alpha}+u_1,
\frac{1}{\alpha}+v_1\}$, one gets
 \begin{align}
0&=h_1^{(a,b_1)}(b_3)=H_a(b_3,b_1)=H_a(c,\tfrac{\alpha c+1}{\alpha^2(\beta c+1)})
 =\tfrac{\beta}{\beta c+1}\left( c^3+\tfrac{1}{\beta}c^2+\tfrac{1}{\alpha^2\beta}c+\tfrac{1}{\alpha^3\beta}\right)
 \label{hab3b1}
 \end{align}
and, using Lemma \ref{tracelemma}-$(b)$, the above cubic polynomial is rewritten as
\begin{align*}
 0&=c^3+\frac{1}{\beta}c^2+\frac{1}{\alpha^2\beta}c+\frac{1}{\alpha^3\beta}
      =1+\frac{\alpha}{\beta}\left(\frac{1}{\alpha c}\right)
      +\frac{1}{\beta}\left(\frac{1}{\alpha c}\right)^2+\frac{1}{\beta}\left(\frac{1}{\alpha
      c}\right)^3 \\
      &=\left(\frac{1}{\alpha c}\right)^3+\left(\frac{1}{\alpha c}\right)^2+\alpha \left(\frac{1}{\alpha
      c}\right)+\beta \\
        &=\left(\frac{1}{ac\alpha^2}+1\right)^3+\left(\frac{1}{ac\alpha^2}+1\right)^2+\alpha
        \left(\frac{1}{ac\alpha^2}+1\right)+\beta,
\end{align*}
where the last expression is exactly same to the expression in the equation \eqref{b3d1eqn}.
 Therefore we have the exactly same polynomial
 $$g_2(Z)=Z^6+Z^5+Z^4+Z^3+\beta Z^2+\beta Z+1$$
 as in the combination $\{\mathcal{B(3)},\mathcal{C}\} +\{\mathcal{B(1)},\mathcal{C}\}'$, and the existence of a root $z\in \Fbn$
 satisfying $g_2(z)=0$ is equivalent to the existence of $a,c\in \Fbn$ satisfying $\bu_G(a,c)\geq 8$ of the type
   $\{\mathcal{B(1)},\mathcal{C}\} +\{\mathcal{B(3)},\mathcal{C}\}'$.

\bigskip

\noindent
  $ \begin{aligned}
  \circled{5}\quad \{\mathcal{B(1)},&\mathcal{C}\} +\{\mathcal{B(2)},\mathcal{C}\}' :
     \\
  & \{\mathcal{B(1)},\mathcal{C}\} :
G(\tfrac{1}{\alpha})+G(b_1+\tfrac{1}{\alpha})=a=G(u_1)+G(v_1),
 \,\,\, \{u_1+c,v_1+c \}=\{\tfrac{1}{\alpha}, b_1+\tfrac{1}{\alpha}\} \\
  &\{\mathcal{B(2)},\mathcal{C}\}' :
  G(\tfrac{1}{\beta})+G(c+\tfrac{1}{\beta})=a=G(u_2)+G(v_2),
\,\,\, u_2+v_2=c
\end{aligned}
$

\medskip
\noindent In this case,  one has  $b_2=c \in \{u_1+\frac{1}{\alpha},v_1+\frac{1}{\alpha} \}$ because
$\{u_1+\frac{1}{\alpha},v_1+\frac{1}{\alpha} \}=\{c, b_1+c\}$.
 Since $\frac{1}{\alpha}+u_1, \frac{1}{\alpha}+v_1$ are
two roots of $\hh^{(a,b_1)}(X+\frac{1}{\alpha})=h^{(a,b_1)}_1(X)$ with $b_2 \in \{\frac{1}{\alpha}+u_1,
\frac{1}{\alpha}+v_1\}$,
\begin{align}
0&=h^{(a,b_1)}_1(b_2)=H_a(b_2,b_1)
 =\tfrac{\alpha}{\beta^2 a(\alpha a+\beta)(\alpha a+1)^2}
   \left( a^4+a^2+\tfrac{\beta^2}{\alpha^2}a+\tfrac{\beta^2}{\alpha^4}\right) \label{f1quartic}
\end{align}
 with $b_1=\frac{1}{\alpha(\alpha a+\beta)}$ and $b_2=\frac{a+1}{\beta(\alpha a +1)}$.
  The above quartic polynomial (of $a$)  has a close connection with $h_{12}(Z)$ (See the equation \eqref{h12x})
  that we defined in Section \ref{sectionh12x}. In fact, by defining
  \begin{align*}
  f_1(Z)=Z^4+Z^2+\frac{\beta^2}{\alpha^2}Z+\frac{\beta^2}{\alpha^4},
  \end{align*}
  it is straightforward
  to verify
  \begin{align*}
   f_1\left(\frac{\beta^2}{\alpha^2}Z+1\right)=\frac{\beta^4}{\alpha^4}Z^4h_{12}\left(\frac{1}{Z}\right),
  \end{align*}
 which implies that there is one to one correspondence between the roots of $f_1$ and the roots of $h_{12}$,
 and $f_1$ can be discarded in our analysis. Note that the above isomorphism can also be given using
 homogeneous polynomials in view of Remark \ref{homo}.

\bigskip

\noindent
  $ \begin{aligned}
   \circled{6} \quad \{\mathcal{B(3)},&\mathcal{C}\} +\{\mathcal{B(2)},\mathcal{C}\}' :
      \\
  & \{\mathcal{B(3)},\mathcal{C}\} :
G(0)+G(b_3)=a=G(u_3)+G(v_3),
 \,\,\, \{u_3+c,v_3+c \}=\{0, b_3\} \\
  &\{\mathcal{B(2)},\mathcal{C}\}' :
  G(\tfrac{1}{\beta})+G(c+\tfrac{1}{\beta})=a=G(u_2)+G(v_2),
\,\,\, u_2+v_2=c
\end{aligned}
$

\medskip
\noindent In this case,  one has  $b_2=c \in \{u_3,v_3 \}$ because $\{u_3,v_3 \}=\{c, b_3+c\}$.
 Since $u_3, v_3$ are
two roots of $\hh^{(a,b_3)}(X)=(X+u_3)(X+v_3)$ with $b_2=c \in \{u_3, v_3\}$,
\begin{align}
0&=\hh^{(a,b_3)}(b_2)=\hh^{(a,b_3)}(b_2+\tfrac{1}{\alpha}+\tfrac{1}{\alpha}) \notag \\
 &=H_a(b_2+\tfrac{1}{\alpha}, b_3) =H_a(\tfrac{\alpha^2 a+1}{\alpha\beta(\alpha a+1)},\tfrac{\alpha a+1}{\alpha^2 a})
 =\tfrac{\alpha}{\beta^2 a^2(\alpha a+1)^2}\left( a^4+\tfrac{\beta}{\alpha}a^3+
 \tfrac{1}{\alpha^2}a^2 + \tfrac{\beta}{\alpha^3}a +\tfrac{\beta^2}{\alpha^5}\right) \label{f2quartic}
\end{align}
 We claim that the above quartic polynomial (of $a$) is also isomorphic
           to the  polynomial $h_{12}(Z)$. That is, by defining
  \begin{align*}
  f_2(Z)=Z^4+\frac\beta\alpha Z^3+\frac{1}{\alpha^2} Z^2+\frac{\beta}{\alpha^3}Z+\frac{\beta^2}{\alpha^5},
  \end{align*}
  it is straightforward
  to verify
  \begin{align*}
  f_2\left (Z+\dfrac{1}{\alpha^3}\right ) = \dfrac{\alpha^7}{\beta^6}Z^4 h_{12}\left (\dfrac{\beta^3}{\alpha^4Z}\right )
  \end{align*}
 and $f_2$ can also be eliminated in our analysis.

\subsection{$\{\mathcal{B},\mathcal{B}\}+\{\mathcal{B},\mathcal{C}\}$}

In Section \ref{differentialsec}, we showed that $\du_G(a,b)\geq 4$ can  happen
  when both $(p_i, b+p_i), (p_j, b+p_j)$ are of $\cb$ (i.e., $b+p_i, b+p_j \notin \pole$) such that
 \begin{align}
G(p_i)+G(b+p_i)=a=G(p_j)+G(b+p_j),  \label{bbeqn}
 \end{align}
 where $\{p_i,b+p_i\},\{p_j, b+p_j\}\in \mathcal{B}$ and the above $\bu$ equation is called of type
  $\{\mathcal{B},\mathcal{B}\}$.
For given $a\notin \{1,\tfrac\beta\alpha, \tfrac1\alpha\}$, the $\du$ equation $G(p_i)+G(b+p_i)=a$ has a unique
solution  $b=b_i$ with
   $\{b_1,b_2,b_3\}=\{\tfrac1{\alpha(\alpha a+\beta)}, \tfrac{a+1}{\beta(\alpha a+1)}, \tfrac{\alpha a+1}{\alpha^2
   a}\}$, and the conditions on $a$ for which the above simultaneous equation (\ref{bbeqn}) has a solution are
    stated in the equations \eqref{eqnb1b2}, \eqref{eqnb1b3} and \eqref{eqnb2b3}.

\medskip

 In this section, we will show that the only possible case having $\bu_G(a,c)\geq 8$ of the type
 $\{\mathcal{B},\mathcal{B}\}+\{\mathcal{B},\mathcal{C}\}$ is the case where the polynomial
  $\phi(Z)=Z^2+\tfrac\beta\alpha Z+\tfrac{\beta}{\alpha^2}$ (defined in Theorem \ref{butheorem}) has a root in $\Fbn$.
 When we consider $\{\mathcal{B},\mathcal{B}\}+\{\mathcal{B},\mathcal{C}\}$ where $\bu_G(a,c)\geq 8$ happens, we only
need to consider three different cases of $\{i,j,k\}=\{1,2,3\},\{1,3,2\},\{2,3,1\}$ such that
 \begin{align*}
   \{\mathcal{B({\it i})},\mathcal{B({\it j})}\} &:
G(p_i)+G(b+p_i)=a=G(p_j)+G(b+p_j),
 \,\,\, b=b_i=b_j \\
    \{\mathcal{B({\it k})},\mathcal{C}\} &:
G(p_k)+G(b_k+p_k)=a=G(u)+G(v),
 \,\,\, b_k=u+v
\end{align*}
and, for each $\{i,j,k\}=\{1,2,3\},\{1,3,2\},\{2,3,1\}$, one of the following holds;

 \begin{enumerate}[(i)]
 \item $\{p_j+c,b+p_j+c \}=\{p_i, b+p_i\}$ and $\{u+c,v+c \}=\{p_k, b_k+p_k\} :$ In this case,
    one has
    $$c \in \{p_i+p_j, p_i+p_j+b\} \cap \{p_k+u,p_k+v\}$$
     because $\{p_i+p_j, p_i+p_j+b\}=\{c, b+c\}$ and $\{p_k+u,p_k+v\}=\{c, b_k+c\}$.
 \item $\{p_j+c,b+p_j+c \}=\{p_i, b+p_i\}$ with $c=b_k=u+v :$
  In this case, one has
   $$ b_k \in \{p_i+p_j, p_i+p_j+b\}$$
   because $\{p_i+p_j, p_i+p_j+b\}=\{c, b+c\}$ with $c=b_k$.
 \item $\{u+c,v+c \}=\{p_k, b_k+p_k\}$ with  $c=b=b_i=b_j :$
  In this case, one has
  $$b \in \{p_k+u,p_k+v\}$$
   because $\{p_k+u,p_k+v\}=\{c, b_k+c\}$
   with $c=b$.
 \end{enumerate}

Now letting
 \begin{align*}
  r_{ij}(X)=(X+p_i+p_j)(X+p_i+p_j+b) \,\,\, \textrm{with}\,\, b=b_i=b_j,
 \end{align*}
 we will discuss three possible combinations of
$\{\mathcal{B},\mathcal{B}\}+\{\mathcal{B},\mathcal{C}\}$. Our subsequent arguments have certain similarity with the
 case of $\{\mathcal{A},\mathcal{C}\}+\{\mathcal{B},\mathcal{C}\}$ but involve much delicate computations so that
 we will give rather detailed explanations if necessary.

%Note that $b\neq b'$ because if $b=b'$ then all $b_1,b_2,b_3$ are same and we have $\du_G(a,b)=8$ (i.e.,
%$\beta^3+\beta^2+1=0$) which we excluded in our analysis.

\bigskip
$ \begin{aligned}
  \{\mathcal{B(1)},&\mathcal{B(3)}\}+\{\mathcal{B(2)},\mathcal{C}\} : \\
  & \{\mathcal{B(1)},\mathcal{B(3)}\} :
G(\textstyle\frac{1}{\alpha})+G(b_1+\frac{1}{\alpha})=a=G(0)+G(b_1),
 \,\,\, b=b_1=b_3 \\
   & \{\mathcal{B(2)},\mathcal{C}\} :
G(\textstyle\frac{1}{\beta})+G(b_2+\frac{1}{\beta})=a=G(u_2)+G(v_2),
 \,\,\, b_2=u_2+v_2
\end{aligned}
$

\medskip \noindent Letting $r_{13}(X)=(X+\frac{1}{\alpha})(X+b_1+\frac{1}{\alpha})$ and
  $h_2^{(a,b_2)}(X)=\hh^{(a,b_2)}(X+\frac{1}{\beta})=(X+\frac{1}{\beta}+u_2)(X+\frac{1}{\beta}+v_2)$,
 one of the following three cases holds;

 \begin{enumerate}[(i)]
 \item $\{c,b_1+c \}=\{\frac{1}{\alpha}, b_1+\frac{1}{\alpha}\}$
 and $\{u_2+c,v_2+c \}=\{\frac{1}{\beta}, b_2+\frac{1}{\beta}\}$ \\
   $\,\,\,\Rightarrow\,\,\,c \in \{\frac{1}{\alpha}, b_1+\frac{1}{\alpha}\} \cap
\{u_2+\frac{1}{\beta}, v_2+\frac{1}{\beta}\} $
  $\,\,\,\Rightarrow\,\,\, \res(r_{13},h_2^{(a,b_2)})=0$
 \item $\{c,b_1+c \}=\{\frac{1}{\alpha}, b_1+\frac{1}{\alpha}\}$ with $c=b_2$
  $\,\,\,\Rightarrow\,\,\, b_2 \in \{\frac{1}{\alpha}, b_1+\frac{1}{\alpha}\}$
   $\,\,\,\Rightarrow\,\,\, r_{13}(b_2)=0$
 \item $\{u_2+c,v_2+c \}=\{\frac{1}{\beta}, b_2+\frac{1}{\beta}\}$ with  $c=b=b_1=b_3$
  $\,\,\,\Rightarrow\,\,\, b \in \{u_2+\tfrac1\beta,v_2+\tfrac1\beta\}$
   $\,\,\,\Rightarrow\,\,\, h_2^{(a,b_2)}(b)=0$
 \end{enumerate}

\medskip

\noindent  $\{\mathcal{B(1)},\mathcal{B(3)}\}+\{\mathcal{B(2)},\mathcal{C}\}$\,\,\, Case (i) :

\noindent Since the resultant is invariant under the translation of $X$, one needs to check $0=\res(r_{13}(X),
\hh^{(a,b_2)}(X+\frac{1}{\beta}))=\res(r_{13}(X+\frac{1}{\alpha\beta}), \hh^{(a,b_2)}(X+\frac{1}{\alpha}))$ where
\begin{align*}
r_{13}(X+\tfrac{1}{\alpha\beta})
 &=X^2+b_1X+\tfrac{1}{\beta}(\tfrac{1}{\beta}+b_1) =X^2+\tfrac{a}{\beta}X+\tfrac{a+1}{\beta^2}\\
\hh^{(a,b_2)}(X+\tfrac{1}{\alpha})&=X^2+b_2X+\tfrac{b_2}{\alpha^2a}
                       =X^2+(a+\tfrac{\beta^2+\alpha}{\alpha\beta})X+\tfrac{\beta^2+\alpha}{\alpha\beta^2}a+\tfrac{1}{\beta}
\end{align*}
Note that the linear expression of $b_1=\frac{a}{\beta}$ and $b_2=a+\frac{\beta^2+\alpha}{\alpha\beta}$ come from
$b_1=\frac{1}{\alpha(\alpha a+\beta)}=\frac{\alpha a+1}{\alpha^2 a}=b_3$ (i.e., $\alpha^2 a^2+\alpha\beta a+\beta=0$).
Via explicit computations, one has
\begin{align*}
\res(r_{13}(X+\tfrac{1}{\alpha\beta}),
\hh^{(a,b_2)}(X+\tfrac{1}{\alpha}))=\tfrac{1}{\alpha^2\beta^3}(\alpha^2a^2+\alpha\beta a+\beta)=0
\end{align*}
Conversely, if the following polynomial
\begin{align}
 \phi(Z)=Z^2+\frac\beta\alpha Z+\frac{\beta}{\alpha^2}  \label{phiz}
\end{align}
has a root $z\in \Fbn$ satisfying $\phi(z)=0$, then letting $(a,c)=(z, \tfrac{\alpha z+\beta}{\alpha\beta})$
 where $c=b_1+\tfrac{1}{\alpha}$ is the common root of $r_{13}(X)$ and $h_2^{(a,b_2)}(X)$, one
 gets $\bu_G(a,c)\geq 8$. So, in this case,
 the necessary and sufficient condition for $\bu_G(a,c)\geq 8$ is the existence of
 solution of $\phi(Z)=0$ in $\Fbn$, and a solution exists if and only $\tr(\frac{1}{\beta})=0$.

\begin{remark}
Note that we are always assuming $\beta^3+\beta^2+1\neq 0$ throughout sections after Remark \ref{betaremark} in Section
\ref{sectionacbc}.
 However when $\beta^3+\beta^2+1=0$, then from $\beta^2+\beta=\frac{1}{\beta}$, one has
$\tr(\frac{1}{\beta})=\tr(\beta^2+\beta)=0$. Therefore $\beta^3+\beta^2+1=0$ implies $\tr(\frac{1}{\beta})=0$,
 and eliminating the case $\tr(\frac{1}{\beta})=0$ also eliminates the case $\beta^3+\beta^2+1=0$.
\end{remark}

\medskip

\noindent $\{\mathcal{B(1)},\mathcal{B(3)}\}+\{\mathcal{B(2)},\mathcal{C}\}$\,\,\,  Case (ii) :

\noindent Using the (linear) expression of $b_1=\tfrac{1}{\beta} a, b_2=a+\tfrac{\beta^2+\alpha}{\alpha\beta}$,
$r_{13}(b_2)=(b_2+\tfrac{1}{\alpha})(b_2+b_1+\tfrac1\alpha)=0$ implies $\beta^3+\beta^2+1=0$ which cannot happen.

\medskip

\noindent  $\{\mathcal{B(1)},\mathcal{B(3)}\}+\{\mathcal{B(2)},\mathcal{C}\}$\,\,\,  Case (iii) :

\noindent One has
 $0=h_2^{(a,b_2)}(b)=\hh^{(a,b_2)}(b+\tfrac{1}{\beta})=
 \hh^{(a,b_2)}(b+\tfrac{1}{\alpha\beta}+\tfrac{1}{\alpha})=H_a(b+\tfrac1{\alpha\beta}, b_2)$ with $b=b_1=b_3$
 which implies that the three simultaneous equations \eqref{b2d1cubic},\eqref{b2d3cubic} and \eqref{eqnb1b3}
 must be satisfied, and it is straightforward to show the common solution does not exist.

\bigskip
$ \begin{aligned}
  \{\mathcal{B(2)},&\mathcal{B(3)}\}+\{\mathcal{B(1)},\mathcal{C}\} : \\
  & \{\mathcal{B(2)},\mathcal{B(3)}\} :
G(\tfrac{1}{\beta})+G(b_2+\tfrac{1}{\beta})=a=G(0)+G(b_2),
 \,\,\, b=b_2=b_3 \\
   & \{\mathcal{B(1)},\mathcal{C}\} :
G(\tfrac{1}{\alpha})+G(b_1+\tfrac{1}{\alpha})=a=G(u_1)+G(v_1),
 \,\,\, b_1=u_1+v_1
\end{aligned}
$

\medskip

\noindent Letting $r_{23}(X)=(X+\frac{1}{\beta})(X+b_2+\frac{1}{\beta})$ and
  $h_1^{(a,b_1)}(X)=\hh^{(a,b_1)}(X+\tfrac{1}{\alpha})
=(X+\tfrac{1}{\alpha}+u_1)(X+\tfrac{1}{\alpha}+v_1)=X^2+b_1X+\tfrac{b_1}{\alpha^2 a}$,
 one of the following three cases holds;

 \begin{enumerate}[(i)]
 \item $\{c,b_2+c \}=\{\frac{1}{\beta}, b_2+\frac{1}{\beta}\}$
 and $\{u_1+c,v_1+c \}=\{\frac{1}{\alpha}, b_1+\frac{1}{\alpha}\}$ \\
   $\,\,\,\Rightarrow\,\,\,c \in \{\frac{1}{\beta}, b_2+\frac{1}{\beta}\} \cap
\{u_1+\frac{1}{\alpha}, v_1+\frac{1}{\alpha}\} $
  $\,\,\,\Rightarrow\,\,\, \res(r_{23},h_1^{(a,b_1)})=0$
 \item $\{c,b_2+c \}=\{\frac{1}{\beta}, b_2+\frac{1}{\beta}\}$ with $c=b_1$
  $\,\,\,\Rightarrow\,\,\, b_1 \in \{\frac{1}{\beta}, b_2+\frac{1}{\beta}\}$
   $\,\,\,\Rightarrow\,\,\, r_{23}(b_1)=0$
 \item $\{u_1+c,v_1+c \}=\{\frac{1}{\alpha}, b_1+\frac{1}{\alpha}\}$ with  $c=b=b_2=b_3$
  $\,\,\,\Rightarrow\,\,\, b \in \{u_1+\tfrac1\alpha,v_1+\tfrac1\alpha\}$
   $\,\,\,\Rightarrow\,\,\, h_1^{(a,b_1)}(b)=0$
 \end{enumerate}

\medskip

\noindent $\{\mathcal{B(2)},\mathcal{B(3)}\}+\{\mathcal{B(1)},\mathcal{C}\}$\,\,\, Case (i) :

\noindent Please note that the resultant has a multiplicative expression
  $\res(r_{23},h_1^{(a,b_1)})=h_1^{(a,b_1)}(\tfrac1\beta)h_1^{(a,b_1)}(b_2+\tfrac1\beta)$.
 We will show that $h_1^{(a,b_1)}(\tfrac1\beta)\neq 0\neq h_1^{(a,b_1)}(b_2+\tfrac1\beta)$
  (i.e., $\res(r_{23},h_1^{(a,b_1)})\neq 0$) under
our basic assumption $\beta^3+\beta^2+1\neq 0$. When $c=\tfrac{1}{\beta}$,
\begin{align*}
h_1^{(a,b_1)}(\tfrac{1}{\beta})
 &=\tfrac{1}{\beta^2}+b_1(\tfrac{1}{\beta}+\tfrac{1}{\alpha^2 a})
    = \tfrac{1}{\beta^2}+\tfrac{1}{\alpha(\alpha a+\beta)}(\tfrac{1}{\beta}+\tfrac{1}{\alpha^2   a}) \notag \\
 &=\tfrac{1}{\alpha^3\beta^2a(\alpha a+\beta)}\left(\alpha^3 a(\alpha a +\beta)+ \alpha^2\beta a+\beta^2 \right)
 =\tfrac{1}{\alpha^3\beta^2a(\alpha a+\beta)}\left(\alpha^4a^2+\alpha^2\beta^2a+\beta^2\right) \notag \\
 &=\tfrac{\alpha}{\beta^2a(\alpha a+\beta)}\left(a^2+\tfrac{\beta^2}{\alpha^2}a+\tfrac{\beta^2}{\alpha^4}\right)
\end{align*}
By using the relation $a^2+\tfrac{1}{\alpha}a+\tfrac{\beta}{\alpha^3}=0$
 (i.e., the condition $b_2=b_3$ in the equation \eqref{eqnb2b3}), the
last expression in the above equations can be written as
\begin{align*}
a^2+\tfrac{\beta^2}{\alpha^2}a+\tfrac{\beta^2}{\alpha^4}=a^2+\tfrac{\beta^2}{\alpha^2}a
  +\tfrac{\beta}{\alpha}\left(a^2+\tfrac{1}{\alpha}a \right)
  = \tfrac1\alpha a^2+\tfrac\beta\alpha a=\tfrac1\alpha a(a+\beta)
\end{align*}
Thus $h_1^{(a,b_1)}(\tfrac{1}{\beta})=0$ can happen only if $a=\beta$, however again from $b_2=b_3$,
\begin{align*}
0=a^2+\tfrac{1}{\alpha}a+\tfrac{\beta}{\alpha^3}=\beta^2+\tfrac\beta\alpha+\tfrac\beta{\alpha^3}
     =\tfrac\beta{\alpha^3}\left(\alpha^3\beta+\alpha^2+1  \right)
     =\tfrac{\beta^2}{\alpha^3}\left(\alpha^3+\beta  \right)
     =\tfrac{\beta^2}{\alpha^3}\left(\beta^3+\beta^2+1 \right)\neq 0,
\end{align*}
which is a contradiction. In a similar manner, when $c=b_2+\tfrac1{\beta}=\tfrac{a}{\alpha a+1}$,
\begin{align*}
h_1^{(a,b_1)}(b_2+\tfrac{1}{\beta})
 &=\tfrac{a^2}{(\alpha a+1)^2}+b_1(\tfrac{a^2}{(\alpha a+1)^2}+\tfrac{1}{\alpha^2 a})
    =  \tfrac{a^2}{(\alpha a+1)^2}+\tfrac{1}{\alpha(\alpha a+\beta)}(\tfrac{a^2}{(\alpha a+1)^2}+\tfrac{1}{\alpha^2 a})\notag \\
 &=\tfrac{1}{\alpha^3 a(\alpha a+\beta)(\alpha a+1)^2}\left( \alpha^4a^4+\alpha^4a^3+1  \right)
 =\tfrac{\alpha}{ a(\alpha a+\beta)(\alpha a+1)^2}\left(a^4+a^3+\tfrac1{\alpha^4}\right)
\end{align*}
 Thus $h_1^{(a,b_1)}(b_2+\tfrac{1}{\beta})=0$ can happen only when $ a^4+a^3+\tfrac1{\alpha^4}=0$, however using
 the same relation $a^2+\tfrac{1}{\alpha}a+\tfrac{\beta}{\alpha^3}=0$ again,
\begin{align*}
 0=a^4+a^3+\tfrac1{\alpha^4}&=a^4+\tfrac1{\alpha^2}a^2+ a^3+\tfrac1{\alpha^2}a^2+\tfrac1{\alpha^4}
  = \left(a^2+\tfrac1\alpha a\right)^2+ a\left(a^2+\tfrac1\alpha a+\tfrac{\beta}{\alpha^2}a
  \right)+\tfrac1{\alpha^4}\\
  &=\tfrac{\beta^2}{\alpha^6}+ a\left(\tfrac{\beta}{\alpha^3}+\tfrac{\beta}{\alpha^2}a
  \right)+\tfrac1{\alpha^4} =
  \tfrac\beta{\alpha^2}a^2+\tfrac\beta{\alpha^3}a+\tfrac1{\alpha^6}
    =\tfrac\beta{\alpha^2}\left(a^2+\tfrac1\alpha a\right)+\tfrac1{\alpha^6}\\
    &=\tfrac\beta{\alpha^2}\cdot\tfrac\beta{\alpha^3}+\tfrac1{\alpha^6}=\tfrac1{\alpha^6}(\beta^3+\beta^2+1)\neq 0
\end{align*}
which is a contradiction.

\medskip

\noindent  $\{\mathcal{B(2)},\mathcal{B(3)}\}+\{\mathcal{B(1)},\mathcal{C}\}$\,\,\,  Case (ii) :

\noindent Using $b_2=b_3$, one finds the linear expression $b_2=b_3=\tfrac{\alpha}{\beta}a+\tfrac{1}{\alpha\beta}$
 and $b_1=\tfrac{\alpha}{\beta^3}(a+1)$. From this, one concludes
 $r_{23}(b_1)=(b_1+\tfrac1\beta)(b_1+b_2+\tfrac1\beta)=0$ implies $\beta^3+\beta^2+1=0$ when $b_1+\tfrac1\beta=0$,
 which is not possible. Also, $b_1+b_2+\tfrac1\beta=0$ implies $a=\alpha, b_1=\beta^2+\beta+1, b_2=b_3=\beta^3$ and
 $\beta^5=1$. This last case gives $\du_G(a,b_1)=2<4$
  (i.e., there is no $u_1,v_1\in \mathcal{C}$ satisfying $a=G(u_1)+G(v_1),\,\, b_1=u_1+v_1$)
   because $b_1=\tfrac{a}{\alpha a+1}$ and one has $\# \mathcal{C}_G^{(a,b)}=0$ by Corollary \ref{abalphamcoro}.
 Therefore $\bu_G(a,c)\geq 8$ is not possible in this case.

\medskip

\noindent  $\{\mathcal{B(2)},\mathcal{B(3)}\}+\{\mathcal{B(1)},\mathcal{C}\}$\,\,\,  Case (iii) :

\noindent One has $0=h_1^{(a,b_1)}(b)=\hh^{(a,b_1)}(b+\tfrac{1}{\alpha})=H_a(b, b_1)$ with $b=b_2=b_3$
 which  implies that the three simultaneous equations \eqref{hab3b1},\eqref{f1quartic} and \eqref{eqnb2b3}
 must be satisfied, and it is easy to show the common solution does not exist.

\bigskip
   $ \begin{aligned}
  \{\mathcal{B(1)},&\mathcal{B(2)}\}+\{\mathcal{B(3)},\mathcal{C}\} : \\
  & \{\mathcal{B(1)},\mathcal{B(2)}\} :
G(\tfrac{1}{\alpha})+G(b_1+\tfrac{1}{\alpha})=a=G(\tfrac{1}{\beta})+G(b_1+\tfrac{1}{\beta}),
 \,\,\, b=b_1=b_2 \\
   & \{\mathcal{B(3)},\mathcal{C}\} :
G(0)+G(b_3)=a=G(u_3)+G(v_3),
 \,\,\, b_3=u_3+v_3
\end{aligned}
$

\medskip

\noindent Letting $r_{12}(X)=(X+\frac{1}{\alpha\beta})(X+b_1+\frac{1}{\alpha\beta})$ and
  $h_3^{(a,b_3)}(X)=\hh^{(a,b_3)}(X)
=(X+u_3)(X+v_3)$,
 one of the following three cases holds;

 \begin{enumerate}[(i)]
 \item $\{\tfrac{1}{\beta}+c,b_1+\tfrac{1}{\beta}+c \}=\{\tfrac{1}{\alpha}, b_1+\tfrac{1}{\alpha}\}$
 and $\{u_3+c,v_3+c \}=\{0, b_3\}$ \\
   $\,\,\,\Rightarrow\,\,\,c \in \{\frac{1}{\alpha\beta}, b_1+\frac{1}{\alpha\beta}\} \cap
\{u_3, v_3\} $
  $\,\,\,\Rightarrow\,\,\, \res(r_{12},h_3^{(a,b_3)})=0$
 \item $\{\tfrac{1}{\beta}+c,b_1+\tfrac{1}{\beta}+c \}=\{\tfrac{1}{\alpha}, b_1+\tfrac{1}{\alpha}\}$ with $c=b_3$
  $\,\,\,\Rightarrow\,\,\, b_3 \in \{\frac{1}{\alpha\beta}, b_1+\frac{1}{\alpha\beta}\}$
   $\,\,\,\Rightarrow\,\,\, r_{12}(b_3)=0$
 \item $\{u_3+c,v_3+c \}=\{0, b_3\}$ with  $c=b=b_1=b_2$
  $\,\,\,\Rightarrow\,\,\, b \in \{u_3,v_3\}$
   $\,\,\,\Rightarrow\,\,\, h_3^{(a,b_3)}(b)=0$
 \end{enumerate}

\medskip

\noindent   $\{\mathcal{B(1)},\mathcal{B(2)}\}+\{\mathcal{B(3)},\mathcal{C}\}$\,\,\, Case (i) :

\noindent Since the resultant is invariant under the translation of $X$, one needs to check $0=\res(r_{12}(X),
\hh^{(a,b_3)}(X)=\res(r_{12}(X+\frac{1}{\alpha}), \hh^{(a,b_3)}(X+\frac{1}{\alpha}))$ where
\begin{align*}
r_{12}(X+\tfrac{1}{\alpha})
 &=X^2+b_1X+\tfrac{1}{\beta}(\tfrac{1}{\beta}+b_1) =X^2+\tfrac{\alpha a+1}{\alpha^2\beta}X+\tfrac{\alpha a+\beta^2}{\alpha^2\beta^2}\\
\hh^{(a,b_3)}(X+\tfrac{1}{\alpha})&=X^2+b_3X+\tfrac{b_3}{\alpha^2a}
                       =X^2+\tfrac{\alpha a+\beta^2+\alpha}{\alpha\beta^2}X+
                        \tfrac{(\alpha^2+\alpha\beta^2)a+\alpha\beta^2+1}{\alpha^2\beta^4}
\end{align*}
Note that the linear expression of $b_1=\tfrac{\alpha a+1}{\alpha^2\beta}$ and $b_3=\tfrac{\alpha
a+\beta^2+\alpha}{\alpha\beta^2}$ come from
 $b_1=\frac{1}{\alpha(\alpha a+\beta)}=\frac{a+1}{\beta(\alpha a+1)}=b_2$
   (i.e., $\alpha^2 a^2+\alpha^2 a+\beta^2=0$ from the equation \eqref{eqnb1b2}). Via
explicit computations, one has
\begin{align*}
\res(r_{12}(X+\tfrac{1}{\alpha}),
 \hh^{(a,b_3)}(X+\tfrac{1}{\alpha}))=\tfrac{1}{\alpha^6\beta^8}(\alpha a+\alpha^3+\beta^2)^2(\alpha\beta a+\alpha\beta+1)
\end{align*}
 The above resultant is zero if and only if $a\in
 \{\frac{\alpha^3+\beta^2}{\alpha}, \frac{\alpha\beta+1}{\alpha\beta}
 \}$, however in this case, one has $\beta^3+\beta^2+1=0$ which is a contradiction.  That is, when
 $a=\frac{\alpha^3+\beta^2}{\alpha}$, one may proceed as
 \begin{align*}
0 &=\alpha^2 a^2+\alpha^2
a+\beta^2=(\alpha^3+\beta^2)^2+\alpha(\alpha^3+\beta^2)+\beta^2 \\
 &=\alpha^6+\beta^4+\alpha^4+\alpha\beta^2+\beta^2=\alpha^6+1+\beta^3
 \\
 &=\beta^6+\beta^4+\beta^3+\beta^2=\beta^2(\beta+1)(\beta^3+\beta^2+1)
 \end{align*}
\noindent  The case $a=\frac{\alpha\beta+1}{\alpha\beta}$ can be dealt in a similar manner.

\medskip

\noindent $\{\mathcal{B(1)},\mathcal{B(2)}\}+\{\mathcal{B(3)},\mathcal{C}\}$\,\,\,  Case (ii) :

\noindent Again using the linear expression of $b_1=b_2$ and $b_3$, one finds
 $r_{12}(b_3)=0$ implies $\beta^3+\beta^2+1=0$ when $b_3+\tfrac{1}{\alpha\beta}=0$, which is
  not possible. When $b_3+b_1+\tfrac{1}{\alpha\beta}=0$, one finds $a=\alpha^4=b_3, b_1=b_2=\beta^3+\beta+1$
   and $\beta^5=1$. Similarly one has $\du_G(a,b_3)=2<4$
  (i.e., there is no $u_3,v_3\in \mathcal{C}$ satisfying $a=G(u_3)+G(v_3),\,\, b_3=u_3+v_3$)
   because $b_3=\tfrac{a}{\beta(\alpha a+\beta)}$ and one has $\# \mathcal{C}_G^{(a,b)}=0$ by Corollary \ref{abalphamcoro}.
 Therefore $\bu_G(a,c)\geq 8$ is not possible in this case.

\medskip

\noindent $\{\mathcal{B(1)},\mathcal{B(2)}\}+\{\mathcal{B(3)},\mathcal{C}\}$\,\,\,  Case (iii) :

\noindent One has
  $0=h_3^{(a,b_3)}(b)=\hh^{(a,b_3)}(b)=\hh^{(a,b_3)}(b+\tfrac1\alpha+\tfrac1\alpha)=H_a(b+\tfrac1\alpha, b_3)$ with
$b=b_1=b_2$
 which implies that the three simultaneous equations \eqref{hab1b3},\eqref{f2quartic} and \eqref{eqnb1b2}
 must be satisfied, and it is easy to show the common solution does not exist.

\bigskip
\noindent {\bf Proof of Theorem \ref{butheorem} :} \hfill

\medskip

\noindent
 From our previous analysis of all possible types of combinations of differential uniformities
 producing $\bu_G(a,c)\geq 8$, we already showed that one has $\bu_G\leq 6$ if and only if
none of the above mentioned five polynomials $h_{12},h_{13},g_1,g_2$ and $\phi$ have a root in $\Fbn$. Therefore it
remains to show that
 there do exist $a,c \in \Fbn$ such that $\bu_G(a,c)\geq 6$.
 We will start from the following $\bu$ equation,
 \begin{align}
 G(0)+G(b)=a=G(u)+G(v), \quad b=u+v\quad \text{with}\,\,\,b=\tfrac{\alpha a+1}{\alpha^2 a}
 \end{align}
The equation $a=G(0)+G(b)\,\, (b\notin \pole)$ is always satisfied with the parametrization $b=\tfrac{\alpha
a+1}{\alpha^2 a}=\tfrac1\alpha+\tfrac1{\alpha^2 a}$, which gives one to one correspondence between $a\in \Fbn \setminus
\{0, \tfrac1\alpha, \tfrac\beta\alpha\}$ and $b\in \Fbn \setminus \pole$. To satisfy the second equation $a=G(u)+G(v)$
for
 some $u,v  \notin \pole$ with $u+v=b$, one needs to have $\tr(\tfrac{1}{ab\alpha^2})=0$. Since the
equation $a=G(0)+G(b)$ implies $\tfrac{1}{\alpha b}+1=\tfrac{1}{ab\alpha^2}$ as is already mentioned in the
 Lemma \ref{tracelemma}-$(b)$, one gets $\tr(\tfrac{1}{ab\alpha^2})=0$ if and only if $\tr(\tfrac{1}{\alpha b}+1)=0$.
Via one to one correspondence, $\tfrac{1}{\alpha b}+1$ takes all values of $\Fbn$ except $3$ values, and half of $\Fbn$
take zero trace value. Therefore one has $\tr(\tfrac{1}{\alpha b}+1)=0$ for approximately $2^{n-1}$ values of $a$ (or
$b$). That is, one has $\du_G(a,b)\geq 4$ with $b=\tfrac{\alpha a+1}{\alpha^2 a}$ for approximately $2^{n-1}$ values of
$a\in \Fbn$. Now letting $c=u$ or $c=v$, the $\bu$ equation $G(x)+G(y)=a=G(x+c)+G(y+c)$ has $2$ solutions
$\{x,y\}=\{0,b\}, \{u,v\}$ (i.e., $4$ solutions of ordered pair $(x,y)$), which means $\bu_G(a,c)\geq 4$. We claim that
there exists
 $a\in \Fbn$, among those $2^{n-1}$ values producing $u,v$, such that $\du_G(a,c)\geq 2$ for some $c=u$ or $v$.
More precisely we will show that there is $a$ satisfying $\tr(\tfrac1{ac\alpha^2})=0$. From the equation
\eqref{basicpoly}, one has
 \begin{align*}
b &= a(u\alpha +1)(v\alpha +1)=a(uv\alpha^2  +b \alpha +1) =au(b+u)\alpha^2 + ab\alpha+a \\
              &=au^2\alpha^2+uab\alpha^2+ab\alpha +a=au^2\alpha^2+u(\alpha a+1)+\dfrac{\alpha a+1}{\alpha} +a \\
              &=au^2\alpha^2+u(\alpha a+1)+\dfrac{1}{\alpha}
 \end{align*}
where $\tfrac{1}{\alpha b}+1=\tfrac{1}{ab\alpha^2}$ (i.e., $ab\alpha^2=\alpha a+1$) is used. Therefore, using the same
equation again, it follows that
\begin{align*}
0&=au^2\alpha^2+u(\alpha a+1)+\dfrac{1}{\alpha}+b=au^2\alpha^2+u(\alpha a+1)+\dfrac{1}{a\alpha^2}\\
 & =u\left(au\alpha^2+(\alpha a+1)+\dfrac{1}{au\alpha^2}\right)
\end{align*}
which implies
\begin{align*}
0= au\alpha^2+(\alpha a+1)+\dfrac{1}{au\alpha^2}
\end{align*}
Using the symmetry between $u$ and $v$, we get the same equation for $v$ and $a$ so that
 \begin{align}
 au\alpha^2+\dfrac{1}{au\alpha^2} = \alpha a+1 = av\alpha^2+\dfrac{1}{av\alpha^2} \label{boomerang6eq}
\end{align}
(Note that the above equation can also be derived using the
 fact $au\alpha^2\cdot av\alpha^2=1$ (i.e., $uv=\tfrac{1}{a^2\alpha^4}$) which can be
derived directly from the equation (\ref{eqnhh}) by comparing constant terms.) Now, for the moment, suppose that
 there is $a\in \Fbn$ such that
\begin{align}
\tr(\alpha a+1)=1 \quad \text{and} \quad \tr(\tfrac{1}{\alpha a+1})=0 \label{tracekoolsterman}
\end{align}
Then one has $\tr(ab\alpha^2)=1$  and
 $\tr(\tfrac{1}{ab\alpha^2})=0$ because $ab\alpha^2=\alpha a+1$. From the equation (\ref{boomerang6eq}),
one gets
$$1=\tr(ab\alpha^2)=\tr(au\alpha^2)+\tr(av\alpha^2)=\tr(\tfrac{1}{au\alpha^2})+\tr(\tfrac{1}{av\alpha^2})$$
Therefore the existence of $a\in \Fbn$ satisfying $\tr(\tfrac{1}{\alpha a+1})=0$ and $\tr(\alpha a +1)=1$ implies
$$
\tr(\tfrac{1}{ab\alpha^2})=0 \quad \text{and} \quad \tr(\tfrac{1}{au\alpha^2})+\tr(\tfrac{1}{av\alpha^2})=1
$$
Thus, choosing $c=u$ or $v$ satisfying $ \tr(\tfrac{1}{ac\alpha^2})=0$, one has $\bu_G(a,c)\geq 6$. Finally, it remains
to show that there exists such $a$ satisfying (\ref{tracekoolsterman}). The result on Koolsterman sum (See \cite{LW90})
implies
 $-2^{\frac{n}{2}+1}<\displaystyle\sum_{z\neq 0 \in \Fbn} (-1)^{\tr(z+\frac{1}{z})} <2^{\frac{n}{2}+1}$. Letting $k$ be the number
 of $z\in \Fbn \setminus \{0\}$ satisfying $\tr(z+\frac{1}{z})=1$, the result on Koolsterman sum implies $k$ is bounded as
 $2^{n-1}-2^{\frac{n}{2}}-1 <k<2^{n-1}+2^{\frac{n}{2}}$. Therefore there are plenty of $z\in \Fbn$ satisfying
$\tr(z+\frac{1}{z})=1$ and one can choose $\alpha a+1=z$ or $\tfrac{1}{z}$ satisfying $\tr(\alpha a+1)=1$.
 \qed

\section{Application to Involutions with Boomerang Uniformity 6 and Implementations}

An involution  $F$ is a permutation whose compositional inverse is itself, i.e., $F\circ F(x)=x$ for all $x\in \Fbn$.
Due to the natural (sequence like) structure of the
 permutations with Carlitz form,  one can easily construct involutions in Carlitz form. For example,
we easily see that $F(x)=[a_{m+1},a_m,\ldots, a_2, a_1+a_0x]$ is an involution if $a_0=1$ and $a_{m+1-i}=a_{1+i}$ for
$0\le i \le m$ (i.e., palindromic structure). Now let us consider involutions of Carlitz rank three of the form
$F(x)=[a_4,a_3,a_2,a_1+a_0x]=[\gamma, \beta, \beta, \gamma+x]$, or equivalently,
\begin{equation}\label{CR3_invol}
F(x)=(((x+\gamma)^{2^n-2} +\beta)^{2^n-2} +\beta)^{2^n-2}+\gamma
\end{equation}
where $\beta \ne 0$. From Proposition \ref{EAequivalent lemma}-$(2)$, $F(x)$ is affine equivalent to $[0,\beta^2,1,x]$.
Also since
$$[0,\beta^2,1,x] = x^{2}\circ [0,\beta,1, x]\circ x^{2^{n-1}},$$
we see that $[0,\beta^2,1,x]$ is affine equivalent to  $[0,\beta,1,x]$.
 We already mentioned in Section \ref{sectioncarlitz3} that $[0,\beta,1,x]$ and $[0,1,\beta,x]$
  are affine equivalent via the relation $\beta[0,\beta,1,\beta x]=[0,1,\beta,x]$. Therefore, $F(x)$ with any $\gamma \in \Fbn$
in \eqref{CR3_invol} is affine equivalent to $[0,1,\beta,x]$ whose boomerang and differential uniformity are
investigated in Section \ref{differentialsec} and \ref{boomerangsec}. We summarize this result in the following
theorem.

\begin{theorem}
Let $\beta\ne 0 \in \Fbn$. Then the involution $F(x)=[\gamma,\beta,\beta,\gamma+x]$ with any $\gamma\in \Fbn$ is affine
equivalent to $[0,1,\beta, x]$, hence the involution $F$ has the same boomerang uniformity with $[0,1,\beta, x]$.
\end{theorem}
\begin{corollary}
Up to affine equivalence, there is one to one correspondence between permutations of Carlitz rank $3$ and involutory
permutations of Carlitz rank $3$. In other words, any permutation of Carlitz rank $3$ can be expressed as an involution
via affine equivalence.
\end{corollary}
\begin{remark}
One may choose $\gamma=0$ such that the involution $F(x)=[0,\beta,\beta,x]$ can be written as
$$  F(x)=
   \begin{cases}
\tfrac{\beta}{\alpha^2} &\text{ if } x= 0 \\
\tfrac{1}{\beta} &\text{ if }  x= \tfrac{1}{\beta} \\
0 &\text{ if }  x= \tfrac{\beta}{\alpha^2} \\
\frac{\beta x+1}{\alpha^2 x+\beta } &\text{ if } x\neq 0,\tfrac1\beta,\tfrac\beta{\alpha^2}
   \end{cases}
  $$
  where $\alpha=\beta+1$. Therefore the above involutions with $\beta$ satisfying Theorem \ref{butheorem}
  provide a COMPLETE LIST of involutions $($up to affine equivalence$)$ with boomerang uniformity six. Please note that $[0,1,\beta, x]$ that
  we considered in Section \ref{differentialsec}, \ref{boomerangsec} is a compositional inverse of $[0,\beta,1,x]$,
  while $[0,\beta,\beta,x]$ is an involution which is affine equivalent to both $[0,\beta,1,x]$ and $[0,1,\beta,x]$.
\end{remark}

\noindent It should be mentioned that, in  \cite{FF19,LWY13e}, involutions of the form $\pi(x)^{2^n-2}$ with 3-cycle
$\pi(x)=(0\,\beta\,\beta^{-1})$ are considered and some  conditions for which $\pi(x)^{2^n-2}$ becomes differentially
4-uniform are stated in terms of trace conditions of $\beta$. Since $\crk(\pi(x)^{2^n-2})=3$,
 they found some classes of differentially 4-uniform involutions of Carlitz rank 3. However, using Theorem
\ref{dutheorem} and the above theorem, we
 have a complete answer for the conditions for which an involution $F$ with $\crk(F)=3$ becomes differentially 4-uniform. Also, using
 Theorem \ref{butheorem} and the above theorem, we have a complete answer for the conditions
 for which an involution $F$ with $\crk(F)=3$ becomes boomerang 6-uniform.

 \bigskip

 \begin{table}[!b]
\begin{center}
\begin{tabular}{|c|c|c|c|}
\hline
\hline  & the number of  $\beta\in\Fbn\setminus\F_4$ with  & the number of  $\beta\in\Fbn\setminus\F_4$ with & Timing \\
 $n$ & $\texttt{BU}_G=6$ & $\texttt{DU}_G=4$ and $\texttt{BU}_G=6$ & (seconds)\\
\hline 4 & 4 & 0 & 0.013\\
\hline 6 & 6 & 6 & 0.006\\
\hline 8 & 16 & 8 & 0.035 \\
\hline 10 & 80 & 50 & 0.149 \\
\hline 12 & 264 & 180 & 0.676 \\
\hline 14 & 1148 & 784 & 3.456 \\
\hline 16 & 3696 & 2080 & 21.1\\
\hline 18 & 16020 & 9828 & 244.5 \\
\hline 20 & 63760 & 38120 & 3346.7 \\
\hline 22 & 252538 & 152020 & 57717.7\\
\hline \hline
\end{tabular}
\caption{Number of $\beta\in\Fbn \setminus \F_4$ having optimal $\texttt{BU}_G$ and $\texttt{DU}_G$}
\label{table_bu6du4}
\end{center}
\end{table}
By using Theorem \ref{dutheorem} and Theorem \ref{butheorem}, we did implementations using a software SageMath for the
permutation $G(x)=[0,1,\beta,x]$ on $\Fbn$ for even $4\le n\le 22$.
 The
result is shown  in Table \ref{table_bu6du4}, where the number of $\beta\in\Fbn\setminus\F_4$ such that
$\texttt{BU}_G=6$ and the number of $\beta\in\Fbn\setminus\F_4$ such that $\texttt{BU}_G=6, \texttt{DU}_G=4$ are listed
in the second and third column respectively. In our experiments, it is enough to try  $\frac{2^n}{n}$ number of $\beta
\in \Fbn$ for each $n$, because $[0,1,\beta,x]$ is affine equivalent to $[0,1,\beta^{2^i},x]$ ($1\le i \le n-1$) from
the equivalence $[0,\beta^{2^i},1,x] = x^{2^i}\circ [0,\beta,1, x]\circ x^{2^{n-i}}$. Please note that we have not
checked  the equivalences such as EA or CCZ equivalences among the permutations that we found. Our experiments were
performed via SageMath on Intel Core i7-4770 3.40GHz with 8GB memory processor.

\section{Conclusion}

In this paper, we presented a new methodology for computing boomerang and differential uniformity of permutations of
low Carlitz rank. We applied our method to permutations of Carlitz rank $3$, and obtained a complete list of
permutations having the boomerang uniformity six, which is the least possible case. We also gave complete
classifications of all possible differential uniformities of permutations of Carlitz rank $3$. As a  consequence, we
discovered new classes of permutations having the boomerang uniformity six and differential uniformity four, which were
previously unknown.

Since all permutations $F$ on $\Fbn$ of Caritz rank $3$ have  high algebraic degree $n-1$ and also  high nonlinearity
$\geq 2^{n-1}-2^{\frac{n}2}-2$, a permutation $F$ with low boomerang uniformity is a good example of cryptographic
S-box. In addition, due to the nice
 structure of Carlitz form, one can always choose an involutory permutation $F$ (via affine equivalence)
  of boomerang uniformity six, which is also a new result.

For future research, we hope our approach with inductive technique can be extended to permutations with Carlitz rank
$\geq 4$ and find new families  of permutations with good cryptographic properties.

\end{document}